\documentclass[12pt,dvipsnames,accepted=2016-08-18,a4paper,onecolumn]{quantumarticle}
\pdfoutput=1
\usepackage[english]{babel}
\usepackage[utf8]{inputenc}
\usepackage{amssymb,amsthm,amsfonts,amsmath}
\usepackage{mathrsfs}
\usepackage{verbatim}
\usepackage{mathtools}
\usepackage[shortlabels]{enumitem}
\usepackage[normalem]{ulem}

\usepackage{tikz}
\usepackage{pgfmath}
\usetikzlibrary{matrix,arrows}
\usetikzlibrary {graphs.standard}
\usetikzlibrary{decorations.pathmorphing,positioning}
\tikzset{mydeco/.style={decoration={random steps,segment length=.08em,amplitude=.05em},decorate,line cap=round}}

\usepackage[normalem]{ulem}
\usepackage{fontawesome5}

\allowdisplaybreaks

\usepackage{autobreak}

\usepackage{url}
\usepackage{amsopn}
\usepackage{graphicx}

\usepackage{enumitem,titletoc}
\definecolor{darkblue}{rgb}{0.0, 0.0, 0.55}
\usepackage[colorlinks,breaklinks,linkcolor=BrickRed,citecolor=OliveGreen,urlcolor=darkblue,hypertexnames=true]{hyperref}
\usepackage[a4paper,margin=2.5cm]{geometry}
\usepackage{ytableau}
\def\starn{{\bigstar_{n}}}
\usepackage[baseline]{euflag}

\usepackage{setspace}
\renewcommand{\qedsymbol}{\rule[.12ex]{1.2ex}{1.2ex}}
\makeatletter
\newtheorem*{rep@thm}{\rep@title}
\newcommand{\newreptheorem}[2]{%
	\newenvironment{rep#1}[1]{%
		\def\rep@title{#2 \ref{##1}}%
		\begin{rep@thm}}%
		{\end{rep@thm}}}
\makeatother
\newreptheorem{thm}{Theorem}
\newreptheorem{prop}{Proposition}
\newreptheorem{cor}{Corollary}

\def\la{\lambda}

\DeclareMathOperator{\spec}{spec}
\DeclareMathOperator{\EE}{E}
\DeclareMathOperator{\VV}{V}
\DeclareMathOperator{\het}{ht}

\DeclareMathOperator{\id}{id}

\DeclareMathOperator{\Tr}{tr}

\newcommand{\floor}[1]{\left\lfloor #1 \right\rfloor}

\def\la{\lambda}

\def\bal{\tikz[baseline=(char.base)]{
  \node[scale=.6,opacity=0.6] (char) {\faBalanceScale}}
}
\def\lab{\lambda_{\!\bal}}

\numberwithin{equation}{section}

\newcommand{\RR}{\mathbb R}

\newcommand{\NN}{\mathbb N}
\newcommand{\N}{\mathbb N}
\newcommand{\ZZ}{\mathbb Z}
\newcommand{\CC}{\mathbb C}

\newcommand{\cA}{\mathcal A}
\newcommand{\cB}{\mathcal B}
\newcommand{\cF}{\mathcal F}

\newtheorem{theorem}{Theorem}[section]
\newtheorem{corollary}[theorem]{Corollary}
\newtheorem{lemma}[theorem]{Lemma}
\newtheorem{proposition}[theorem]{Proposition}

\newtheorem{conjecture}[theorem]{Conjecture}

\theoremstyle{definition}
\newtheorem{definition}[theorem]{Definition}
\newtheorem{remark}[theorem]{Remark}
\newtheorem{example}[theorem]{Example}

\newcommand\blfootnote[1]{%
  \begingroup
  \renewcommand\thefootnote{}\footnote{#1}%
  \addtocounter{footnote}{-1}%
  \endgroup
}

\begin{document}

\title{Quantum Max $d$-Cut via qudit swap operators}

\author{Igor Klep}
\address{Faculty of Mathematics and Physics, Department of Mathematics,  University of Ljubljana \& 
FAMNIT, University of Primorska, Koper \&
Institute of Mathematics, Physics and Mechanics, Ljubljana, Slovenia}
\email{igor.klep@fmf.uni-lj.si}
\thanks{IK was supported by the Slovenian Research Agency program P1-0222 and grants J1-50002, N1-0217, J1-3004,
J1-50001, J1-60011, J1-60025.
Partially supported by the Fondation de l'École polytechnique as part of the Gaspard Monge Visiting Professor Program.
IK thanks École Polytechnique and Inria Saclay for hospitality.
IK's work was performed within the project
COMPUTE, funded within the QuantERA II Programme that has
received funding from the EU’s H2020 research and innovation
programme under the GA No 101017733. {\normalsize\euflag}}

\author{Tea \v{S}trekelj}
\address{Department of Mathematics, FAMNIT, University of Primorska, Koper, Slovenia}
\email{tea.strekelj@famnit.upr.si}
\thanks{TŠ was supported by the Slovenian Research Agency grant J1-60011.}
\author{Jurij Vol\v{c}i\v{c}}
\address{Department of Mathematics, University of Auckland, Auckland, New Zealand}
\email{jurij.volcic@auckland.ac.nz}
\thanks{JV was supported by the NSF grant DMS-2348720, the Marsden Fund MFP-UOA2528 from Government funding, administered by the Royal Society Te Ap\={a}rangi, and the New Staff grant 3734441 by the Faculty of Science, the University of Auckland.}



\begin{abstract}
Quantum Max Cut (QMC) problem for systems of qubits is an example of a 2-local Hamiltonian problem, and a prominent paradigm in computational complexity theory. This paper investigates the algebraic structure of a higher-dimensional analog of the QMC problem for systems of qudits. The Quantum Max $d$-Cut ($d$-QMC) problem asks for the largest eigenvalue of a Hamiltonian on a graph with $n$ vertices whose edges correspond to swap operators acting on $(\CC^d)^{\otimes n}$.
The algebra generated by the swap operators is identified as a quotient of a free algebra modulo symmetric group relations and a single additional relation of degree $d$. This presentation leads to a tailored hierarchy of semidefinite programs,
leveraging noncommutative polynomial optimization (NPO) methods,
that converges to the solution of the $d$-QMC problem. 
For a large class of complete bipartite graphs, exact solutions for the $d$-QMC problem are derived using the representation theory of symmetric groups and Littlewood-Richardson coefficients. 
Lastly, the paper addresses a refined $d$-QMC problem focused on finding the largest eigenvalue within each isotypic component (irreducible block) of the graph Hamiltonian. 
It is shown that the spectrum of the star graph Hamiltonian distinguishes between isotypic components of the $3$-QMC problem. 
For general $d$, low-degree relations for separating isotypic components are presented, enabling adaptation of the global NPO hierarchy to efficiently compute the largest eigenvalue in each isotypic component.
\end{abstract}

\maketitle

\setstretch{1.1}

\tableofcontents

\setstretch{1.08}


\section{Introduction}

The local Hamiltonian problem
is a renowned problem in quantum computational complexity theory. It involves determining the largest (or smallest) eigenvalue of a given self-adjoint matrix $H.$ The input matrix $H$ acts on a space of $n$ qubits and is hence of size $2^n \times 2^n.$ It is expressed as a sum of \textit{local} terms, i.e., for a chosen $k\leq n,$
$$
H = \sum_{\substack{S \subseteq \{1,\ldots,n\} \\[.5mm] |S| = k}} H_S.
$$
Here each $H_S$ acts nontrivially only on a subset (of $S$) of at most $k$ qubits. Such an $H$ is called a $k$-local Hamiltonian.

The general $k$-local Hamiltonian problem is hard to solve; in fact, it is hard for the Quantum Merlin Arthur (QMA) complexity class \cite{KSV02,KKR06}, which is a quantum analog of the NP-hard class. Hence, it is easier to approach by considering its specific instances, either by computing exact arithmetic solutions \cite{LM62} or designing efficient (polynomial-time) high-precision algorithms to approximate the largest eigenvalue \cite{LVV15}. Additional work was done on approximating the maximum eigenvalue up to a constant factor \cite{GK12,BH13,BGKT19,HM17}, and exploring hardness of computing ground space properties \cite{GH}.

We investigate generalizations of the Quantum Max Cut (QMC) problem, which is a special instance of the $2$-local Hamiltonian problem, and was named by Gharibian and Parekh \cite{GP19} as a quantum analog of the classical Max Cut problem for the Ising model (Section \ref{sec:clasMC}). The QMC problem naturally arises in physics as it 
seeks the ground state energy of
the anti-ferromagnetic Heisenberg 
model for a system of interacting particles. 
The latter is used to describe magnetic properties of insular crystals, under the assumption that only the interactions of neighbor electrons in a lattice are significant (2-locality) \cite{Aue,BDZ}.
The QMC problem has recently become popular within the field of computational complexity theory. It is a simple prototype of a QMA-complete problem \cite{PM17} and can hence be used for designing approximation algorithms to solve other QMA-hard problems \cite{AMG20,PT21,PT22,Lee22,Kin23}. Arithmetic solutions to the QMC problem are known for certain families of graphs, such as complete bipartite graphs \cite{LM62} and one-dimensional chains \cite{LM16}.
More recently, second order cone relaxations of the QMC problem capable of providing approximations for large graphs were introduced \cite{HTPG}, and approximation algorithms tailored to triangle-free and  bipartite graphs were designed and analyzed \cite{gribling}.

The  main objects used to define the QMC Hamiltonian are the \textbf{Pauli matrices} 
\begin{equation}\label{eq:pauli}\tag{Pauli}
\sigma_X = 
\begin{bmatrix}
0 & 1 \\
1 & 0
\end{bmatrix},
\quad
\sigma_Y = 
\begin{bmatrix}
0 & -\mathfrak{i} \\
\mathfrak{i} & 0
\end{bmatrix},
\quad
\sigma_Z = 
\begin{bmatrix}
1 & 0 \\
0 & -1
\end{bmatrix}.
\end{equation}
Together with the identity $\sigma_I:=I,$ they form a basis for $M_2(\CC).$ For fixed $n$ let
\[\sigma_W^k = \underbrace{I_{2} \otimes \cdots\otimes I_2}_{k-1} \otimes\, \sigma_W \otimes \underbrace{I_{2} \otimes \cdots\otimes I_2}_{n-k} \in M_{2}(\CC)^{\otimes n}=M_{2^n}(\CC)\]
with $W\in\{X,Y,Z\}$ and $k\in\NN.$ Now 
\begin{equation}\label{eq:basis2}
    \{\sigma^1_{W_1} \sigma^2_{W_2}\cdots \sigma^n_{W_n} \mid W_j \in \{I,X,Y,Z\} \}
\end{equation}
is a basis for $M_{2^n}(\CC).$

A QMC Hamiltonian pertains to a given graph $G$ on say $n$ vertices. We denote by V$(G)$ the vertex set of $G$ and by E$(G)$ the edge set of $G.$ 

\begin{definition}
    Let $G$ be a graph on $n$ vertices and edge weights $\{w_{ij} \ | \ (i,j) \in \textnormal{E}(G)\}.$  The Quantum Max Cut (QMC) Hamiltonian is defined as
    \begin{equation}\label{eq:qmc}\tag{$H_G$}
         H_G = \sum_{(i,j) \in \EE(G)} w_{ij}\big(I - \sigma_X^i \sigma_X^j - \sigma_Y^i \sigma_Y^j - \sigma_Z^i \sigma_Z^j\big)\in M_{2^n}(\CC)_{\rm sa}.
    \end{equation}
    The Quantum Max Cut (QMC) problem is about finding the largest eigenvalue of the QMC Hamiltonian $H_G$; that is, the ground state energy of $-H_G$. 
\end{definition}

\subsection{Connection to the classical Max Cut}\label{sec:clasMC}
The QMC problem is named after the classical Max Cut (MC) problem \cite{BPT13}
of partitioning the vertices of a given graph into two sets such that the number or weight of the edges between the two sets is maximized. Equivalently, if the given graph $G$ has edge set $\EE(G)$ and edge weights $w_{ij}\geq0,$ maximize
$$
\sum_{(i,j) \in \EE(G)} w_{ij}\frac{1-x_ix_j}{2}
$$
over all possible evaluations at $x_i \in \{\pm 1\}.$ 
Note that the MC problem is equivalent to the ``diagonal'' modification of the QMC problem, where the $\sigma_X^i \sigma_X^j$ and $\sigma_Y^i \sigma_Y^j$ terms in \eqref{eq:qmc} are dropped.
Alternatively, while the QMC problem seeks the ground state energy of the Heisenberg XXX model, classical MC problem seeks the ground state energy of the Ising model (without an external field).

Solving the MC problem in general is NP-hard, thus several approximation algorithms were developed. 
The most famous approximation algorithm is by Goemans and Williamson \cite{GW95}, and is based on semidefinite programming (SDP) \cite{BPT13}. 
It can be understood as the first level of
Lasserre's Moment-SOS (Sum-of-Squares) hierarchy of SDP relaxations \cite{Las} (see also \cite{Lau,HKL,Nie23}) that give a converging sequence of upper bounds to the exact solution of the MC problem. 
Raghavendra \cite{Rag08,Rag09} showed, assuming the Unique Games Conjecture of Khot \cite{Kho02}, that no polynomial-time algorithm for the MC problem is better than the Goemans-Williamson algorithm (unless P$=$NP).

\subsection{Quantum Max Cut}
To tackle the QMC problem, the algebraic structure of the QMC Hamiltonian 
is investigated in \cite{BCEHK24,TRZ}. This approach starts by rephrasing $H_G$ in terms of the {swap matrices} Swap$_{ij}.$\footnote{Physics literature often calls these SWAP or exchange operators \cite{NC}.
} 

\begin{definition}
For fixed $n$ and  $1\leq i < j \leq n,$   the swap matrix \textnormal{Swap}$_{ij} \in M_{2^n}(\CC)$ is defined by sending any rank one tensor
    \[v_1\otimes \cdots \otimes v_i \otimes \cdots \otimes v_j \otimes \cdots \otimes v_n\in(\CC^2)^{\otimes n}\phantom{,}\] 
        to the rank one tensor
        \[
        v_1\otimes \cdots \otimes v_j \otimes \cdots \otimes v_i \otimes \cdots \otimes v_n\in(\CC^2)^{\otimes n},\] 
        where $v_k\in\CC^2$.
\end{definition}

One can directly compute that any swap matrix Swap$_{ij}$ is expressed in terms of the Pauli matrices as 
\begin{equation}\label{eq:swapviap}
        \textnormal{Swap}_{ij} = \frac{1}{2} (I + \sigma_X^i \sigma_X^j + \sigma_Y^i \sigma_Y^j + \sigma_Z^i \sigma_Z^j).
\end{equation}
Using \eqref{eq:swapviap}, the QMC Hamiltonian  \eqref{eq:qmc} can be expressed in terms of the swap matrices rather than the Pauli matrices \eqref{eq:basis2}.

\begin{proposition}
    The QMC Hamiltonian from \eqref{eq:qmc} is given in terms of the swap matrices \textnormal{Swap}$_{ij}$ as
    \begin{equation}\label{eq:HamSw}
        H_G = \sum_{(i,j) \in \EE(G)} 2 w_{ij} (I - \textnormal{Swap}_{ij}).
    \end{equation}
\end{proposition}

\subsection{Swap matrices on general qudit spaces}
In this article we consider the QMC problem
on qudits instead of qubits. As qudits store more information than qubits, systems of interacting qudits are a natural framework for quantum computing with less resources \cite{WHSK}.
Here, the swap matrices Swap$_{ij}^{(d)}$  act on $(\CC^d)^{\otimes n}$ for some $d\geq 2.$ In analogy with the $d=2$ case,  they act as transpositions on $n$-qudit states.

\begin{definition} For fixed $n$ and $1\leq i < j \leq n,$ the (\textbf{qudit}) \textbf{swap matrix} Swap$_{ij}^{(d)}$  is defined by its action on rank one tensors as
	$$
	\text{Swap}_{ij}^{(d)}(v_1 \otimes \cdots \otimes v_i \otimes \cdots \otimes v_j \otimes \cdots \otimes v_n) = v_1 \otimes \cdots \otimes v_j \otimes \cdots \otimes v_i \otimes \cdots \otimes v_n
	$$
 for any $v_1, \ldots, v_n \in \CC^d.$
\end{definition}

	The action of swap matrices on qudits yields a representation $\rho_n^{(d)}$ of $S_n$ on $(\mathbb{C}^d)^{\otimes n}$ defined by
	$$
	\rho_n^{(d)}(\pi) (v_1 \otimes \cdots \otimes v_n) = v_{\pi^{-1}(1)} \otimes \cdots \otimes v_{\pi^{-1}(n)}.
	$$
	We denote the image $\rho_n^{(d)}\big(\CC[S_n] \big),$ which is a subalgebra of $M_{d^n}(\CC),$ by $M^{\text{Sw}_d}_n(\mathbb{C}).$ It is called the \textbf{$d$-swap algebra}. Guided by the expression \eqref{eq:HamSw} of the QMC Hamiltonian in terms of the swap matrices, one can define the Quantum Max $d$-Cut Hamiltonian via the qudit swap matrices Swap$^{(d)}_{ij}.$

\begin{definition}
     Let $G$ be a graph on $n$ vertices and edge weights $\{w_{ij} \ | \ (i,j) \in \textnormal{E}(G)\}.$  The \textbf{Quantum Max $d$-Cut} (\textbf{$d$-QMC}) \textbf{Hamiltonian} is defined as
    \begin{equation}\label{eq:qmcd}\tag{$H_G^d$}
        H_G^{d} = \sum_{(i,j) \in \EE(G)} 2 w_{ij} \left(I - \textnormal{Swap}^{(d)}_{ij}\right).
    \end{equation}
\end{definition}

The \textbf{$d$-QMC problem} again asks for the largest eigenvalue of the $d$-QMC Hamiltonian $H_G^d$ in \eqref{eq:qmcd}. 
The problem is motivated by determining ground state energies of $\operatorname{SU}(d)$-Heisenberg models on lattices \cite{KT,BAMC,PM21}.
While the QMC problem is the quantum analog of the classical MC problem, the $d$-QMC problem is the quantum analog of the $d$-MC problem pertaining to maximal $d$-colorable subgraphs, and the anti-ferromagnetic $d$-state Potts model \cite{FJ}. 
The $d$-QMC problem was first considered in this context by \cite{CJKKW23+} in 2023.
There the authors define the $d$-QMC Hamiltonian  with the use of the \textit{Gell-Mann matrices}, which are a generalization of the Pauli matrices to any size $d \times d.$ We give more insight into this approach 
and show that is equivalent to ours in Section \ref{sec:gm} below.

In order to develop an algebraic toolbox for solving the $d$-QMC problem, it is essential to determine the precise relations that define the $d$-swap algebra.
 Since the transpositions $(i,j)$ generate $S_n,$ the swap matrices generate $M^{\text{Sw}_d}_n(\mathbb{C}).$ Hence, similar to the transpositions, for distinct indices $i,j,k,l,$ the swap matrices satisfy the relations
	\begin{equation}\label{eq:snrel}
	\begin{split}
		(\text{Swap}^{(d)}_{ij})^2 &= \textcolor{black}{I}, \\
		\text{Swap}^{(d)}_{ij}\,\text{Swap}^{(d)}_{kl} &= \text{Swap}^{(d)}_{kl}\,\text{Swap}^{(d)}_{ij},\\
		\text{Swap}^{(d)}_{ij}\,\text{Swap}^{(d)}_{jk} &= \text{Swap}^{(d)}_{ik}\,\text{Swap}^{(d)}_{ij} = \text{Swap}^{(d)}_{jk}\,\text{Swap}^{(d)}_{ik}.\\
	\end{split}
\end{equation}

For $d=2$, 
it is known (see \cite[Theorem 3.6]{BCEHK24} and \cite[Theorem 3.8]{TRZ}, and \cite[Theorem 11.6.1]{Probook} for general $d$) 
that
the swap matrices additionally satisfy the \textbf{degree-reducing relation}
\begin{equation}\label{eq:degred}
    \text{Swap}^{(2)}_{ij} \,\text{Swap}^{(2)}_{jk} + \text{Swap}^{(2)}_{jk} \,\text{Swap}^{(2)}_{ij} = \text{Swap}^{(2)}_{ij}+\text{Swap}^{(2)}_{jk}+\text{Swap}^{(2)}_{ik} -\textcolor{black}{I},
\end{equation}
and that the symmetric group relations \eqref{eq:snrel} together with the degree-reducing relation \eqref{eq:degred} precisely define $M^{\text{Sw}_2}_n(\mathbb{C}).$ In Section \ref{sec:qudit-rel} we show that the general swap matrices Swap$_{ij}^{(d)}$ are also characterized by a (slightly more complicated) degree-reducing relation;
see 
Proposition \ref{prop:degred}.

\subsection{Main results}

This paper applies the representation theory of the symmetric group $S_n$ to explore and take advantage of the algebraic structure and symmetries inherent to the $d$-QMC problem. 
Throughout the text, we refer to irreducible representations as \textbf{irreps}.
Our contributions are as follows.

\subsubsection{Defining relations of the $d$-swap algebra}

In Section \ref{s:quotalg}, we identify the $d$-swap algebra $M^{\text{Sw}_d}_n(\mathbb{C})$ as a quotient of the free algebra generated by the $\binom{n}{2}$ freely noncommuting variables $\text{swap}_{ij}$ for $1\le i<j\le n$.

For $k\in\N$  let $C_k$ be the set of all products of $k$ variables $\text{swap}_{ij}$  that arise from permutations on a subset of $d+1$ letters (when each $\text{swap}_{ij}$ is replaced by the transposition $(i,j)$), which cannot be written as a product of less than $k$ transpositions (i.e., to each permutation we assign one product). 
Theorem \ref{th:iso} below states that $M^{\text{Sw}_d}_n(\mathbb{C})$ is isomorphic to 
the quotient of the free algebra
$\CC\langle\text{swap}_{ij}\colon 1\le i<j\le n\rangle$ modulo the relations
\begin{equation}\label{eq:defeqs_intro}
\begin{split}
\text{swap}_{ij}^2 &= 1,\\
\text{swap}_{ij}\text{swap}_{jk} &= \text{swap}_{ik}\text{swap}_{ij} = \text{swap}_{jk}\text{swap}_{ik},\\
\text{swap}_{ij}\text{swap}_{kl} &= \text{swap}_{kl}\text{swap}_{ij},\\
\sum_{s\in C_d}s
&{{}
=\sum_{s\in C_{d-1}}s
-\sum_{s\in C_{d-2}}s
+\cdots+(-1)^{d-1}
\sum_{s\in C_{1}}s
+(-1)^d.
}
\end{split}
\end{equation}
We acknowledge that this isomorphism may not be new to experts in representation theory, who will recognize the last equation in \eqref{eq:defeqs_intro} as the vanishing of 
an antisymmetrizer $\sum_{\sigma\in S_{d+1}}\operatorname{sgn}(\sigma)\sigma$ of $d+1$ vectors in $(\CC^d)^{\otimes n}$ 
(e.g., \cite[Section 11.6]{Probook}; note that every $\sigma\in S_{d+1}$ appears as a term in exactly one of the $c_k$, with its sign in \eqref{eq:defeqs_intro} matching the one in the antisymmetrizer). 
Nevertheless, in Section \ref{sec:qudit-rel} we provide an elementary and self-contained proof that the last relation in \eqref{eq:defeqs_intro} completely determines $M^{\text{Sw}_d}_n(\mathbb{C})$ as the quotient of the group algebra of $S_n$. To achieve this, the Schur-Weyl duality is invoked to assess the precise decomposition of $M^{\text{Sw}_d}_n(\mathbb{C})$ into irreps, as follows.
\begin{repthm}{th: s-w-swaps}
	The $d$-swap algebra $M^{\text{Sw}_d}_n(\mathbb{C})$ decomposes into a direct sum of simple algebras generated by the irreps $\rho_\lambda$ of $S_n$ corresponding to partitions of $n$ with at most $d$ rows,
	$$
	M^{\text{\textnormal{Sw}}_d}_n(\mathbb{C}) \cong \bigoplus_{\substack{\lambda\vdash n \\[.5mm] \het(\lambda)\le d}} \rho_{\lambda} (\mathbb{C}S_n).
	$$
\end{repthm}

\subsubsection{NPO hierarchy for the $d$-QMC problem}
A widely used approach for solving local Hamiltonian problems is through semidefinite programming (SDP) relaxations and noncommutative polynomial optimization (NPO) \cite{NPA08,NPA,DLTW,BKP16}.
While the $d$-QMC problem is already an SDP of the form
$$
\max_\rho \Tr (\rho H_G^d)\quad\text{ subject to }
\rho \succeq 0\text{ and }\Tr(\rho)=1,
$$
this formulation is hopeless for large $n$ because the semidefinite constraint is a $d^n\times d^n$ matrix.
Instead, one needs to explore the 2-locality of the $d$-QMC problem. As in the case of the classical MC, one can define a hierarchy of SDP relaxations which can be computed efficiently, i.e., in polynomial time, and give upper bounds to the true maximum eigenvalue of $H$
\cite{BH13, BGKT19, GP19, PT21, HO22}.
However, due to the exponential growth of the size of the matrices, only the first few levels are tractable. 

Having identified the $d$-swap algebra as a quotient of the free algebra $\CC\langle\text{swap}_{ij}\rangle$ in Section \ref{s:quotalg}, the $d$-QMC problem is written as a more efficient instance of a NPO problem in Section \ref{s:npo}. 
The $d$-QMC Hamiltonian $H^d_G$ is represented by an element $h_G\in \CC\langle\text{swap}_{ij}\rangle$, and 
its largest eigenvalue is
\begin{align*}
\alpha_*=&\min\left\{
\alpha:\alpha-h_G\text{ is a sum of hermitian squares in } \CC\langle\text{swap}_{ij}\rangle \text{ modulo \eqref{eq:defeqs_intro}}
\right\} .
\end{align*}
By adapting the non-commutative Sum-of-Squares hierarchy (ncSoS) from \cite{BCEHK24}, we give a sequence of semidefinite programs (SDPs) whose solutions approximate $\alpha_*$ from above. This scheme is specifically tailored to the algebraic structure of the swap matrices defining the $d$-QMC Hamiltonian. 
Since the $d$-swap algebra satisfies the symmetric group relations,
this hierarchy is exact at level $n-1.$
For large graphs $G$, only the first few levels of the hierarchy are practical for computations. For this reason we also focus on low-degree relations of swap matrices, which play a role in the construction of the SDPs for the first two levels of the hierarchy. Appendices \ref{a:d-1} and \ref{a:3and4} provide explicit bases for products of swap operators of low degree.

\subsubsection{Exact solutions for cliques and star graphs}
In Section \ref{sec:6} we turn our attention to computing the exact solutions to the $d$-QMC problem for certain families of graphs. To achieve this, we explore the isotypic structure of $d$-QMC Hamiltonians. Given a partition $\lambda\vdash n$, the $\lambda$-block of a $d$-QMC Hamiltonian is its isotypic component corresponding to $\rho_\lambda$ under the isomorphism of Theorem \ref{th: s-w-swaps} above.
The $d$-QMC problem for cliques $K_n$ on $n$ vertices is easiest to address as the isotypic blocks of the corresponding $d$-QMC Hamiltonian are scalar matrices (see Lemma \ref{lemma scalar}). Let $\eta_\lambda$ denote the eigenvalue of the block corresponding to the partition $\lambda.$ The following theorem gives an explicit expression of $\lambda$ in terms of its rows $\lambda_1,\ldots,\lambda_d$, and identifies the partition $\lambda$ that maximizes $\eta_\lambda$; i.e., the solution to the $d$-QMC problem for an $n$-clique $K_n$ is computed.

\begin{theorem}\label{th:clique}
    For any $\lambda \vdash n$ with rows $\lambda_1\ge\cdots\ge \lambda_d\geq 1$,
\begin{equation*}
\eta_\lambda=
n^2 +\frac{d(d-1)(2d-1)}{6}
-\sum_{k=1}^d\big( \lambda_k - (k-1)\big)^2.
\end{equation*}
The maximum value of $\eta_\lambda$ among all partitions $\lambda\vdash n$ with $\het(\lambda)\leq d$ is obtained at
	\begin{equation*}
	\lambda=\Big( 
\underbrace{1+\frac{n-r}d, \ldots, 1+\frac{n-r}d}_{r},\ 
\underbrace{\frac{n-r}d, \ldots, \frac{n-r}d}_{d-r}
	\Big)
	\end{equation*}
	for $n\equiv r\mod d.$ The solution to the $d$-QMC problem for an $n$-clique hence equals 
 \begin{equation*}
	     n^2 + (d-1) n + r^2 - r(d+1) - \frac{n^2-r^2}{d}.
	\end{equation*}
\end{theorem}
For the proof of Theorem \ref{th:clique} see Proposition \ref{prop:etaD}
and   Corollary \ref{prop:69++}.
We acknowledge that the solution to the $d$-QMC problem for an $n$-clique $K_n$ (the second part of Theorem \ref{th:clique}) has already been derived in \cite[Theorem 4.1]{grinko}, in the context of exchangeability (or clique graph extendibility) of Werner states. Nevertheless, the explicit formula for $\eta_\lambda$ in the first part of Theorem \ref{th:clique} is essential for tackling the $d$-QMC problem on a more general class of graphs.

Towards this goal, we refine a principle from \cite{BCEHK24} called \textit{clique decomposition}. 
It expresses the $d$-QMC Hamiltonian of a given graph as an alternating sum of Hamiltonians of cliques and simpler graphs in a way that is suitable for eigenvalue analysis. A graph with a simple clique decomposition is the star graph $\starn$ on $n$ vertices, on which we focus in Section \ref{ss:star}.
The relation
$$
 \starn = K_n - K_{n-1}
$$
holds, where the minus sign means deleting from $K_n$ the edges that appear in $K_{n-1}.$
This decomposition was used in \cite{BCEHK24} to show that the solution to the $2$-QMC problem for the star graph $\starn$ is $2n$, attained at the partition $\lambda=(n-1,1)$.
Extending this result, we solve the $d$-QMC problem for $\starn.$
\begin{repthm}{th:Eigstarn}
    If $\lambda = (\lambda_1,\ldots, \lambda_d) \vdash n$ has $d$ rows $\lambda_1 \geq \cdots \geq \lambda_d$, then the eigenvalues of the $\lambda$-block of the $d$-QMC Hamiltonian $H^d_\starn$  form a subset of 
$$
\{2(n-\lambda_1), \,2(n-\lambda_2+1),\, \ldots,\, 2(n-\lambda_d+d-1)\}
$$
containing the value $\eta_\star = 2(n-\lambda_d + d-1).$ Hence, the solution to the $d$-QMC problem for $\starn$ is $2(n + d-2),$ obtained by plugging $\lambda_d=1$ into $2(n-\lambda_d + d-1).$
\end{repthm}

\subsubsection{Exact solutions for complete bipartite graphs}

Star graphs are special cases of complete bipartite graphs.
In Section \ref{sec:bipartite}, we use the clique decomposition to exactly solve the $d$-QMC problem for complete bipartite graphs $K_{n-k,k}$ if $k\le 4$ or $d\le 3$.

The statement of Theorem \ref{th:bipintro} below involves three parameters associated to the problem data $(n,k,d)$: 
\begin{equation}\label{eq:defe0intro}
\begin{split}
    e_0 & =\max \left\{e\in\{1,\ldots,d-1\} \colon \floor{\frac{n-k}e} \geq \bigg\lceil\frac{k}{d-e}\bigg\rceil  \right\},\\
    e_1 & =\min \left\{e\in\{1,\ldots,d-1\} \colon
\floor{\frac{k}{d-e}} \geq 
     \bigg\lceil\frac{n-k}e\bigg\rceil \right\}, \\
    e^* & = \frac{d}{2} + \frac{n-2k}{2(q+1)}, \quad 
    q=\floor{\frac nd}.
\end{split}
\end{equation}
If the set in the definition of $e_1$ is empty
(e.g., $n=5$, $k=2$, $d=2$), we set $e_1=d-1$. 

A partition $\lambda= (\lambda_1,\ldots,\lambda_\ell)\vdash m$ is \textbf{balanced} if $\lambda_1-\lambda_\ell\leq1$ (e.g., the optimal partition in Theorem \ref{th:clique}); note that such $\lambda$ is uniquely determined by $m$ and $\ell$. 
A subpartition of $\lambda$ is a partition obtained from $\lambda$ by discarding some rows. Then,
\[
\begin{split}
    \mathfrak E & :=\{e \mid \text{the balanced partition of $n-k$ of height $e$} \text{ is a subpartition of } \\
    &  \phantom{:=\{e \mid \ }
\text{ the balanced partition of $n$ of height $d$}     \}
\end{split}
\]
is a discrete interval in $\{e_0, e_0+1,\ldots,e_1\}$
by Lemma \ref{lem:frakinterval}, and contains
$\{e_0+1, e_0+2,\ldots,e_1-1\}$ by Lemma \ref{lem:frakinterval+}.

\begin{theorem}\label{th:bipintro}
Let $2k\le n$.
\begin{enumerate}[\rm (1)]
\item
The solution to the $2$-QMC problem for $K_{n-k,k}$ is
$2k(n-k+1)$. The solution to the $3$-QMC problem for $K_{n-k,k}$ is $2 ( k +1) (n-k)$ if $n<3k$, and $2k(n-k+2)$ if $n\ge3k$.

\item Let $k\le 4$ and $d<n$.
Let $e_\star$ be the closest integer in 
$\mathfrak E$ to $e^*$. Then the solution to the $d$-QMC problem for $K_{n-k,k}$ is the biggest of the three values of
\[\eta_\la-\eta_\mu-\eta_\nu\]
for the balanced $\mu\vdash n-k$ of height $e$, the balanced $\nu\vdash k$ of height $d-e$, and $\la\vdash n$ obtained by merging (and sorting) $\mu$ and $\nu$, where $e$ is one of $e_0,e_\star,e_1$.
\end{enumerate}
\end{theorem}

The proof of Theorem \ref{th:bipintro} (see Corollaries, \ref{cor:d=2}, \ref{cor:d=3} and \ref{cor:dbip}) occupies the entire 
Section \ref{sec:bipartite}. 

\begin{remark}
    The case $d=2$ in part (1) of Theorem \ref{th:bipintro} recovers the result originally established in \cite{LM62}.
\end{remark}

\begin{example}\label{exa:jakab}
Let us compare the $d$-QMC problem for complete bipartite graphs and the formula in Theorem \ref{th:bipintro} with the assertions in the earlier work \cite{jakab}. Therein, the authors investigate the extendibility of bipartite Werner states of local dimension $d$. The parameters $n_A=n-k$ and $n_B=k$ refer to the number of parties (divided into two groups pertaining to the original two parties) that share an extension of the original bipartite state. By \cite[Eq. (4)]{jakab}, the investigated problem reduces to finding the smallest eigenvalue $\alpha_{n-k,k,d}$ of
$$\frac{1}{k(n-k)}\sum_{\substack{i=1,\dots,k \\ j=k+1,\dots,n}} \textnormal{Swap}^{(d)}_{ij}.$$
If $\beta_{n-k,k,d}$ denotes the solution to the $d$-QMC problem for $K_{n-k,k}$, then $\beta_{n-k,k,d}$ is the largest eigenvalue of
$$H_{K_{n-k,k}}^{d} = \sum_{\substack{i=1,\dots,k \\ j=k+1,\dots,n}} 
2\left(I - \textnormal{Swap}^{(d)}_{ij}\right),$$
and therefore
\begin{equation}\label{e:comparison}
\beta_{n-k,k,d} = 2k(n-k)(1-\alpha_{n-k,k,d}).
\end{equation}
A formula for $\alpha_{k,n-k,d}$ is proposed in \cite[Eq. (16)]{jakab}. There may be a typographical lapse in \cite[Eq. (16)]{jakab}; the said formula involves the expression $\frac{n_A}{n_A+n_B}=\frac{n-k}{n}$, while the preceding paragraph in \cite{jakab} perhaps suggests that $\frac{d n_A}{n_A+n_B}=\frac{d(n-k)}{n}$ should be used in its place instead. 
Both variants of this formula yield values compatible with those from Theorem \ref{th:bipintro} for small values of parameters $n,k,d$. 
However, in general there are disparities, as follows.

Let $(n,k,d)=(6,3,4)$. This is a balancing triple, and $\beta_{n-k,k,d}=28$ by Theorem \ref{th:dbip}. 
More precisely, in this case we have $\mathfrak{E}=\{2\}$ and $(e_0,e^*,e_1)=(1,2,3)$, and the maximal value 28 is given by Theorem \ref{th:bipintro} at $\lambda=(2,2,1,1)$ and $\mu=\nu=(2,1)$.
Alternatively, this can be verified by a brute-force calculation of the largest eigenvalue of the $4096\times 4096$ Hamiltonian $H_{K_{n-k,k}}^{d}$. In contrast, both variants of \cite[Eq. (16)]{jakab} evaluate $\alpha_{n-k,k,d}$ as $-\frac13$, 
even though the direct calculation of eigenvalues for this $H_{K_{n-k,k}}^{d}$ shows that $\alpha_{n-k,k,d}=-\frac59$. Thus, Theorem \ref{th:bipintro} amends the assertions of \cite{jakab}.
\end{example}

\subsubsection{Separation of irreps} 

In Section \ref{sec:sep}, we refine the NPO hierarchy for the $d$-QMC problem, with the intention of calculating the maximum eigenvalue of the $\lambda$-block in $H_G^d$ corresponding to a given irrep of $S_n$ given by the partition $\lambda\vdash n$. The idea is to find low-degree polynomials that distinguish distinct irreps indexed by partitions with at most $d$ rows. 
In Theorem \ref{t:sep_all_rows} we show that irreps of $S_n$ are separated by the polynomials $c_k$ of degree $k$ {corresponding to the sum of all $(k+1)$-cycles in $S_n$ (see also \eqref{e:ck} below)}. This leads to an NPO formulation of the localized $d$-QMC problem.

\begin{theorem}\label{t:localized}
Let $\lambda\vdash n$ have at most $d$ rows.
There are constants $\gamma_1,\dots,\gamma_{d-1}\in\ZZ$ such that
the largest eigenvalue of the $\lambda$-block in $H_G^d$ equals
\begin{align*}
&\min\big\{
\alpha:\alpha-h_G\text{ is a sum of hermitian squares in } \CC\langle {\rm swap}_{ij}\rangle \\
&\qquad\qquad\qquad\qquad\text{ modulo \eqref{eq:defeqs_intro} and }c_1=\gamma_1,\, \dots,\, c_{d-1}=\gamma_{d-1}
\big\} .
\end{align*}
\end{theorem}

The NPO problem in Theorem \ref{t:localized}  can be tackled with an NPA-like hierarchy of SDPs, and the values $\gamma_k$ can be evaluated using explicit Lassalle's character formulas for cycles in $S_n$ \cite{Lassalle}.

We also consider distinguishability of irreps of $S_n$ from the perspective of the $d$-QMC problem. As $d$-QMC Hamiltonians can only admit $\lambda$-blocks for $\lambda\vdash n$ with at most $d$ rows, one can only hope to distinguish such irreps through the $d$-QMC problem.
For $d=2$, it is known that the value $\eta_\lambda$ itself separates irreps with at most two rows \cite{BCEHK24}; in other words, the spectrum of the Hamiltonian for $K_n$ separates irreps with at most two rows.
This is not the case anymore for $d>2$.
However, we show that for $d=3$, the spectrum of the $3$-QMC Hamiltonian for $\starn$ separates irreps with at most three rows.

\begin{theorem}\label{th:sep}
Let $\lambda,\mu\vdash n$ be partitions with at most three rows.
Then $\lambda=\mu$ if and only if the spectra of the $\lambda$-block and the $\mu$-block of $H_{\starn}^3$ coincide.
\end{theorem}

See Proposition \ref{p:sep3} for the proof.
For general $d,n\in\N$, we leave it as an open problem whether there exist graphs $G_1,\dots,G_\ell$ on $n$ vertices such that for all $\lambda,\mu\vdash n$ with at most $d$ rows, $\lambda=\mu$ if and only if the spectra of the $\lambda$-block and the $\mu$-block of $H_{G_i}^d$ coincide for all $i=1,\dots,\ell$.

\subsection{Comparison with the work of Carlson, Jorquera, Kolla, Kordonowy, Wayland}\label{sec:gm} The $d$-QMC problem was  introduced in \cite{CJKKW23+}, where it was defined via a generalization of the Pauli matrices to arbitrary size $d\times d,$ called {Gell-Mann matrices}. 

\subsubsection{The Gell-Mann matrices} \label{subsec:gm} For each $d \ge 2,$ there is a family of $d^2-1$ trace zero self-adjoint matrices, which, together with the identity $I_d,$ form a basis for $M_d(\CC).$
More concretely, for $d=2$ these are the Pauli matrices, and for  $d\ge3$ there are three kinds of $d \times d$ Gell-Mann matrices (see Appendices \ref{subsec: gm3} and \ref{subsec: gm4}, where these matrices are given explicitly for $d=3$ and $d=4$, respectively):
\begin{equation}
\begin{split}
	\text{symmetric:}\quad \quad \quad \ \  &\lambda_{a,b}^{\text{sym}_d} = E_{a,b} + E_{b,a},\quad \quad \quad \ 1\leq a<b \leq d, \\
	\text{antisymmetric:} \quad \quad&\lambda_{a,b}^{\text{asym}_d} = \mathfrak{i}(E_{b,a} - E_{a,b}),\quad \quad 1\leq a<b \leq d, \\
	\text{diagonal:}\quad \quad \quad \quad \ \,
	&\lambda_k^d = \lambda_k^{d-1} \oplus 0, \quad \quad \quad \quad \quad \ \ 2\leq k < d,\\
	&\lambda_d^d = \sqrt{\frac{2}{d(d-1)}}\big(I_{d-1} \oplus (1-d)\big).
\end{split}
\end{equation}
Here, $E_{a,b}$ are standard matrix units, and $I_{d-1}$ is the identity  of size $d-1.$
Note that there are $\binom{d}{2}$ (non-diagonal) symmetric, $\binom{d}{2}$ antisymmetric and $d-1$ diagonal matrices. Summing up we get $2\binom{d}{2} + d-1 = d(d-1)+d-1 = d^2-1$ as expected.
For fixed $d$ denote by $GM(d)$ the set consisting of the $d^2-1$ Gell-Mann $d\times d$ matrices of size $d \times d$.

As for the Pauli matrices, define for a fixed $n \in \NN,$
$$
\lambda^{j} := \underbrace{I \otimes \cdots \otimes I}_{j-1} \otimes \lambda \otimes I \otimes \cdots \otimes I \in M_{d^n}(\CC)
$$
for any $\lambda \in GM(d)\cup\{I\}$. By definition, $\lambda_1^i$ and $\lambda_2^j$ commute if $i \neq j$ and  $\lambda_1,\lambda_2 \in GM(d)$, and
\begin{align}\label{eq: basis-gm d}
	\left\{\lambda_{1}^1 \lambda_{2}^2 \cdots \lambda_{n}^n \ | \ \lambda_{j} \in GM(d)\cup\{I\}\text{ for }  j=1,\ldots,n\right\}
\end{align}
is a basis of $M_{d^n}(\CC).$ 

\subsubsection{The $d$-QMC Hamiltonian via the Gell-Mann matrices}\label{ssec:gelmann}
The formula \eqref{eq:swapviap} expressing the swap matrices Swap$^{(2)}_{ij}$ in terms of products of Pauli matrices, i.e., with respect to the basis \eqref{eq: basis-gm d} for $d=2,$ can be generalized to an arbitrary $d$ as shown below.

\begin{proposition}\label{prop:swaptogm}
	 For any $i<j$ we have
	\begin{align}\label{eq: sw-id}
		\text{\textnormal{Swap}}^{(d)}_{ij} = \frac{1}{d} \,I + \frac{1}{2} \sum_{\lambda \in GM(d)} \lambda^i \lambda^j.
	\end{align}
\end{proposition}

\begin{proof} Denote the right-hand side of \eqref{eq: sw-id} by $R.$
	Since \eqref{eq: sw-id} only involves two indices $i,j$ we may assume $n=2$ and prove that $R$ acts the same as Swap$_{ij}^{(d)}$ on basis vectors of the form $v_{p,q} = e_p \otimes e_q \in \mathbb{C}^d \otimes  \mathbb{C}^d$ for $1 \leq p,q \leq d.$ 
	
	If $p=q,$ then only the diagonal Gell-Mann matrices $\lambda_k^d$ for $p\leq k \leq d$ act nontrivially on $v_{p,q},$
	\begin{align*}
		Rv_{p,p} & = \frac{1}{d} e_p \otimes e_p + 
		\frac{1}{2} \left(\sqrt{\frac{2}{p(p-1)}}(1-p) e_p\right) \otimes \left(\sqrt{\frac{2}{p(p-1)}}(1-p) e_p\right)\\
		& \phantom{{}={}} + \frac{1}{2} \sum_{j=p+1}^d \left(\sqrt{\frac{2}{j(j-1)}} e_p\right) \otimes \left(\sqrt{\frac{2}{j(j-1)}}e_p \right) \\
		&= \left(\frac{1}{d} + \frac{p-1}{p} + \sum_{j=p+1}^d \frac{1}{j(j-1)}  \right)v_{p,p} =v_{p,p}.
	\end{align*}
 	Finally, if $p<q,$ then in addition to $\lambda_k^d$ for $q\leq k \leq d,$ also $\lambda_{p,q}^{\text{sym}_d}$ and $\lambda_{p,q}^{\text{asym}_d}$ act nontrivially on $v_{p,q}$ and they both map it to $v_{q,p},$
 	\begin{align*}
 		Rv_{p,q} &= \frac{1}{d} e_p \otimes e_q + 
 		\frac{1}{2} \left(\sqrt{\frac{2}{q(q-1)}} e_p\right) \otimes \left(\sqrt{\frac{2}{q(q-1)}}(1-q) e_q\right)\\
 		& \phantom{{}={}} + \frac{1}{2} \sum_{j=q+1}^d \left(\sqrt{\frac{2}{j(j-1)}} e_p\right) \otimes \left(\sqrt{\frac{2}{j(j-1)}}e_q \right) \\
 		& \phantom{{}={}} + \frac{1}{2}e_q \otimes e_p + \frac{1}{2}e_q \otimes e_p \\
 		&=\left(\frac{1}{d} - \frac{1}{q} + \sum_{j=q+1}^d \frac{1}{j(j-1)}  \right)v_{p,q} + v_{q,p} =v_{q,p}. 
        \tag*{\qed}
 	\end{align*}
\renewcommand{\qedsymbol}{}
\end{proof}

Note that by \eqref{eq: sw-id}, the $d$-QMC Hamiltonian can be expressed in terms of the $d \times d$ Gell-Mann matrices as 
\begin{align}\label{eq:HviaGM}
	H_G^{d} = \sum_{(i,j) \in \text{E}(G)} 
	2w_{ij} \, \left(  
	\frac{d-1}{d}I - 
	\frac{1}{2} \sum_{\lambda \in GM(d)} \lambda^i \lambda^j
	\right),
\end{align}
where $GM(d)$ denotes the set of all $d \times d$ Gell-Mann matrices. 
This is in fact the form of the $d$-QMC Hamiltonian used in \cite{CJKKW23+}. 

An advantage of our approach is that it incorporates algebraically \emph{all} symmetries inherent to the $d$-QMC problem.
Computations (such as the NPO relaxations in Section \ref{s:npo}, the clique decomposition in Section \ref{sec:6} or the decompositions along irreps) with $H_G^d$ as in \eqref{eq:qmcd} scale better with both $n$ and $d$ or are only made possible once one passes to qudit swap matrices championed in this paper.

\subsection*{Acknowledgments}
The authors thank 
Darij Grinberg for sharing enlightening comments and his expertise in representation theory, and
Dmitry Grinko for bringing to our attention the relationship between the Quantum Max $d$-Cut and extendibility of symmetric bipartite states, and recent developments on the latter topic.
{The authors thank the anonymous referees for their valuable feedback, which helped strengthen the clarity and rigor of this paper.}

 \section{Preliminaries on the representations of the symmetric group}

In this section we review some standard elements of the representation theory of symmetric groups that are used throughout this paper; for a comprehensive source, see, e.g., \cite{FH,Probook}.
For $n \in \NN$ we denote by $S_n$ the symmetric group on $n$ elements, i.e., the group of permutations of $n$ elements. A \textbf{representation} of $S_n$ is a group homomorphism $\rho : S_n \to \text{GL}(V)$ for a vector space $V,$ also called $S_n$-module.  Any representation $\rho$ of $S_n$ naturally defines a representation $\tilde{\rho}$ of the \textbf{group algebra} $\CC [S_n]$ of $S_n.$ The resulting representation $\tilde{\rho} : \CC [S_n] \to \text{End}(V)$ is defined by
 $$
 \tilde{\rho}\left(\sum_{\pi\in S_n} \alpha_\pi \pi  \right) = \sum_{\pi\in S_n} \alpha_\pi \rho(\pi).
 $$
 For simplicity, the letter $\rho$ often refers to both the representation of $S_n$ and the representation of $\CC[S_n].$ 

\subsection{Irreducible representations of the symmetric group}
 An $S_n$-module $V$ is \textbf{irreducible} if its only nontrivial submodule is $V.$ Throughout we abbreviate irreducible representation by \textit{irrep} and use the terms irrep and irreducible module interchangeably.  Note that by Maschke's theorem \cite[Section 6.1.5]{Probook}, any $S_n$-module $V$ decomposes as a direct sum of irreducible $S_n$-modules.

 It is well-known that the irreducible $S_n$-modules are indexed by partitions $\lambda$ of $n$ (often denoted by $\lambda \vdash n$), where
 $$
 \lambda = (\lambda_1, \ldots, \lambda_k)\in\N^k, \quad \lambda_1 \geq \cdots \geq \lambda_k > 0, 
 \quad \sum_{i=1}^k \lambda_i = n.
 $$
The number of summands $k$ is called the \textbf{height} of $\lambda$ and denoted $k=\het(\lambda)$. For $\ell<k$, the partition $(\lambda_1,\ldots,\lambda_\ell)$ is a \textbf{head}, and the partition $(\lambda_{e+1},\ldots,\lambda_f)$ is a \textbf{tail} of $\lambda$. 
A partition $\lambda \vdash n$ is usually depicted by its \textbf{Young diagram}. A Young diagram of shape $\lambda$ has $k$ rows and the $i$th row consists of $\lambda_i$ boxes. For example, if $\lambda= (5,3,2) \vdash 10,$ then $\het(\lambda)=3$ and its Young diagram is
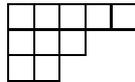
\[
\ytableausetup{smalltableaux}
\ydiagram{5,3,2}
\]

\vspace{2mm}

A \textbf{Young tableau} of shape $\lambda$ is a Young diagram whose boxes are filled with numbers $1,\ldots,n$ such that each box gets a different integer.
The symmetric group $S_n$ acts on a Young tableau $t$ of shape $\lambda \vdash n$ by permuting the entries of $t.$ This action defines an equivalence relation, where two tableaux are equivalent if one can be obtained from the other by permuting the entries within each of the rows.
Equivalence classes with respect to this relation are called \textbf{tabloids}.

The irreducible $S_n$-module $V_\lambda$ corresponding to the partition $\lambda \vdash n$ is called a \textbf{Specht module} and it has a well-known basis consisting of  \textbf{polytabloids} 
	$$
	e_T = \sum_{\pi \in C_T} \text{sgn}(\pi) \pi\{T\}.
	$$ 
	Here $T$ ranges over all tabloids of shape $\lambda$, $C_T$ is the set of all permutations that permute the elements only within the columns of $T$ and for each $\pi \in C_T,$  $\pi\{T\}$ is the tabloid obtained from $T$ by permuting the entries according to $\pi.$

\subsection{Schur-Weyl duality}\label{ssec:sw}

As a complex representation of $S_n,$ the $d$-swap algebra $M^{\text{Sw}_d}_n(\mathbb{C})$ is semisimple. Key to solving the $d$-QMC problem for certain graphs is the precise knowledge of the block decomposition of  $M^{\text{Sw}_d}_n(\mathbb{C})$ into simple matrix algebras. 
We now explain how this block decomposition 
can be deduced using the Schur-Weyl duality of the actions of $S_n$ and GL$_d(\mathbb{C})$ on $\big(\mathbb{C}^d\big)^{\otimes n}$ (see e.g. \cite[Section 6.1]{FH} or \cite[Section 7.1.2]{Probook}). 
The natural representation of GL$_d(\mathbb{C})$ on $\big(\mathbb{C}^d\big)^{\otimes n}$ is defined by the diagonal action; $g \in $ GL$_d(\mathbb{C})$ acts on the tensor product of $v_1, \ldots,v_n \in \mathbb{C}^d$ by
$$
\zeta_n(g) (v_1 \otimes \cdots \otimes v_n) = g(v_1) \otimes \cdots \otimes g(v_n).
$$
The actions of $S_n$ and GL$_d(\mathbb{C})$ on $\big(\mathbb{C}^d\big)^{\otimes n}$ commute and there is a bijection between the irreducible modules of $S_n$ and GL$_d(\mathbb{C}).$ 
This interplay between permutations of particles and change of coordinates is indispensable in investigating qudit systems, see e.g. \cite{GNW21}.
Furthermore, if $\lambda$ is a partition of $n,$ then to the irreducible module $V_\lambda$ of $S_n$ corresponds (up to isomorphism) exactly one irreducible module $L_\lambda$ of GL$_d(\mathbb{C})$ and $L_\lambda$ are precisely the maps from $V_\lambda$ to $\big(\mathbb{C}^d\big)^{\otimes n}$ that commute with the action of $S_n,$
$$
L_\lambda = \text{Hom}_{S_n}(V_\lambda,\big(\mathbb{C}^d\big)^{\otimes n}).
$$

It is well-known that $L_\lambda$ is nonzero precisely when $\lambda$ is a partition with at most $d$ rows \cite[Proposition 9.3.1]{Probook}. In fact, the dimensions of the modules $L_\lambda$ can be explicitly computed by the Weyl character formula \cite[Section 9.6.2]{Probook}.

The next proposition is a restatement of the Schur-Weyl duality \cite[Theorem 9.3.1]{Probook} for $S_n$ and GL$_d(\mathbb{C})$, taking into account \cite[Proposition 9.3.1]{Probook}.

\begin{proposition}
	The algebras $M^{\text{\textnormal{Sw}}_d}_n(\mathbb{C})$  and $\zeta_n(\text{\textnormal{GL}}_d(\mathbb{C}))$ are centralizers of one another inside \text{\textnormal{End\big($\big(\mathbb{C}^d\big)^{\otimes n}$\big)}} = $M_{d^n}(\mathbb{C}),$ and with respect to the action of the direct product $\text{GL}_d(\mathbb{C}) \times S_n,$ the space $\big(\mathbb{C}^d\big)^{\otimes n}$ decomposes as
	\begin{equation}\label{eq:decomp}
	    \big(\mathbb{C}^d\big)^{\otimes n} \cong \bigoplus_{\substack{\lambda\vdash n \\[.5mm] \het(\lambda)\le d}} L_\lambda \otimes V_\lambda. 
	\end{equation}
\end{proposition}

Since $S_n$ acts trivially on each $L_\lambda,$ the space $\big(\mathbb{C}^d\big)^{\otimes n}$ decomposes as an $S_n$-module into irreducible modules $V_\lambda$ (with multiplicities) as follows
$$
\big(\mathbb{C}^d\big)^{\otimes n} \cong \bigoplus_{\substack{\lambda\vdash n \\[.5mm] \het(\lambda)\le d}} V_\lambda^{\text{ dim}(L_\lambda)}. 
$$

As a corollary we get the desired decomposition of the $d$-swap algebra $	M^{\text{Sw}_d}_n(\mathbb{C}).$
\begin{theorem}\label{th: s-w-swaps}
	The $d$-swap algebra decomposes into a direct sum of simple algebras generated by the irreps $\rho_\lambda$ of $S_n$ corresponding to partitions of $n$ with at most $d$ rows,
	$$
	M^{\text{\textnormal{Sw}}_d}_n(\mathbb{C}) \cong \bigoplus_{\substack{\lambda\vdash n \\[.5mm] \het(\lambda)\le d}} \rho_{\lambda} (\mathbb{C}S_n) \cong \bigoplus_{\substack{\lambda\vdash n \\[.5mm] \het(\lambda)\le d}}M_{{\rm dim}(V_\lambda)}(\CC).
	$$
\end{theorem}

\begin{proof}
    Using \eqref{eq:decomp} we deduce that as a GL$_d(\CC)$-module, the space  $(\CC^d)^{\otimes n}$ decomposes as
    $$
\big(\mathbb{C}^d\big)^{\otimes n} \cong \bigoplus_{\substack{\lambda\vdash n \\ \het(\lambda)\le d}} L_\lambda^{\text{ dim}(V_\lambda)}. 
$$
Now considering the GL$_d(\CC)$-endomorphisms on both sides gives the desired result. Indeed, the GL$_d(\CC)$-endomorphisms of $(\CC^d)^{\otimes n}$ are by definition the endomorphisms of $(\CC^d)^{\otimes n}$ that commute with the action of GL$_d(\CC).$ By the Schur-Weyl duality these are precisely the elements from the $d$-swap algebra $M^{\text{\textnormal{Sw}}_d}_n(\mathbb{C}) .$ On the other hand, 
\begin{align*}
    \text{End}_{\text{GL}_d(\CC)}\Bigg( \bigoplus_{\substack{\lambda\vdash n \\[.5mm] \het(\lambda)\le d}} L_\lambda^{\text{ dim}(V_\lambda)} \Bigg) &=
 \bigoplus_{\substack{\lambda\vdash n \\[.5mm] \het(\lambda)\le d}} \text{End}_{\text{GL}_d(\CC)} \Big(L_\lambda^{\text{ dim}(V_\lambda)}\Big)\\
 &= \bigoplus_{\substack{\lambda\vdash n \\[.5mm] \het(\lambda)\le d}} M_{\text{ dim}(V_\lambda)} \Big( \text{End}(L_\lambda)^{\text{op}} \Big)\\
 &\cong \bigoplus_{\substack{\lambda\vdash n \\[.5mm] \het(\lambda)\le d}}M_{\text{ dim}(V_\lambda)}(\CC).
 \tag*{\qed}
\end{align*}
\renewcommand{\qedsymbol}{}
\end{proof}

\begin{remark}
    The dimension of any irreducible $S_n$-module $V_\lambda$ can be computed via the well-known hook length formula (see Section \ref{sec:6} for some explicit calculations).
\end{remark}

\section{Degree-reducing relation for qudit swap matrices}\label{sec:qudit-rel}

    It is known (see, e.g. \cite[Section 9.3]{Probook}) that, in addition to the symmetric group axioms, the swap matrices Swap$^{(d)}_{ij}$ also satisfy a degree-reducing relation 
    of degree $d.$ For instance, 
    \eqref{eq:degred} above is such an equation for $d=2$.
    We now present the general form of 
    the degree-reducing relation for general $d$,
    and, for the reader's convenience, give an elementary and self-contained proof. 

Since the symmetric group is generated by transpositions, each permutation can be written as a product of transpositions (non-uniquely).
For fixed $d$ and $k=1,\ldots, d$ let $C_k$ be the set of all products of $k$ swap matrices that arise from permutations on a subset of $d+1$ letters, which cannot be written as a product of less than $k$ transpositions (i.e., to each permutation we assign one product).
The next proposition gives 
the analog of the relation \eqref{eq:degred} in the case of a general $d.$

\begin{proposition}\label{prop:degred} The swap matrices Swap$^{(d)}_{ij}$ satisfy the following degree-reducing relation
    \begin{align}\label{eq:qdit-rel}
	\sum_{s \in C_d} s = \sum_{s \in C_{d-1}} s - \sum_{s \in C_{d-2}} s + \cdots + (-1)^{d-1} \sum_{s \in C_1} s + (-1)^d \cdot 1.
\end{align}
\end{proposition}

\begin{remark}\label{rem:nonred} We often simplify the notation of sums involving products of swap matrices (e.g., the ones in \eqref{eq:qdit-rel}) by summing over (a subset of) the symmetric group and using the fact that every product of swap matrices corresponds to a permutation of the tensor factors of $(\CC^d)^{\otimes n}.$ On the other hand, every such permutation can be written as a product of swap matrices Swap$^{(d)}_{ij}$ in a non-redundant way, i.e., relations \eqref{eq:snrel} are applied to simplify the expression as much as possible. 
\end{remark}

\begin{remark}
    Note that it is enough to assume that the $d+1$ letters in Equation \eqref{eq:qdit-rel} are the numbers $1,\ldots,d+1.$ In fact, an analogous equation with indices from some other $(d+1)$-subset  $J$ of $\{1,\ldots,n\}$ can be obtained by conjugating Equation \eqref{eq:qdit-rel} with any permutation sending $\{1,\ldots,d+1\}$ to $J.$
\end{remark}

\begin{example}
    Equation \eqref{eq:qdit-rel} for $d=3$  with $(i,j,k,l)=(1,2,3,4)$ is as follows:
\begin{equation*}\label{eq:swap-rel}
	\begin{split}
		\text{Swap}^{(3)}_{12}\,&\text{Swap}^{(3)}_{23}\,\text{Swap}^{(3)}_{34}\,  + \, \text{Swap}^{(3)}_{12}\,\text{Swap}^{(3)}_{24}\,\text{Swap}^{(3)}_{4,3}\,  + \, \text{Swap}^{(3)}_{13}\,\text{Swap}^{(3)}_{3,2}\,\text{Swap}^{(3)}_{24} \,  + \\ \, \text{Swap}^{(3)}_{13}\,&\text{Swap}^{(3)}_{34}\,\text{Swap}^{(3)}_{4,2} \,  + \, \text{Swap}^{(3)}_{14}\,\text{Swap}^{(3)}_{4,2}\,\text{Swap}^{(3)}_{23} \,  + \, \text{Swap}^{(3)}_{14}\,\text{Swap}^{(3)}_{4,3}\,\text{Swap}^{(3)}_{3,2}  = \\
		\text{Swap}^{(3)}_{12}&\,\text{Swap}^{(3)}_{13}\,  + \,\text{Swap}^{(3)}_{12}\,\text{Swap}^{(3)}_{14}\,  + \, \text{Swap}^{(3)}_{12}\,\text{Swap}^{(3)}_{23}\,  + \,\text{Swap}^{(3)}_{12}\,\text{Swap}^{(3)}_{24}\,  + \\ 
		\text{Swap}^{(3)}_{12}&\,\text{Swap}^{(3)}_{34}\,  + \, 
		\text{Swap}^{(3)}_{13}\,\text{Swap}^{(3)}_{14}\,  + \,\text{Swap}^{(3)}_{13}\,\text{Swap}^{(3)}_{24}\,  + \, \text{Swap}^{(3)}_{13}\,\text{Swap}^{(3)}_{34}\,  + \\
		& \ \text{Swap}^{(3)}_{14}\,\text{Swap}^{(3)}_{23}\,  + \, \text{Swap}^{(3)}_{23}\,\text{Swap}^{(3)}_{24}\,  + \, \text{Swap}^{(3)}_{23}\,\text{Swap}^{(3)}_{34}\,  - \\
		\text{Swap}^{(3)}_{12}&\,  - \,
		\text{Swap}^{(3)}_{13}\,  - \,\text{Swap}^{(3)}_{14}\,  - \,\text{Swap}^{(3)}_{23}\,  - \,\text{Swap}^{(3)}_{24}\,  - \,\text{Swap}^{(3)}_{34}\,  + \, 1.
	\end{split}
\end{equation*}
\end{example}

\begin{remark} 
Equation \eqref{eq:qdit-rel} can be written in a more condensed form as 
$$
\sum_{s \in S_{d+1}} \text{sgn}(s)\, s = 0,
$$
saying that the \textit{antisymmetrizer} on $d+1$ letters equals zero (see \cite[Section 9.3.1]{Probook}). Here sgn denotes the sign of a permutation and each $s$ is expressed in terms of the swap matrices Swap$^{(d)}_{ij}$ in a non-redundant way as explained in Remark \ref{rem:nonred}.
\end{remark}

Next is a preliminary lemma on the way to give an elementary proof of Proposition \ref{prop:degred}.

\begin{lemma}\label{lemma 1.1} Let $v$ be of the form $e_{i_1} \otimes e_{i_2} \otimes \cdots \otimes e_{i_{k+1}}$ where $i_m \in \{1,2,\ldots, d\}$ for $m=1, \ldots, k+1$ and exactly two of the $i_m$ are the same. 
	
	Then any cycle of length $k+1$ acts on $v$ as a product of $k-1$ transpositions and corresponds to a product of two smaller disjoint cycles (so of length at most $k$, singletons also count)  such that none of them acts on a subset of the tensor factors of $v$ containing two equal factors.
\end{lemma}

\begin{proof}
 Let $\sigma$ be a cycle of length $k+1$ and suppose without loss of generality that $1$ is the index $i_m$ that appears twice in $v$ (i.e., $v$ has two copies of $e_1$). Divide the letters $1,\ldots,k+1$ into disjoint tuples $B_1, B_2$ such that for each $i,$ the indices of the factors of $v$ and $\sigma(v)$ at position $k \in B_i$ are in the same tuple and none of the tuples contains two copies of $1.$ Then $\sigma$  acts on $v$ as the product of two disjoint cycles, represented by the tuples $B_1,B_2.$ This will also prove that the above set decomposition of $\{1,\ldots,k+1\}$ can be done in a unique way. 
	
	The algorithm to find the tuples $B_i$ is the following: find a position $j_1 \in \{1,\ldots,k+1\}$ of one of the two $e_1$ in $v$ and assign $j_1$ to $B_1.$ Then consider the factor $e_{i_{m_1}}$ of $\sigma(v)$ at position $j_1.$ If $i_{m_1} = 1,$ then we add no more elements to $B_1$ and start the process all over again with the second factor $e_1$ of $v$ whose position $j_2$ is assigned to $B_2.$ Otherwise, if $m_1 \neq 1,$ add $j_2$ to $B_1$ and find the position $j_2$ of $e_{i_{m_1}}$ in $v.$  Consider the factor $e_{i_{m_2}}$ of $\sigma(v)$ at position $j_2$ and repeat the procedure from before according to whether $i_{m_2}$ equals $1$ or not. Proceed until the basis vector $e_{i_{m_r}}$ equals $e_1$ for some $r$ and we cover all the letters.
	
	Since $v$ only has one index that repeats twice and all the other indices are distinct, this construction gives the desired cyclic decomposition of $\sigma$ into precisely two shorter cycles.
\end{proof}

\begin{example}
	We provide a concrete example for Lemma \ref{lemma 1.1} in the case $d = 5$ and $n=6.$ Let $v = e_1 \otimes e_2 \otimes e_3 \otimes e_4 \otimes e_5 \otimes e_1$ and $\pi = (1 2 3 4 5 6)$ and set
	\begin{align*}
		\sigma_i = \rho^{(5)}_ 6(\pi^i),  \quad i = 1, \ldots, 6.
	\end{align*}
	Then 
	\begin{align*}
		\sigma_1 (v) &= e_1 \otimes e_1 \otimes e_2 \otimes e_3 \otimes e_4 \otimes e_5 =
		\rho^{(5)}_6 \big( (1)(23456) \big) (v),\\
		\sigma_2 (v) &= e_5 \otimes e_1 \otimes e_1 \otimes e_2 \otimes e_3 \otimes e_4 =
		\rho^{(5)}_6 \big( (135)(246) \big) (v),  \\
		\sigma_3 (v) &= e_4 \otimes e_5 \otimes e_1 \otimes e_1 \otimes e_2 \otimes e_3 =
		\rho^{(5)}_6 \big( (14)(25)(36) \big) (v),  \\
		\sigma_4 (v) &= e_3 \otimes e_4 \otimes e_5 \otimes e_1 \otimes e_1 \otimes e_2 =
		\rho^{(5)}_6 \big( (135)^2(246)^2 \big) (v),  \\
		\sigma_5 (v) &= e_2 \otimes e_3 \otimes e_4 \otimes e_5 \otimes e_1 \otimes e_1 =
		\rho^{(5)}_6 \big( (12345)(6) \big) (v).  
	\end{align*}
\end{example}

\begin{proof}[Proof of Proposition \ref{prop:degred}] 
	First note that it is enough to verify Equation \eqref{eq:qdit-rel} on basis vectors $v$ of $\big(\mathbb{C}^d\big)^{\otimes\, (d+1)}$ of the form $e_{i_1}\otimes e_{i_2} \otimes \cdots \otimes e_{i_{d+1}},$ where $i_m \in \{1,2,\ldots,d\}$ for $m=1, \ldots, d+1$ and exactly two of the $i_m$ are the same. Indeed, if such vectors satisfy \eqref{eq:qdit-rel}, then the basis vectors with more recurring indices also satisfy \eqref{eq:qdit-rel} (introduce new indices for the recurring indices, apply the results for the basis vectors as above and then bring back the old indices). 
 Since \eqref{eq:qdit-rel} is invariant under permutation of indices, it is enough to prove that \ref{eq:qdit-rel} holds for basis vectors $v$ of the form 
 $$
 v = e_{1} \otimes e_2 \otimes e_3 \otimes \cdots \otimes e_d \otimes e_1.
$$	
Lemma \ref{lemma 1.1} shows that, when evaluated on such a basis vector, each term $s^\prime$ in the sum over $C_d$ cancels with a different term $s$ in the sum over $C_{d-1}$ such that none of the disjoint cycles of $s$ acts on a subset of the factors of $v$ with two equal factors. More precisely, if $s^\prime = (1\,2\,3\cdots d+1) \in C_d,$ the images of $v$ under the powers $s^\prime, (s^\prime)^2,\ldots,(s^\prime)^d$ are the $d$ permutations with corresponding cyclic structures $(1,d),(2,d-1),\ldots,(d,1)$ such that the two factors $e_1$ of $v$ are not in the same cycle.
Here $(i,j)$ with $i+j=d+1$ stands for a product of two disjoint cycles, one of length $i$ and one of length $j.$   By permuting the letters of $(s^\prime)^j,$
we obtain all the elements $s$ in the sum over $C_{d-1}$ 
corresponding to products with cyclic structure $(i,j)$ 
for any $i,j$ with $i+j=d+1$ 
that separate the two factors $e_1$ of $v.$

This means that by applying on $v$ all the terms $s^\prime$ in the sum over $C_d,$  we obtain the actions on $v$ of all the terms $s$  in the sum over $C_{d-1}$ that act on subsets of the factors of $v$ with no equal factors. 
 Hence, the remaining terms in the sum over $C_{d-1}$ are such that one of their disjoint cycles acts on a subset of the factors of $v$ with two equal factors. We then again apply Lemma \ref{lemma 1.1} and proceed inductively.
\end{proof}

\section{Identifying the qudit swap algebra $M^{\text{Sw}_d}_n(\mathbb{C})$ as a quotient of the free algebra}\label{s:quotalg}

Let $\CC\langle \text{swap}_{ij} \ | \ 1 \leq i < j \leq n\rangle$ be the free $*$-algebra on $\binom{n}{2}$ generators endowed with the involution $*$ that fixes each $\text{swap}_{ij}$ and acts as conjugation on $\CC.$ For $d\in\NN$ let $\mathcal{I}^{\text{Sw}_d}_n$ be the its ideal generated by
	\begin{equation}\label{eq:from sn}
	\begin{split}
		\text{swap}_{ij}^2 &= 1,\\
		\text{swap}_{ij}\text{swap}_{jk} &= \text{swap}_{ik}\text{swap}_{ij} = \text{swap}_{jk}\text{swap}_{ik},\\
        \text{swap}_{ij}\text{swap}_{kl} &= \text{swap}_{kl}\text{swap}_{ij},
	\end{split}
\end{equation}
	for all distinct indices $i,j,k,l,$
 and the relation in \eqref{eq:qdit-rel} with the swap matrices $\text{Swap}^{(d)}_{ij}$ replaced by the free variables $\text{swap}_{ij}.$
 We use the convention that whenever $i>j,$ then $\text{swap}_{ij}$ is interpreted as $\text{swap}_{ji}.$ Denote
	$$
	\mathcal{A}^{\text{Sw}_d}_n := \CC\langle \text{swap}_{ij} \ | \ 1 \leq i < j \leq n\rangle /\, \mathcal{I}^{\text{Sw}_d}_n
	$$
	and observe that there is a natural surjective $*$-homomorphism $\rho: \CC S_n \to \mathcal{A}^{\text{Sw}_d}_n$ defined by
	\begin{equation}\label{eq:eval}
		\rho((i,j)) = \text{swap}_{ij}+\mathcal{I}^{\text{Sw}_d}_n.
	\end{equation}
The algebras $\mathcal{A}^{\text{Sw}_d}_n$ and $M^{\text{Sw}_d}_n(\mathbb{C})$ are isomorphic.
{Indeed, this follows from \cite[Theorem 11.6.1]{Probook}: namely, $M^{\text{Sw}_d}_n(\mathbb{C})$ is isomorphic to $\CC[S_n]$ modulo the two-sided ideal generated by the antisymmetrizer on $d+1$ letters (see also \cite[Theorem 4.2]{dCP} or \cite[Theorem 2.8.1]{Gri}). We now present an elementary self-contained argument.
}

\begin{proposition}\label{prop: geq d}
	The generators of $\mathcal{I}^{\text{Sw}_d}_n$ do not all vanish under the irreps of $S_n$ corresponding to partitions 
	$\lambda$ of $n$ 
with $\het(\lambda)> d$.
\end{proposition}

\begin{proof}
	To show that the swap relations \eqref{eq:qdit-rel} are not satisfied by any irrep of $S_n$ corresponding to a partition with at least $d+1$ rows, consider an irrep corresponding to a partition $\lambda \vdash n$ of shape $(\lambda_1,\ldots,\lambda_k)$ with $k\geq d+1.$ We know that any irrep of $S_n$ is spanned by polytabloids 
	$$
	e_T = \sum_{\pi \in C_T} \text{sgn}(\pi) \pi\{T\},
	$$ 
	where $T$ ranges over all tabloids of shape $\lambda$, $C_T$ is the set of all permutations that permute the elements only within the columns of T and for each $\pi \in C_T,$  $\pi\{T\}$ is the tabloid obtained from $T$ by permuting the entries according to $\pi.$ 
	
	 Let $T$ be the standard Young tableaux of shape $\lambda$ and consider the action of the polynomial 
	$$
	g = (-1)^{d-1} \sum_{s \in C_d} s + (-1)^{d} \sum_{s \in C_{d-1}} s + (-1)^{d-1}\sum_{s \in C_{d-2}} s \pm \cdots - \sum_{s \in C_1} s + 1
	$$ 
	on the polytabloid $e_T$ via
	$$
	s_{ij}\, e_T = e_{(i,j)T}.
	$$
	Choose the indices $i_1, \ldots, i_{d+1}$ to be $1, \lambda_1 +1, \ldots ,\lambda_{d} +1$ respectively and note that the coefficient at $T$ in the resulting polytabloid is 
	$$
	(-1)^{d+1}\,|C_{d}| \cdot\text{sgn}(s \in C_{d}) +(-1)^d \,|C_{d-1}| \cdot\text{sgn}(s \in C_{d-1})+    \cdots+|C_2|\cdot 1-|C_1|\cdot(-1)+1.
	$$
	But the latter is strictly positive since, for $k=1,\ldots,d,$ the sign of the elements in $C_k$ is $(-1)^k.$ This shows that the polynomial $g$ does not vanish under the evaluation $s_{ij} = \rho_\lambda(i\,j)$ for the chosen $\lambda.$ Thus
	the swap relations \eqref{eq:qdit-rel} are incompatible with any Young Tableaux with more than $d$ rows. 
\end{proof}

	\begin{proposition} \label{prop: leq d}
		 All the irreps of $S_n$ corresponding to a partition of $n$ with at most $d$ rows in its Young tableaux satisfy \eqref{eq:qdit-rel}.
	\end{proposition}

 \begin{proof} For any irrep of $S_n$ corresponding to a partition $\lambda \vdash n$ with at most $d$ rows 
      it is enough to prove that \eqref{eq:qdit-rel} holds when evaluated at basis vectors, i.e., polytabloids
     \begin{equation}\label{eq:polyt}
         e_T = \sum_{\pi \in C_T} \text{sgn}(\pi) \pi\{T\},
     \end{equation}
 where $T$ is any tabloid of shape $\lambda.$  We use the canonical identification between tabloids $T$ and rank one vectors $v \in (\CC^d)^{\otimes n}$ of the form $v= e_{i_1}\otimes \cdots \otimes e_{i_k}$ with $i_j \in \{1,\ldots, d\}.$  Suppose $\lambda = (\lambda_1,\ldots,\lambda_k)$ with $\lambda_i \geq \lambda_{i+1}$ and let $T$ be a tabloid of shape $\lambda.$ Define $v$ to be the vector 
 with tensor factors $e_k$ at positions, which are the numbers appearing in the $k$th row of $T.$ Now permuting the tensor factors of $v$ is the same as permuting the entries of $T.$ The inverse of this procedure assigns to a rank one vector $v$ the tabloid $T$ with $k$th row consisting of the numbers that index the positions of tensor factors $e_k$ in $v.$
 Now the  proof of Equation \eqref{eq:qdit-rel} in Section \ref{sec:qudit-rel} implies that Equation \eqref{eq:qdit-rel} holds when evaluated at each summand in \eqref{eq:polyt}, from which the claim follows. 
 \end{proof}

\begin{theorem}\label{th:iso}
    The algebras $\mathcal{A}^{\text{Sw}_d}_n$ and $M^{\text{Sw}_d}_n(\mathbb{C})$ are isomorphic.
\end{theorem}

\begin{proof} The algebras $M^{\text{Sw}_d}_n(\mathbb{C})$ and $\mathcal{A}^{\text{Sw}_d}_n$ are both homomorphic images of the semisimple finite-dimensional algebra $\CC[S_n]$. Therefore $M^{\text{Sw}_d}_n(\mathbb{C})$ and $\mathcal{A}^{\text{Sw}_d}_n$ and semisimple and finite-dimensional as well.
To show that they are isomorphic, we prove that $\mathcal{A}^{\text{Sw}_d}_n$ and $M^{\text{Sw}_d}_n(\mathbb{C})$ have the same block decomposition into simple matrix algebras.

The block decomposition of $M^{\text{Sw}_d}_n(\mathbb{C})$ is described in Theorem \ref{th: s-w-swaps}. Recall that $
	\mathcal{A}^{\text{Sw}_d}_n = \CC\langle \text{swap}_{ij} \ | \ 1 \leq i < j \leq n\rangle /\, \mathcal{I}^{\text{Sw}_d}_n,
	$
 where  $\mathcal{I}^{\text{Sw}_d}_n$ is the ideal generated by the relations \eqref{eq:from sn} defining the symmetric group $S_n$ and the degree-reducing relation \eqref{eq:qdit-rel}, which holds precisely on the irreps of $S_n$ indexed by partitions with at most $d$ rows (see Proposition \ref{prop: geq d} and Proposition \ref{prop: leq d}). It is now immediate that the algebras $\mathcal{A}^{\text{Sw}_d}_n$ and $M^{\text{Sw}_d}_n(\mathbb{C})$ have the same semisimple decomposition.
\end{proof}

\section{NPO hierarchy}\label{s:npo}

The identification of the swap algebra $M^{\text{Sw}_d}_n(\mathbb{C})$ as a quotient of the free algebra in Section \ref{s:quotalg} allows one to view the $d$-QMC as an example of a noncommutative polynomial optimization (NPO) problem.

Let $\cF_n=\CC\langle \text{swap}_{ij} \ | \ 1 \leq i < j \leq n\rangle$ be the $*$-free algebra on $\binom{n}{2}$ generators, and $V_\ell=\{s\in\cF_n\colon \deg s\le \ell\}$ its subspace spanned by the products of at most $\ell$ \text{swap} symbols.
Recall the isomorphism between $M^{\text{Sw}_d}_n(\CC)$ and $\cA^{\text{Sw}_d}_n=\cF_n /\, \mathcal{I}^{\text{Sw}_d}_n$ from Theorem \ref{th:iso}.
We can view the Hamiltonian $H_G^{d}$ from \eqref{eq:qmcd} as an element of $\cA^{\text{Sw}_d}_n$, and let
$$h_G = \sum_{(i,j) \in \EE(G)} 2 w_{ij} \left(1 - \textnormal{swap}_{ij}\right)$$
be the corresponding element in $\cF_n$.
Since the $*$-algebra $M^{\text{Sw}_d}_n(\CC)$ is finite-dimensional, it is a C*-algebra. Therefore the largest eigenvalue of $H_G^{d}$ equals
\begin{align*}
\alpha_*=&\min\left\{
\alpha:\alpha-H_G^{d}=a^*a \text{ for some } a\in \cA^{\text{Sw}_d}_n
\right\} \\
=&\min\left\{
\alpha:\alpha-h_G=\sum_k s_k^*s_k+q  
\text{ for some } 
s_k\in\cF_n,\ q\in \mathcal{I}^{\text{Sw}_d}_n
\right\}.
\end{align*}
For $\ell=1,\dots,n$ define two sequences,
\begin{equation}\label{e:relax2}
\alpha'_\ell
=
\min\left\{
\alpha:\alpha-h_G={\bf u}_\ell^*\left(A+\sum_m g_m A_m\right){\bf u}_\ell \text{ for some }A_m=A_m^\intercal,\text{ and }A\succeq0 
\right\}
\end{equation}
and
\begin{equation}\label{e:relax}
\begin{split}
\alpha_\ell
&=
\min\left\{
\alpha:\alpha-h_G=\sum_k s_k^*s_k+q  
\text{ for some } 
s_k\in V_\ell,\ q\in \mathcal{I}^{\text{Sw}_d}_n
\right\}\\
&=
\min\left\{
\alpha:\alpha-h_G\equiv{\bf u}_\ell^*A{\bf u}_\ell \mod\mathcal{I}^{\text{Sw}_d}_n \text{ for some } A\succeq0 
\right\},
\end{split}
\end{equation}
where ${\bf u}_\ell$ is a column of products of at most $\ell$ \text{swap} symbols, and the $g_m$ are the generators of the ideal $\mathcal{I}^{\text{Sw}_d}_n$ as in Section \ref{s:quotalg}.
Then $\alpha_\ell\le \alpha'_\ell$ for every $\ell$, the sequences $\{\alpha_\ell\}_\ell$ and $\{\alpha^\prime_\ell\}_\ell$ are decreasing, and $\alpha_{n-1}=\alpha'_{n-1}=\alpha_*$. The last equality holds since every permutation in $S_n$ is a product of at most $n-1$ transpositions, and $\cA^{\text{Sw}_d}_n$ is a quotient of $\CC[S_n]$.

Clearly, \eqref{e:relax2} is a semidefinite program (SDP). 
The second line in \eqref{e:relax} is likewise an SDP once the calculation modulo $\mathcal{I}^{\text{Sw}_d}_n$ is resolved (see the paragraph below). We refer to them as the $\ell$\textsuperscript{th} relaxations of the $d$-QMC. 
Thus we obtained hierarchies of SDPs whose solutions converge to the solution of the $d$-QMC from above. The hierarchy associated with $\alpha'_\ell$ is a very special case of the analog of the Lasserre hierarchy \cite{Las} for NPO that is based on a noncommutative Positivstellensatz \cite{HM}, and whose dual is the Navascu\'es-Pironio-Ac\'in hierarchy \cite{NPA08,NPA} in quantum physics.

While the expression \eqref{e:relax2} is readily an SDP, it involves more unknowns than \eqref{e:relax} (i.e., in addition to unknowns $\alpha$ and $A\succeq0$, it also involves several unknown symmetric $A_m$). Thus, it is preferable to work with \eqref{e:relax}. To prepare the linear constraints in the SDP \eqref{e:relax} that arise from $\alpha-h_G\equiv{\bf u}_\ell^*A{\bf u}_\ell \mod\mathcal{I}^{\text{Sw}_d}_n$ (note that the right-hand side involves products of at most $2\ell$ swap symbols), one needs to identify a subset $\cB^{d}_{2\ell}$ in $V_{2\ell}$ that maps to a basis under the quotient map $q:V_{2\ell}\to(V_{2\ell}+\mathcal{I}^{\text{Sw}_d}_n)/\mathcal{I}^{\text{Sw}_d}_n$. 
To do this, one can start with a basis of $V_{2\ell}$, reduce it modulo $\mathcal{I}^{\text{Sw}_d}_n$ via a noncommutative Gr\"obner basis algorithm \cite{Mor86}, and then identify a basis $\cB^{d}_{2\ell}$ in the resulting set.
Alternatively, one can obtain a concrete instance of $\cB^{d}_{2\ell}$ as follows.
In \cite{Pro21}, a permutation $\pi\in S_n$ is called {\bf $(d+1)$-good} if there is no increasing sequence $j_0<\cdots<j_d$ such that $\pi(j_0)>\cdots>\pi(j_d)$. By \cite[Theorem 8]{Pro21}, $(d+1)$-good permutations form a basis of $M^{\text{Sw}_d}_n(\CC)$. For $\cB^{d}_{2\ell}$ one can thus choose the set of all $(d+1)$-good permutations that are products of at most $2\ell$ transpositions.

\begin{remark}
For reference, the ``naive'' SDP for $\alpha^\ast=\mathrm{eigmax}(H_G^{(d)})$ can be written as
$\alpha_*=\max\{\Tr(H_G^{(d)}\rho): \rho\succeq 0,\ \Tr(\rho)=1\}$, where the matrix variable
$\rho$ is a $d^n\times d^n$ matrix (so the semidefinite constraint has size $d^n$ and the number
of scalar variables is $\Theta(d^{2n})$).
In contrast, the $\ell$th relaxation \eqref{e:relax} uses a Gram matrix $A\succeq 0$ indexed by
a basis of $(V_\ell+I_n^{\text{Sw}_d})/I_n^{\text{Sw}_d}$; denote its size by
$N_\ell:=\dim\big((V_\ell+I_n^{\text{Sw}_d})/I_n^{\text{Sw}_d}\big)$.
A crude upper bound is obtained by counting words of length $\le \ell$ in transpositions:
$N_\ell\le 1+\sum_{t=1}^{\ell}\binom{n}{2}^t = O(n^{2\ell})$ (for fixed $\ell$), hence the semidefinite
constraint size grows only polynomially in $n$ at any fixed level.
Consequently, the Gram-matrix dimension overtakes the naive one only at intermediate levels,
roughly when $N_\ell\gtrsim d^n$, i.e., $\ell \gtrsim \frac{n\log d}{2\log n}$.
Further, the hierarchy is exact 
 at the terminal level $\ell=n-1$ and
$N_{n-1}=\dim M_n^{\text{Sw}_d}(\mathbb{C})\le n!$.
\end{remark}

When $\ell$ is large, the size of the SDP for $\alpha_\ell$ is typically too large for available SDP solvers. In practice, one thus often has to settle for computing only the first two relaxations of $\alpha_*$, namely $\alpha_1$ and $\alpha_2$.
To solve these two SDPs, the sets $\cB^{d}_2$ and $\cB^{d}_4$ are required in view of the preceding paragraph. For $d=2$, these are given in \cite[Subsection 4.3.2 and Appendix B.2]{BCEHK24}.
For $d\ge 5$, one can take $\cB^{d}_4$ (resp. $\cB^{d}_2$) consisting of all permutations that are products of at most 4 (resp. 2) transpositions; see Appendix \ref{a:d-1}.
For $d\in\{3,4\}$, the bases are presented in Appendices \ref{a:d3l2}, \ref{a:d3l4} and \ref{a:d4l4}.

\begin{example}
We computed the first two relaxations of \eqref{e:relax} in the case $d=3$ for all 853 connected graphs on $n=7$ vertices, with {unit edge weights} $w_{ij}=1$.
The list of graphs was generated using Nauty \cite{nauty}. 
To construct the explicit SDP instances 
of \eqref{e:relax}
for $\ell=1,2$ we reduced products of swap symbols
modulo $I_n^{\text{Sw}_d}$ via noncommutative Gr\"obner bases computed in
{Magma} \cite{magma} (alternatively, one can use the bases described in
Appendices \ref{a:d3l2} and \ref{a:d3l4}). 
The produced SDPs were solved 
on a Macbook Air laptop
using Mathematica\footnote{\url{https://www.wolfram.com/mathematica}}.
To verify numerical exactness of the second relaxation, for each graph $G$ we also computed
$\mathrm{eigmax}(H_G^{(3)})$ by {direct diagonalization} of the $3^7\times 3^7$ matrix
representation of $H_G^{(3)}$ and compared it to $\alpha_2$; the maximal absolute discrepancy
over all $853$ graphs was at most $10^{-8}$. 
On the other hand, the first relaxation performed very poorly. The reason is that in low degrees (so degree $\leq2$ when working with the first relaxation)
the nontrivial relation
\eqref{eq:qdit-rel} defining the $3$-swap algebra
does not enter computations, so one is effectively optimizing over the
corresponding group-algebra quotient; cf. Subsection \ref{ssec:larged}. 
Thus the $3$-QMC provides a large class of examples where the 
second NPO relaxation clearly outperforms the first one.
\end{example}

\section{Quantum Max $d$-Cut and irreps}\label{sec:6}

The decomposition of the $d$-swap algebra $M^{\text{\textnormal{Sw}}_d}_n(\mathbb{C})$ described in Section \ref{ssec:sw} is a valuable tool for calculating the eigenvalues of the qudit Quantum Max Cut Hamiltonian of a complete graph on $n$ vertices.
Recall from \eqref{eq:HamSw} that given a graph $G,$ 
the $d$-QMC irrep Hamiltonian $H_G^{d}$ is defined as 
	\begin{align*}
		H_G^{d} =   \rho_n^{(d)} \left(
		\sum_{(i,j) \in \text{E}(G)} 
		2w_{ij} \,\big(\id - (i\ j)\big)\right) =
		\sum_{(i,j) \in \text{E}(G)} 
		2w_{ij} \,\big(I - \text{Swap}^{(d)}_{ij}\big).
	\end{align*}
Here the $\text{Swap}^{(d)}_{ij}$ denote the qudit swap matrices in $M^{\text{Sw}_d}_n(\mathbb{C}).$

\begin{definition}
	Let $G$ be a graph on $n$ vertices with edge set E$(G)$ and edge weights $w_{ij}.$ \looseness=-1
Let $\lambda \vdash n$ be a partition labelling an irrep of $S_n.$ The {\bf QMC irrep Hamiltonian} $H_G^{\lambda}$ is defined as
\begin{align*}
	H_G^{\lambda} =  \rho_\lambda \left(
	\sum_{(i,j) \in \text{E}(G)} 
	2w_{ij} \,\big(\id - (i\ j)\big)\right).
\end{align*}
\end{definition}

The following is a straightforward corollary of Theorem \ref{th: s-w-swaps}. 

\begin{corollary}
	The spectrum of the $d$-\text{\textnormal{QMC}} Hamiltonian of a graph $G$ is the union of the spectra of all the Hamiltonians corresponding to the irreps of $S_n$ with at most $d$ rows. That is,
	$$
	\text{\textnormal{eigs}}(H_G^d) = \bigcup_{\substack{\lambda\vdash n \\[.5mm] \het(\lambda)\le d}} \text{\textnormal{eigs}}(H_G^\lambda),
	$$
	and, in particular,
	$$
	\text{\textnormal{eig}}_{\max}(H_G^d) = \max_{\substack{\lambda\vdash n \\[.5mm] \het(\lambda)\le d}} \left( \text{\textnormal{eig}}_{\max}(H_G^\lambda)  \right).
	$$
\end{corollary}

\subsection{Exact solution for sufficiently large $d$}\label{ssec:larged}

We record the largest eigenvalue of the Hamiltonian
$$H_G^d = \sum_{(i,j) \in \EE(G)} 2 w_{ij} \left(I - \textnormal{Swap}^{(d)}_{ij}\right)
$$
if $d\ge n=|\VV(G)|$ and $w_{ij}\ge0$ for all $(i,j)\in \EE(G)$. 

\begin{proposition}\label{prop:d>n}
If all the edge weights in $G$ are nonnegative and $d\ge n$, the largest eigenvalue of $H_G^d$ is $4\sum_{i,j} w_{ij}$.
\end{proposition}

\begin{proof}
Clearly,
$$
\|H_G^d\| \le \sum_{(i,j) \in \EE(G)} 2 w_{ij} \left\|I - \textnormal{Swap}^{(d)}_{ij}\right\|
=4\sum_{i,j} w_{ij},
$$
so the largest eigenvalue of $H_G^d$ is at most $4\sum_{i,j} w_{ij}$. If $d\ge n$, then
$$v=
\sum_{\pi\in S_n} {\rm sgn}(\pi) e_{\pi(1)}\otimes\cdots\otimes e_{\pi(n)}
$$
satisfies $\textnormal{Swap}^{(d)}_{ij}v=-v$ for all $i\neq j$. Therefore 
$H_G^dv=\left(4\sum_{i,j} w_{ij}\right)v$.
\end{proof}

Let us end this short subsection with a comment on the case $d=n-1$. While $M^{\text{Sw}_n}_n(\mathbb{C})\cong \CC[S_n]$, the swap algebra $M^{\text{Sw}_{n-1}}_n(\mathbb{C})$ is isomorphic to the direct sum of all the irreps of $S_n$ apart from the one-dimensional sign representation of $S_n$. The latter is, as a sub-representation of $\CC[S_n]$, spanned by $a=\frac{1}{n!}\sum_{\pi \in S_n}\operatorname{sgn}(\pi)\pi$, the antisymmetrizer in $\CC[S_n]$.
Thus $M^{\text{Sw}_{n-1}}_n(\mathbb{C})$ is, as a C*-algebra, isomorphic to the orthogonal complement of $a$ in $\CC[S_n]$. Under this identification, the Hamiltonian $H_G^{n-1}$ corresponds to
\begin{equation}\label{e:n-1}
\sum_{(i,j) \in \EE(G)} 2 w_{ij} \big(\id - (i\ j)\big) -4\left(\sum_{(i,j) \in \EE(G)} w_{ij}\right)a\in \CC[S_n]
\end{equation}
because the projection of $\id - (i\ j)$ onto the span of $a$ equals $2a$. While \eqref{e:n-1} lacks the sparsity (2-locality) of $H_G^{n-1}$, it can at least be viewed as an operator on a slightly smaller space of dimension $n!<(n-1)^n$ via the left regular representation of $S_n$.
We speculate that $M^{\text{Sw}_{n-1}}_n(\mathbb{C})$ differing from $\CC[S_n]$ only for the (very simple) sign representation might offer further insight into the $(n-1)$-QMC problem, which is currently beyond reach.

\subsection{Exact solutions for clique Hamiltonians with uniform edge weights}\label{subsec ev-3}
We now present the main steps in the computation of the spectrum of the $d$-QMC Hamiltonian of a complete graph  with uniform edge weights. For the rest of this section we assume all edge weights $w_{ij}=1$. 

The clique is the easiest graph for tackling the $d$-QMC problem since 
the isotypic components of the $d$-QMC Hamiltonian are scalar matrices in this case.

\begin{lemma}\label{lemma scalar}
	Let $\lambda \vdash n$ be a partition. Then
	\begin{align}\label{eq: eta}
			H_{K_n}^\lambda = \eta_\lambda I,
	\end{align}
	where $\eta_\lambda$ is a scalar depending only on the irrep $\lambda$ and $I$ is the identity matrix of the appropriate dimension.
\end{lemma}

\begin{proof}
	Follows by \cite[Lemma 2.11]{BCEHK24}.
\end{proof}

For any partition $\lambda \vdash n,$
the dimension of the irrep $\rho_\lambda$ of $S_n$ is the value of the corresponding character $\chi_\lambda : S_n \to \mathbb{C}$ at the identity element $e \in S_n.$ So
$$
\chi_\lambda (\pi) = \text{Tr}(\rho_\lambda(\pi)), \quad \quad \pi \in S_n,
$$
and, in particular,
$$
\chi_\lambda (e) = \text{Tr}(\rho_\lambda(e))
$$
is the dimension of the irrep $\rho_\lambda$ of $S_n.$ 
From Lemma \ref{lemma scalar} if follows that the eigenvalue $\eta_\lambda$ can be expressed through the values of the character $\chi_\lambda$ at the identity $e$ and at any transposition $(i\,j).$

\begin{lemma}\label{l:char}
For any $\lambda \vdash n$ let $\chi_\lambda$ be the character corresponding to $\rho_\lambda$ and let $\eta_\lambda$ be as in Lemma \ref{lemma scalar}. Then
    \begin{align}\label{eq:etafromchi}
\eta_\lambda = 2 \binom{n}{2} \bigg(1-\frac{\chi_\lambda\big((i\ j)\big)}{\chi_\lambda(e)}\bigg).	
\end{align}
\end{lemma}

\begin{proof}
For $\lambda \vdash n,$ the constant $\eta_\lambda$ can be explicitly computed by taking the trace on both sides of \eqref{eq: eta}. Indeed, since 
$$
H_{K_n}^{\lambda} =  \rho_\lambda \left(
\sum_{(i,j) \in \text{E}(G)} 
2\,\big(I - (i\ j)\big)\right),
$$
we get, by taking the trace, that
$$
\text{Tr} \big[ H_{K_n}^{\lambda}\big] = \sum_{(i,j) \in \text{E}(G)} 
2 \,\big[\chi_\lambda(e) - \chi_\lambda\big((i\ j)\big)\big] = 2\binom{n}{2}\big[\chi_\lambda(e) - \chi_\lambda\big((i\ j)\big)\big].
$$
On the other hand,
$$
\text{Tr} \big[ H_{K_n}^{\lambda}\big] = \eta_\lambda\, \chi_\lambda(e),
$$
so that
\begin{align*}
\eta_\lambda & = 2 \binom{n}{2} \bigg(1-\frac{\chi_\lambda\big((i\ j)\big)}{\chi_\lambda(e)}\bigg).\tag*{\qed}
\end{align*}
\renewcommand{\qedsymbol}{}
\end{proof}

\begin{example}\label{rem:eta2}
	For a two-row partition $\lambda = (n-k,k),$ it was computed in \cite[Lemma 2.12]{BCEHK24} that
	$$
	\eta_\lambda = 2k(n+1)-2k^2.
	$$
\end{example}

We now compute the eigenvalue $\eta_\lambda$ for any partition $\lambda$ using a formula by Frobenius \cite{Fro01}. For a more direct approach, where we explicitly compute the value of $\chi_\lambda$ at a transposition using the well-known hook-length formula, see Appendix \ref{appC}.

\begin{proposition} \label{prop:etaD}
Let $\eta_\lambda$ be as in Lemma \ref{lemma scalar}. For any $\lambda \vdash n$ with rows $\lambda_1\ge\cdots\ge \lambda_d$,
\begin{equation}\label{eq:etaD}
\eta_\lambda=
n^2 +\frac{d(d-1)(2d-1)}{6}
-\sum_{k=1}^d\big( \lambda_k - (k-1)\big)^2.
\end{equation}
\end{proposition}

\begin{proof}
Let $\lambda \vdash n$ be a partition with rows $\lambda_1\ge\cdots\ge \lambda_d\geq1.$
    Recall that the conjugate partition $\lambda^\prime$ of $\lambda$ is the partition of $n,$ whose $k$th row is the $k$th column of $\lambda.$ It follows from \cite[p.~534]{Fro01} (or \cite[Theorem 4]{Lassalle}) that for any transposition $(i\,j),$
    $$
    \binom{n}{2} \frac{\chi_\lambda\big((i\ j)\big)}{\chi_\lambda(e)} = 
    \sum_{k=1}^d \left[ \binom{\lambda_k}{2} - \binom{\lambda^\prime_k}{2}    \right].
    $$
    Moreover, by \cite[Proposition 1.8.3]{Sta99} we have
    $$
   \sum_{k=1}^d  \binom{\lambda^\prime_k}{2} = \sum_{i=1}^d (k-1)\lambda_k.
    $$
    Hence,
    \begin{align*}
        \eta_\lambda & = 2 \binom{n}{2} \bigg(1-\frac{\chi_\lambda\big((i\ j)\big)}{\chi_\lambda(e)}\bigg)\\
        & = 2 \binom{n}{2} - 2 \sum_{k=1}^d \left[ \binom{\lambda_k}{2} - \binom{\lambda^\prime_k}{2}    \right]\\
        & = 2 \binom{n}{2} - 2 \sum_{k=1}^d \left[ \binom{\lambda_k}{2} -  (k-1)\lambda_k  \right]\\
        & = n(n-1) + \sum_{k=1}^d \lambda_k - \sum_{k=1}^d \left(\lambda_k^2 - 2(k-1)\lambda_k\right) \\
        & =  n^2 - \sum_{k=1}^d \left(\lambda_k^2 - 2(k-1)\lambda_k\right)\\
        & = n^2 + \frac{d(d-1)(2d-1)}{6} - \sum_{k=1}^d \left(\lambda_k - (k-1)\right)^2.\tag*{\qed}
    \end{align*}
    \renewcommand{\qedsymbol}{}
\end{proof}

Since $\frac{d(d-1)(2d-1)}{6}=1^2+\cdots+(d-1)^2$, the formula \eqref{eq:etaD} is valid even for $\lambda\vdash n$ with $\het(\lambda)\le d$, i.e., $\lambda=(\lambda_1,\dots,\lambda_d)$ with $\lambda_1\ge\cdots\ge \lambda_d\ge0$.

Using Proposition \ref{prop:etaD}, one can deduce the solution to the $d$-QMC problem for a clique, i.e., the maximal $\eta_\lambda,$  where $\lambda \vdash n$ ranges over all partitions with at most $d$ rows. 
Moreover, the form \eqref{eq:etaD} of $\eta_\lambda$ eases the computation of the precise partition $\lambda \vdash n$ at which the maximum is obtained.

\begin{corollary}\label{prop:69++}
	The maximum value of $\eta_\lambda$ among all partitions $\lambda\vdash n$ with $\het(\lambda)\leq d$ is obtained at
	\begin{equation}\label{eq:answ}
	\lambda=\Big( 
\underbrace{1+\frac{n-r}d, \ldots, 1+\frac{n-r}d}_{r},\ 
\underbrace{\frac{n-r}d, \ldots, \frac{n-r}d}_{d-r}
	\Big)
	\end{equation}
	for $n\equiv r\mod d.$ Moreover, the solution to the $d$-QMC problem for an $n$-clique is \begin{equation}\label{eq:etamax}
	     n^2 + (d-1) n + r^2 - r(d+1) - \frac{n^2-r^2}{d}.
	\end{equation}
\end{corollary}

\begin{proof}
The statement for $d=n$ is routine (or see Proposition \ref{prop:d>n}). We thus assume $d<n$.

First, for each partition $\lambda\vdash n$ with
$e=\het(\lambda)<d$ we find a partition $\tilde\lambda\vdash n$ with 
$\het(\tilde\lambda)=e+1\leq d$ such that $\eta_{\tilde\lambda}>\eta_\lambda$. Let $e'$ be the largest index for which $\lambda_{e'}>1$ 
($e'$ exists since $d<n$). Then construct $\tilde\lambda\vdash n$ with $\het(\tilde\lambda)=e+1$ as follows:
\[
\tilde\lambda_j=
\begin{cases}
\lambda_j & 1\leq j\leq e, \; j\neq e'\\
\lambda_j-1 & j=e'\\
1 & j=e+1.
\end{cases}
\]
Now 
\[
\begin{split}
\eta_{\lambda}-\eta_{\tilde\lambda} & =
n^2 +\frac{e(e-1)(2e-1)}{6}
-\sum_{k=1}^e\big( \lambda_k - (k-1)\big)^2
 \\
 &\phantom{{}={}} -
\Big(
n^2 +\frac{(e+1)e(2e+1)}{6}
-\sum_{k=1}^{e+1}\big( \tilde\lambda_k - (k-1)\big)^2
\Big)
\\
&= 
-e^2 - (\la_{e'}-(e'-1))^2 + (\la_{e'}-1-(e'-1))^2 + (1-e)^2 
\\
&= 
-2(\la_{e'}+e-e')<0,
\end{split}
\]
as desired.

Thus the solution to the $d$-QMC problem is attained at $\lambda\vdash n$ with $\het(\lambda)=d$.
Since $n,d$ are fixed, maximizing $\eta_\lambda$ is by Proposition \ref{prop:etaD} equivalent to minimizing
\begin{equation}\label{eq:businesseta}
f(\lambda):=\sum_{k=1}^d\big( \lambda_k - (k-1)\big)^2
\end{equation}
over partitions $\lambda\vdash n$ with $\het(\lambda)=d$. That is,
$\sum_{k=1}^d \lambda_k=n$ and 
$\lambda_1\geq\lambda_2\geq\cdots\geq\lambda_d\geq1$.

We claim that any minimizer $\lambda^\star$ of \eqref{eq:businesseta} has at most one jump, i.e., $\lambda^\star_1-\lambda^\star_d\leq1.$ Assume otherwise. Then there are $d\geq k > \ell \geq2$ such that
\[
\lambda^\star_d \leq \cdots \leq \lambda^\star_k < \lambda^\star_{k-1} \leq\cdots\leq\lambda^\star_\ell<\lambda^\star_{\ell-1}\leq\cdots \leq \lambda_1.
\]
We now replace $\lambda^*_k$ with $\lambda^*_k+1$ 
and $\lambda^*_{\ell-1}$ with $\lambda^*_{\ell-1}-1$ to obtain a new partition $\lambda^\dagger\vdash n$ with $\het(\lambda^\dagger)\leq d$. Then
\[
\begin{split}
f(\lambda^*)-f(\lambda^\dagger)
 & = 
\big( \lambda_k - (k-1)\big)^2 + 
\big( \lambda_{\ell-1} - (\ell-2)\big)^2  \\
&\phantom{{}={}}
-
\big( \lambda_k+1 - (k-1)\big)^2 
-
\big( \lambda_{\ell-1}-1 - (\ell-2)\big)^2 
\\
&=
2 (k-\ell) + 2 (\lambda_{\ell-1}-\lambda_k) >0,
 \end{split}
\]
contradicting the minimality of $\lambda^\star$. 
Since a minimizer $\lambda$ of \eqref{eq:businesseta} satisfies $\lambda_1-\lambda_d\leq1,$
\eqref{eq:answ} follows.

It is routine to check that the solution to the $d$-QMC problem
for an $n$-clique, obtained by plugging \eqref{eq:answ} into the formula \eqref{eq:etaD} for $\eta_\lambda,$ is in fact \eqref{eq:etamax}.
\end{proof}

\begin{remark}
    The solution \eqref{eq:etamax} to the $d$-QMC problem  for an $n$-clique is indeed an integer, since
    $$
    \frac{n^2-r^2}{d} = \frac{(n-r)(n+r)}{d}
    $$
    and $n-r$ is divisible by $d$ (as $n\equiv r\mod d$).
\end{remark}

\section{Graph clique decomposition}
\label{sec:graph_clique_decomposition}
In this section we refine an algorithm from \cite[Section 6]{BCEHK24}, called \textit{graph clique decomposition}, to solve the $d$-QMC problem for a larger family of graphs, namely star graphs and a large class of complete bipartite graphs. 
The clique decomposition expresses the $d$-QMC Hamiltonian of a given graph as an alternating sum of $d$-QMC Hamiltonians associated with cliques and simple graphs, in a form suitable for eigenvalue analysis.

\subsection{Exact solutions for star graphs}\label{ss:star}
Let $n\ge 2$, and consider the star graph $\starn$ on $n$-vertices and observe that if we label the vertices of $\starn$ so that $n$ corresponds to the central vertex, then
\begin{equation}\label{eq:star}
    \starn = K_n - K_{n-1}.
\end{equation}
E.g., for $n=8$ we have
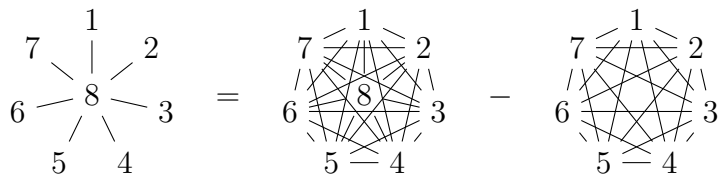
\[
\tikz 
\graph { subgraph I_n [n=7, clockwise] -- 8};
\quad \raisebox{2.5em}{$=$}  \quad
 \tikz \graph { subgraph K_n [n=7, clockwise] -- 8};
\quad \raisebox{2.5em}{$-$} \quad
 \tikz  \graph { subgraph K_n [n=7, clockwise] };
\]
Here, we view $K_{n-1}$ as a graph on $n$ vertices, in which the vertex $n$ is disconnected from the rest.
This is the clique decomposition of $\starn,$ which together with the Young branching rule \cite[§\,2]{Sag01} facilitates the computation of the eigenvalues of $H_{\starn}^{d}$ significantly. The spectrum of  $H_{\starn}^{d}$ in the case $d=2$ was computed in \cite{BCEHK24}.
Note that $H_\starn^{(n)}=0$ for the 1-row partition $(n)\vdash n$.

\begin{example}[{\cite[Lemma 6.1]{BCEHK24}}]\label{p:star2}
Let $n\ge 2$ and $\lambda=(\lambda_1,\lambda_2)$. If $\lambda_1>\lambda_2$ then $H_\starn^\lambda$ has two eigenvalues
$$e_1 = 2(n-\lambda_1), \quad \quad e_2 = 2(n-\lambda_2+1).
$$
If $\lambda_1=\lambda_2$ then $H_\starn^\lambda$ has only one eigenvalue $e_1 =2(n-\lambda_2+1)=n+2$.
The solution to the $2$-QMC problem for $\starn$ is $2n$, attained at the partition $\lambda=(n-1,1)$.
\end{example}

We extend this result by computing the eigenvalues of  $H_\starn^{d}.$
\begin{theorem}\label{th:Eigstarn}
    If $\lambda = (\lambda_1,\ldots, \lambda_e) \vdash n$ has $e \leq d$ rows $\lambda_1 \geq \cdots \geq \lambda_e \geq 1$, then the eigenvalues of the $d$-QMC irrep Hamiltonian $H_\starn^\lambda$  form a subset of 
$$
\{2(n-\lambda_1), \,2(n-\lambda_2+1),\, \ldots,\, 2(n-\lambda_{e}+e-1)\}
$$
containing the value $\eta_\star =2(n-\lambda_{e}+e-1).$ Hence, the solution to the $d$-QMC problem for $\starn$ {and $2\le d\le n$} is $2(n+d-2),$ attained at any partition with $\lambda_d=1.$
\end{theorem}

\begin{proof}
By Lemma \ref{lemma scalar}, $H_{K_n}^\lambda$ is a scalar matrix for every partition $\lambda$.
From \eqref{eq:star} we deduce 
that $H_\starn^\lambda$ is similar to $H_{K_n}^\lambda - H_{K_{n-1}}^\lambda {\otimes I_d.}$
The eigenvalues of $H_{K_{n-1}}^\lambda$ can be computed using the Young branching rule \cite[§\,2]{Sag01}. It states that the restriction of any irrep, say labeled by the partition $\lambda$, of $S_n$ to the subgroup $S_{n-1}$ decomposes as a direct sum of all the irreps of $S_{n-1}$ which can be obtained from $\lambda$ by removing one box.

Hence, the eigenvalues of $H_\starn^\lambda$ are obtained by subtracting from $\eta_\lambda$ (which is the single eigenvalue of $H_{K_n}^\lambda$) each of the (at most $e$) eigenvalues of $H_{K_{n-1}}^\lambda.$ 
Precisely, we use the following procedure: Let the index $j$ run from $e$ to $1$ and let 
$\eta_\lambda$ be the eigenvalue of $H_{K_n}^\lambda$ and $\eta_{\mu_j}$ the eigenvalue of $H_{K_{n-1}}^{\mu_j},$ where $\mu_j$ is obtained from $\lambda$ by removing one box from the $j$th row. Now let
\begin{align*}
\eta_*(j) = \eta_\lambda - \eta_{\mu_j} =
     2(n-\lambda_j+j-1).
\end{align*}
As noted, $\eta_* = \eta_*(e)$ is an eigenvalue of $H_\starn^\lambda$ if $\lambda_e \geq 1,$ since in this case one can remove a box from the last row of the Young diagram of $\lambda$ to obtain a valid Young diagram of a partition of $n-1.$ Next, consider $j=e-1.$ If $\lambda_{e-1} > \lambda_e$ (hence, $\lambda_{e-1}\geq 2$), then $\eta_*(e-1)$ is an eigenvalue of $H_\starn^\lambda,$ because one box can be removed from $\lambda_{e-1}$ to obtain a valid partition of $n-1;$ 
otherwise, if $\lambda_{e-1} = \lambda_e,$ proceed to $j=e-2$ and so on.

It is now immediate that the largest eigenvalue of $H_\starn^{d}$ (for $d\le n$) is $2(n+d-2),$ which is obtained by plugging $j=d$ and $\lambda_d=1$ into the expression for $\eta_\star.$
\end{proof}

To give a more precise description of the spectrum of $H_\starn^\lambda$ for $\lambda$ with $\het(\lambda) \le d$, it is easier to look at that of $nI-\frac12 H_\starn^\lambda$: its eigenvalues are obtained from the strictly decreasing sequence
$$\lambda_1,\, \lambda_2-1,\, \dots,\, \lambda_d-(d-1)$$
by keeping only the smallest element of any subsequence of consecutive values, and then removing $-(d-1)$ if necessary. Indeed, the subsequences of consecutive values correspond to rows in $\lambda$ with equal length (thus when restricting the irrep to $S_{n-1}$, a box can be removed only from the lowest such row), while removing $-(d-1)$ corresponds to the possibility that $\lambda$ has less than $d$ rows. 

As an example we now explicitly present the spectrum of $H_\starn^{3}.$

\begin{example}\label{p:star3} Let $\lambda = (\lambda_1,\lambda_2,\lambda_3)$ be a partition of $n$ with three rows. The $d$-QMC irrep Hamiltonian $H_\starn^{\lambda}$ has at most three distinct eigenvalues, namely
\begin{enumerate}[\rm (1)]
\item if $\lambda_1>\lambda_2>\lambda_3$ then it has three eigenvalues
$$
e_1 = 2(n-\lambda_3+2), \quad e_2 =  2(n-\lambda_2+1), \quad e_3=2(n-\lambda_1),
$$
\item if $\lambda_1=\lambda_2>\lambda_3$ then it has two eigenvalues
$$
e_1 = 2 (n-\lambda_3+2), \quad e_2 = 2(n-\lambda_2+1),
$$
\item if $\lambda_1>\lambda_2=\lambda_3$ then it has two eigenvalues
$$
e_1 = 2 (n-\lambda_3+2), \quad e_2 = 2(n-\lambda_1),
$$
\item if $\lambda_1=\lambda_2=\lambda_3$ then it has one eigenvalue
$$
e_1 = 2 (n-\lambda_3+2).
$$
\end{enumerate}
The solution to the $3$-QMC problem for $\starn$ is $2(n+1),$ attained at partition of the form $\lambda=(\lambda_1,\lambda_2,1)$.
\end{example}

\begin{corollary}\label{c:dist}
Let $n\ge2$. If $\lambda,\mu$ are partitions of $n$ with distinct parts, then
$$\spec(H_{\starn}^\lambda)=\spec(H_{\starn}^\mu)
\quad\iff\quad\lambda=\mu.$$
\end{corollary}

\begin{remark}\label{r:dist}
The assumption about distinct parts in Corollary \ref{c:dist} is necessary. Indeed, $\lambda=(4,2,2,2,2)$ and $\mu=(5,5,1,1)$ satisfy
$\spec(H_{\bigstar_{12}}^\lambda)
=\{16,28\}=
\spec(H_{\bigstar_{12}}^\mu)$.
With a bit more effort, one also can find distinct partitions with equal height such that $H_{\starn}^\lambda$ and $H_{\starn}^\mu$ have the same eigenvalues:
\begin{align*}
\lambda&= (8,5,5,5,5,2,2),\\
\mu&= (9,9,4,4,2,2,2)
\end{align*}
satisfy $\spec(H_{\bigstar_{32}}^\lambda)
=\{24, 31, 36\}=
\spec(H_{\bigstar_{32}}^\mu)$.
\end{remark}

\subsection{Exact solutions for complete bipartite graphs} \label{sec:bipartite}
The star graph $\starn$ from the previous section is a special example of a complete bipartite graph. We now describe how the solution to the $d$-QMC problem can be obtained for a large class of complete bipartite graphs $K_{n-k,k}$ with $k \leq n/2$; by this we mean a graph whose vertices are separated into two subsets $N_1$ and $N_2$ of size $n-k$ and $k$ respectively, each with no internal connections and such that every vertex in $N_1$ is connected to every vertex in $N_2.$
The complement of $K_{n-k,k}$ consists of two cliques $K_{n-k}$ and $K_k$ (cf.~Equation \eqref{eq:star}), hence
\begin{equation}\label{eq:signed}
K_{n-k,k} = K_{n} - (K_{n-k} \oplus K_k).
\end{equation}
E.g., if $n=6$ and $k=2,$ then we have
	
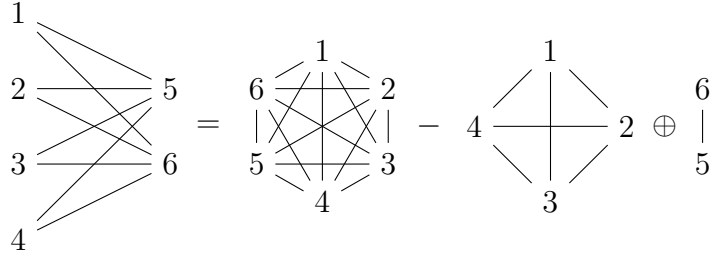
\begin{center}
\begin{tikzpicture}[scale=1]
	\begin{scope}[yshift=2cm, xshift=2cm]
	\foreach \i in {1,...,4} {
		\node (u\i) at (0, -\i) {\i};
	}
	\node (v1) at (2, -1.5 - 0.5) {5};
	\node (v2) at (2, -1.5 - 1.5) {6};
	\node at (2.5,-2.5) {$=$};
	\foreach \i in {1,...,4} {
		\foreach \j in {1,...,2} {
			\draw (u\i) -- (v\j);
		}
	}
\end{scope}
\begin{scope}[xshift=6cm,yshift=-0.5cm]
\graph { subgraph K_n [n=6, clockwise]};

\node at (1.4,0) {$-$};
\node at (4.5,0) {$\oplus$};
\end{scope}
\begin{scope}[xshift=9cm, yshift=-0.5cm]
 \graph { subgraph K_n [n=4, clockwise]};
\end{scope}
	\begin{scope}[xshift=11cm,yshift=-2cm]
	 \graph [clockwise] {
	5 -- 6; 
	};
	\end{scope}
\end{tikzpicture}
\end{center}
This gives a formula for the $d$-QMC Hamiltonian $H_{K_{n-k,k}}^{d},$ 
$$
H_{K_{n-k,k}}^{d} = H_{K_{n}}^{d} - 
\left(H_{K_{n-k}}^{d}{\otimes I_{d^k}} + {I_{d^{n-k}}\otimes }H_{K_k}^{d}\right).
$$
Note that the summands on the right-hand side commute.
Thus, for $\lambda \vdash n$ with at most $d$ rows,
\begin{equation}\label{eq:pre_n+m}
H_{K_{n-k,k}}^\lambda = H_{K_{n}}^\lambda - 
\left( (H_{K_{n-k}}\otimes I)^\lambda + (I\otimes H_{K_k})^\lambda\right).
\end{equation}
First, note that $H_{K_{n}}^\lambda$ is a known scalar matrix (by Lemma \ref{lemma scalar}). 
Second, the matrix $H_{K_{n-k}}\otimes I+I\otimes H_{K_k}$ belongs to the image of the subspace $\RR[S_{n-k}]\otimes \id+\id\otimes \RR[S_k]$ under the representation $\rho_n^{(d)}$ of $S_n$. 
This subspace is contained in the subalgebra $\CC[S_{n-k}]\otimes \CC[S_k]\cong \CC[S_{n-k}\times S_k]$ of $\CC[S_n]$.
In order to determine eigenvalues of \eqref{eq:pre_n+m}, it therefore suffices to consider the restriction of the irrep $\lambda$ of $S_n$ to a representation of $S_{n-k}\times S_k$.
The matrices $(H_{K_{n-k}}\otimes I)^\lambda$ and $(I\otimes H_{K_k})^\lambda$ commute, so the eigenvalues of $(H_{K_{n-k}}\otimes I)^\lambda+(I\otimes H_{K_k})^\lambda$ are sums of matching eigenvalues of $(H_{K_{n-k}}\otimes I)^\lambda$ and $(I\otimes H_{K_k})^\lambda$.
While the latter matrices are similar to $H_{K_{n-k}}^\lambda\otimes I$ and $I\otimes H_{K_k}^\lambda$, respectively, the transition matrices involved in similarities are distinct, because projecting to the irrep $\lambda$ within $\rho_n^{(d)}$ does not preserve the tensor decomposition $(\CC^d)^{\otimes (n-k)}\otimes (\CC^d)^{\otimes k}$. 
To compute the eigenvalues of \eqref{eq:pre_n+m}, one therefore needs to understand how the restriction of the irrep $\lambda$ on $S_n$ to $S_{n-k}\times S_k$ decomposes as a direct sum of irreducible representations of $S_{n-k}\times S_k$.
Similarly to the case of the star graph, a branching rule is invoked, this time the Littlewood-Richardson rule \cite[Section 4.9]{Sag01} together with the Frobenius reciprocity \cite[Theorem 1.12.6]{Sag01}. More precisely, the restriction of the irreducible module of $S_n$ corresponding to $\lambda\vdash n$ decomposes as
\begin{equation}\label{e:LR}
V_\lambda^{\downarrow S_{n-k}\times S_k}=\bigoplus_{\substack{\mu\vdash n-k,\\ \nu\vdash k}} 
\big(V_\mu \otimes V_\nu\big)^{c^\lambda_{\mu\nu}},
\end{equation}
where $c^\lambda_{\mu\nu}$ is the Littlewood–Richardson coefficient \cite[Section 4.9]{Sag01}.
Combining \eqref{eq:pre_n+m} and \eqref{e:LR}, the eigenvalues of $H_{K_{n-k,k}}^\lambda$ are
\begin{equation}\label{eq:lmn}
\Delta(\la,\mu,\nu):=\eta_\lambda-(\eta_\mu+\eta_\nu)
\end{equation}
over all pairs of $\mu\vdash n-k$ and $\nu\vdash k$ such that the Littlewood–Richardson coefficient $c^\lambda_{\mu\nu}$ is nonzero.

\subsubsection{Littlewood-Richardson coefficients and balanced partitions}\label{sec:LR}

In order to explain which triples of partitions $(\lambda,\mu,\nu)$ are admissible in \eqref{eq:lmn}, we introduce some further terminology. Let $\mu$ be contained in $\lambda$, in the sense that $\mu_i\le \lambda_i$ for all $i$. A filling $T$ of the skew-shaped Young diagram $\lambda/\mu$ with natural numbers is a \textit{Littlewood-Richardson (LR) tableau} if
\begin{enumerate}
    \item it is a semistandard Young tableau (its entries weakly increase along each row and strictly increase down each column), and
    \item the concatenation of reversed rows in $T$ is a lattice word (a word in which every prefix contains at least as many $i$s as $(i + 1)$s).
\end{enumerate}
The \textit{content} of a tableau $T$ is the partition whose $i$th part counts the number of $i$s in $T$. 
By \cite[Theorem 4.9.4]{Sag01}, the Littlewood-Richardson coefficient $c^\lambda_{\mu\nu}$ counts the number of LR tableaux of shape $\lambda/\mu$ with content $\nu$.

\begin{remark}
To apply the clique decomposition  with more than two summands to the $d$-QMC problem, let
  \begin{align*}
        H =  H^d_{G_1} \otimes  I_{d^{n-n_1}} + 
I_{d^{n_1}} \otimes H^d_{G_2} \otimes I_{d^{n-n_1-n_2}} + 
        \cdots + I_{d^{n-n_r}}\otimes H^d_{G_r}  ,
    \end{align*}
    where the graphs $G_i$ act on pairwise disjoint sets of indices of size $n_i$ with $n_1 + \cdots +n_r = n.$
   Then
   \begin{align*}
        H^\la =  (H^d_{G_1} \otimes  I_{d^{n-n_1}})^\la + \cdots +(I_{d^{n-n_r}}\otimes H^d_{G_r} )^\la
    \end{align*}
   and the eigenvalues of $H^\la$ are of the form
   $$
   \alpha_{1} + \cdots + \alpha_{r},
   $$
   where $\alpha_i$ are eigenvalues of $H^{\la_i}_{G_i}$ and 
   $\la_i \vdash n_i$ are such that the \textbf{iterated LR coefficient} $c^\la_{\la_1,\ldots,\la_r}$  is nonzero. The iterated LR coefficients are inductively defined via the usual LR coefficients $c_{\la_i, \la_{i+1}}^\cdot$ (see  \cite{KLMS12,GL20}). In fact,
\begin{align*}
   c^\la_{\la_1,\ldots,\la_r} = \sum_{\zeta \vdash n-n_r} c^\zeta_{\la_1,\ldots,\la_{r-1}} \cdot c^\la_{\zeta\,\la_r} 
   = \sum_{\zeta_1,\ldots,\zeta_{k-2}} c^{\zeta_1}_{\la_1\la_2} c^{\zeta_2}_{\zeta_1 \,\la_3} \cdots  c^{\zeta_{r-2}}_{\zeta_{r-3}\la_{r-1}}c^\la_{\zeta_{k-2}\,\la_r},
\end{align*}
where $\zeta_i \vdash n_1+\cdots +n_i.$
The iterated LR coefficients are also invariant under any permutation of the partitions.

In the special case when $G_i = K_{n_i}$ for all $i,$
 the eigenvalues of $H^\la$ are of the form
   $$
   \eta_{\la_1} + \cdots + \eta_{\la_r},
   $$
   where $\la_i \vdash n_i$ are such that $c^\la_{\la_1,\ldots,\la_r}$  is nonzero.

The above observation may be used to reduce the $d$-QMC problem for a given graph $G$ to simpler $d$-QMC problems. 
Every graph $G$ admits a unique tree clique decomposition \cite[Theorem 6.3]{BCEHK24} as a signed sum of simpler graphs $G_i$ (e.g., a complete bipartite graph is a clique minus two cliques as in \eqref{eq:signed}). The Hamiltonian $H_G$ then decomposes as a signed sum of Hamiltonians $H_{G_i}$ \cite[Theorem 6.5]{BCEHK24}. 
Hence, the iterated LR coefficients dictate as above how eigenvalues of $H_G^\lambda$ express as signed sums of eigenvalues of $H_{G_i}^{\lambda_i}$.\footnote{
In \cite[Section 6.5]{BCEHK24}, it is erroneously asserted that eigenvalues of $H_G^\lambda$ are signed Minkowski sums of eigenvalues of $H_{G_i}^{\lambda_i}$ for all $\lambda_i$ (without the non-vanishing condition on the iterated LR coefficients). 
}
In particular, if $G$ is a signed sum of cliques, the iterated LR coefficients and the formula \eqref{eq:etaD} for $\eta_\lambda$ give a combinatorial approach to solving the $d$-QMC problem for $G$.
\end{remark}

We conjecture that the following is true.

\begin{conjecture}\label{conj:bipartite}
The maximum of $\Delta(\lambda,\mu,\nu)$ as in \eqref{eq:lmn} is attained at a triple of partitions $\lambda \vdash n$,  $\mu \vdash n-k$ and $\nu \vdash k$ such that at most one part of $\lambda$ is the sum of a part in $\mu$ and a part in $\nu$, while the other parts of $\lambda$ are distributed in $\mu$ and $\nu$.
\end{conjecture}

Conjecture \ref{conj:bipartite} is supported by numerical experiments for all values $n,k,d$ with $2k \leq n$, $d < n$, $n \leq 26$. 
Proposition \ref{prop:02} below establishes a stronger version of Conjecture \ref{conj:bipartite} for small $k$ or small $d$. In general, Conjecture \ref{conj:bipartite} is a first step towards solving the $d$-QMC problem for complete bipartite graphs, as it resolves its combinatorial aspect. Namely, it reduces the $d$-QMC problem to maximizing \eqref{eq:lmn} over a particular class of triples $(\lambda,\mu,\nu)$ for which evidently $c_{\mu\nu}^\lambda\neq0$. One is then left with an integer optimization problem; nevertheless, as Theorem \ref{th:dbip} below suggests, one may still not expect a straightforward formula for the value of the $d$-QMC problem of a complete bipartite graph.

Given $\mu\vdash n-k$ and $\nu\vdash k$, let $\mu\uplus\nu\vdash n$ be the partition obtained by merging and sorting the parts of $\mu$ and $\nu$ (i.e., if partitions are viewed as multisets of rows, then $\mu\uplus\nu$ is the disjoint union of $\mu$ and $\nu$). We say that $\mu$ is a \textbf{subpartition} of $\lambda$ if $\mu$ is obtained from $\lambda$ by discarding some rows (this is a stronger condition than $\mu$ being contained in $\lambda$). In particular, $\mu$ and $\nu$ are subpartitions of $\mu\uplus\nu$.

 \begin{proposition}\label{prop:02}
When $k\leq 4$ or $d\le3$, the maximum of $\Delta(\lambda,\mu,\nu)$ as in \eqref{eq:lmn} is attained at a triple of partitions $\mu \vdash n-k$, $\nu \vdash k$ and $\lambda=\mu\uplus\nu \vdash n$. 
For such a triple, the coefficient $c^\lambda_{\mu \nu}$ is nonzero.
 \end{proposition}

For the rest of this section we focus on maximizing $\Delta(\mu\uplus\nu,\mu,\nu)$.
The proof of Proposition \ref{prop:02} is rather technical, and is presented in Appendix \ref{app:LR}.

\begin{remark}\label{rem:conj1}
The triple $(n,k,d)=(10,5,5)$ is the first where the conclusion of Proposition \ref{prop:02} fails.
The solution to the $d$-QMC problem for $K_{5,5}$ is 
72 attained at
\[
    \la=(2,2,2,2,2),\; \mu=(2,2,1),\; \nu=(2,2,1).
\]
Maximizing $\Delta(\la,\mu,\nu)$ over triples $(\la,\mu,\nu)$ as in Proposition \ref{prop:02} yields 70, attained at
\[
    \la=(3, 2, 2, 2, 1),
    \; \mu= (2, 2, 1),\; \nu=(3, 2)
    \quad \text{and} \quad
    \la=(3, 2, 2, 2, 1),\; \mu=(3, 2), \; \nu=({2, 2, 1}).
\]

\end{remark}

\begin{lemma}\label{lem:03}
The expression $\Delta(\mu\uplus\nu,\mu,\nu)$ is maximized when $\mu$ and $\nu$ are balanced.
\end{lemma}

\begin{proof} 
Let $\lambda=\mu\uplus\nu$. 
Choose the largest index $k,$ for which there is an $l>k$ such that
$\lambda_k=\mu_i,\lambda_l=\mu_j$ for some $i<j$ and
$\lambda_l -\lambda_k \geq 2.$ Then construct $\lambda^\dagger$ (and $\mu^\dagger$) by moving a box from $\lambda_l$ ($\mu_j$ resp.) to $\lambda_k$ ($\mu_i$ resp.). Note that by the choice of $k,$ the obtained $\lambda^\dagger$ and $\mu^\dagger$ are indeed valid partitions. Then
    \begin{align*}
        \eta_\lambda - \eta_\mu -\eta_\nu- (\eta_{\lambda^\dagger} - \eta_{\mu^\dagger} -\eta_\nu) = & - (\lambda_k-k+1)^2 + (\lambda_k^\dagger-k+1)^2 \\
        & - (\lambda_l-l+1)^2  + (\lambda_l^\dagger-l+1)^2\\
       &+ (\mu_{i}-i+1)^2 - (\mu_i^\dagger-i+1)^2 \\
        &+ (\mu_j-j+1)^2  - (\mu_j^\dagger-j+1)^2\\
         = \ & 2 (j-i - l + k) <0
    \end{align*}
    since clearly, $j-i \leq l-k.$ After repeating this procedure inductively we deduce that a balanced $\mu$ gives the highest value of \eqref{eq:lmn}.
    Since the procedure does not affect $\nu$ and the rows of $\mu$ are still disjoint from the rows of $\nu,$ the coefficient $c_{\mu^\dagger \nu}^\lambda$ is nonzero.
    By symmetry, the same holds for $\nu$ instead of $\mu.$
\end{proof}

\begin{lemma}\label{lem:04}
The expression $\Delta(\mu\uplus\nu,\mu,\nu)$ is maximized when $\het(\mu\uplus\nu)=d$.
\end{lemma}

\begin{proof}
Let $\lambda=\mu\uplus\nu$. 
    If $\lambda \vdash n$ has height $f <d,$ let $l_0$ be the biggest index such that $\lambda_{l_0}>1.$ 
    The partition $\lambda^\dagger \vdash n$ is obtained from $\lambda$ 
    by moving the last box of the $l_0$th row of $\lambda$ to a new row (so that $\lambda^\dagger$ has $f+1$ rows).

Denote $e = \text{ht}(\mu).$
    If $\lambda_{l_0}=\mu_j$ for some $j,$ then $\mu^\dagger$ is constructed from $\mu$ by  moving the last box of the $j$th row of $\mu$ to a new row (so that $\mu^\dagger$ has $e+1$ rows). Note that by the definition of $l_0,$  we indeed obtain a valid partition. Hence,
    \begin{align*}
        \eta_\lambda  - \eta_\mu  - (\eta_{\lambda^\dagger}- \eta_{\mu^\dagger}) = \,& \,  \frac{f(f-1)(2f-1)}{6} - \frac{(f+1)f(2f+1)}{6} \\
        & - \frac{e(e-1)(2e-1)}{6} + \frac{(e+1)e(2e+1)}{6}\\
        \phantom{{}={}}& -   (\lambda_{l_0}-l_0+1)^2 + (\lambda^\dagger_{l_0}-l_0+1)^2 + (\lambda^\dagger_{f+1}-f)^2\\
        \phantom{{}={}}& +  (\mu_j-j+1)^2 - (\mu^\dagger_j-j+1)^2 - (\mu^\dagger_{e+1}-e)^2\\
        = \, & \,2 (e - j -f + l_0)  \leq 0. 
    \end{align*}
    If $\lambda_{l_0}=\nu_k$ for some $k,$ then construct $\nu^\dagger$ from $\nu$ by  moving the last box of the $k$th row of $\nu$ to a new row. By analogy with the above computation, we have 
    $$
    \eta_\lambda  - \eta_\nu  - (\eta_{\lambda^\dagger}- \eta_{\nu^\dagger}) \leq 0.
    $$
    It is finally clear that $c_{\mu \nu^\dagger}^\lambda >0.$ 
\end{proof}

By Lemmas \ref{lem:03} and \ref{lem:04}, we restrict to $\het(\mu\uplus\nu)=d$ and balanced $\mu\vdash n-k$, $\nu\vdash k$ for the rest of the section.
If $\mu\vdash n-k$ and $\nu\vdash k$ are such that the last part of $\mu$ is not smaller than the first part of $\nu$, we write $\mu\uplus\nu$ as $(\mu,\nu)$, to stress that this partition of $n$ obtained by concatenating $\mu$ and $\nu$. In the following lemmas, whenever $(\mu,\nu)\vdash n$ is referred to, it is assumed that $\mu$ and $\nu$ are suitable for $(\mu,\nu)$ to be valid.

\def\la{\lambda}
\begin{lemma}\label{lem:41}
Suppose $\mu\vdash n-k$ and $\nu\vdash k$ are such that $(\mu,\nu)$ is valid. 
Letting $e=\het(\mu)$, we have
\begin{equation}\label{eq:delta}
    \Delta((\mu,\nu),\mu,\nu)= 2 k (e + n -k).
\end{equation}
\end{lemma}

\begin{proof}
Let $\lambda=(\mu,\nu)$.
Using the formula \eqref{eq:etaD} from Proposition \ref{prop:etaD}, we have
\begin{equation}\label{eq:deltapart1}
    \begin{split}
\Delta(\la,\mu,\nu) & = n^2-(n-k)^2-k^2 \\ & 
\phantom{{}={}}
+ \frac{d(d-1)(2d-1)}{6} - \frac{e(e-1)(2e-1)}{6} - \frac{(d-e)(d-e-1)(2d-2e-1)}{6} \\
&
\phantom{{}={}}
 -\sum_{j=1}^d\big( \la_j - (j-1)\big)^2
+ \sum_{j=1}^e\big( \mu_j - (j-1)\big)^2
+ \sum_{j=1}^{d-e}\big( \nu_j - (j-1)\big)^2
    \end{split}
\end{equation}
Since the first $e$ rows of $\la$ form $\mu$, 
the third line in \eqref{eq:deltapart1} simplifies into
    \begin{align*}
&-\sum_{j=e+1}^d\big( \la_j - (j-1)\big)^2
+ \sum_{j=1}^{d-e}\big( \nu_j - (j-1)\big)^2  \\
= &-\sum_{j=1}^{d-e}\big( \la_{j+e} - (j+e-1)\big)^2
+ \sum_{j=1}^{d-e}\big( \nu_j - (j-1)\big)^2 \\
= &-\sum_{j=1}^{d-e}\big( \nu_j - (j+e-1)\big)^2
+ \sum_{j=1}^{d-e}\big( \nu_j - (j-1)\big)^2 \\
= & \sum_{j=1}^{d-e} e \big( 2 \nu_j - 2j -e +2\big) 
= e (d - d^2 - e + d e + 2 k).
    \end{align*}
Putting this back into \eqref{eq:deltapart1} and simplifying the obtained expression yields \eqref{eq:delta}.
\end{proof}

\begin{corollary}\label{cor:d=2}
The solution to the $2$-QMC problem for $K_{n-k,k}$ is
$2k(1+n-k)$.
\end{corollary}

\begin{proof}
As explained in the beginning of this section, we are maximizing \eqref{eq:delta}. 
By Proposition \ref{prop:02}, this is maximized when 
$\la=(n-k,k)$, $\mu=(n-k)$ and $\nu=(k)$. The desired value is then given in Lemma \ref{lem:41}.
\end{proof}

\begin{lemma}\label{lem:42}
Assume $\mu\vdash n-k$ and $\nu\vdash k$ are balanced, 
$\het(\mu)=e$, 
and $\lambda=(\mu,\nu)\vdash n$.
Suppose $\mu^\dagger$ is the balanced partition of $n-k$ on $e+1$ rows, $\nu^\dagger$ is the balanced partition of $k$ on $d-e-1$ rows. If $\lambda^\dagger=(\mu^\dagger,\nu^\dagger)$ is a (valid) partition of $n$, then
\[
    \Delta(\la^\dagger,\mu^\dagger,\nu^\dagger) > 
    \Delta(\la,\mu,\nu).
\]
\end{lemma}

\begin{proof}
Immediate from Lemma \ref{lem:41}.
\end{proof}

\def\dag{\dagger}
\begin{lemma}\label{lem:43}
Suppose $\mu\vdash n-k$ and $\nu\vdash k$ are balanced, 
$\het(\mu)=e$, 
and $\lambda=(\mu,\nu)\vdash n$.
Suppose $\mu^\dagger$ is the balanced partition of $n-k$ on $e+1$ rows, $\nu^\dagger$ is the balanced partition of $k$ on $d-e-1$ rows. Assume $\la^\dag=(\nu^\dag,\mu^\dag)\vdash n$.
Then
\[
    \Delta(\la,\mu,\nu)-\Delta(\la^\dag,\mu^\dag,\nu^\dag) =
    2 \big((-1 + d) k + (1 - d + e) n\big).
\]
In particular, 
\[\Delta(\la,\mu,\nu)\geq \Delta(\la^\dag,\mu^\dag,\nu^\dag) 
\quad\iff\quad
e\geq(d-1) \Big(1-\frac{k}{n}\Big).
\]
\end{lemma}

\begin{proof}
This is again immediate from Lemma \ref{lem:41}.
\end{proof}

Further analysis splits into two main cases, according to the following definition.

\begin{definition}\label{def:unbalancing}
We call a triple $(n,k,d)$ \textbf{balancing} if the
rows of the balanced partition of $n$ of height $d$ 
can be partitioned into a (balanced) partition of $n-k$
and a (balanced) partition of $k$. Otherwise we call the triple
$(n,k,d)$ \textbf{unbalancing}.
\end{definition}

\begin{remark}\label{rem:bal}
Letting $q=\floor{\frac nd}$ and $r=n-qd$, the triple
$(n,k,d)$ is balancing iff there are integers $0\le s\le r$ and $0\le t\le d-r$ such that $k=s(q+1)+tq$. This is equivalent to the existence of an integer $\max\{0,\frac{k-q(d-r)}{q+1}\}\le s\le \min\{r,\frac{k}{q+1}\}$ such that $k\equiv s\mod q$.
\end{remark}

\subsubsection{Unbalancing triples}

First, we maximize $\Delta(\mu\uplus\nu,\mu,\nu)$ when $(n,k,d)$ is unbalancing.

\begin{lemma}\label{lem:balancing}
Let $k,n,d$ be positive integers with $k< n$ and $n\ge d$.  If $\frac{dk}{n}\in\mathbb{N}$, then the rows of the balanced partition $\la$ of $n$ of height $d$ can be split into a partition of $k$ and a complementary partition of $n-k$.
\end{lemma}

\begin{proof}
Write
\[
n = qd + r,
\qquad
0 \le r < d,
\]
so that the balanced partition $\lambda$ 
consists of $r\times (q+1)$ and $(d-r)\times q$.

Set $$t=\frac{dk}{n}\in\mathbb{N}.$$
Because $k< n$, one has $t<d$.  Now define
\[
x =\frac{rt}{d},
\qquad
y =\frac{(d-r)t}{d}.
\]
Then \[
    x=\frac{(n-qd)t}{d}= \frac{nt}d-qt=\frac{n dk}{nd}-qt=k-qt,
\]
 so $x,y\in\mathbb{N}$. Further, $x+y=t$, $x\le r$, and $y\le d-r$.  Finally,
\begin{align*}
x(q+1)+y\,q
= q(x+y)+x 
= qt + \frac{r\,t}{d} 
= (qd+r)\frac{t}{d} 
= n\frac{t}{d} 
= k.
\end{align*}
Hence taking $x$ of the ``large'' parts of $\la$ and $y$ of the ``small'' parts gives a partition of $k$, and the remaining parts sum to $n-k$.
\end{proof}

\begin{lemma}\label{lem:unbalancing}
Suppose $(n,k,d)$ is unbalancing. 
Let $e$ be the largest integer such that  
$\lfloor \frac{n-k}{e}\rfloor \ge \lceil\frac{k}{d-e}\rceil$, i.e., the tail of the balanced partition of $n-k$ of height $e$
is at least as big as the head of the balanced partition of $k$ of height $d-e$.

Then \begin{equation}\label{eq:unbalancing}
e=\left\lfloor d \Big(1-\frac kn\Big)\right\rfloor.
\end{equation}
\end{lemma}

\begin{proof}
Since $(n,k,d)$ is unbalancing, the head of the balanced partition of $n-k$ of height $e+1$
is at most as big as the tail of the balanced partition of $k$ of height $d-e-1$. This yields the following two inequalities,
\[
\left\lfloor \frac{n-k}{e}\right\rfloor \ge \left\lceil\frac{k}{d-e}\right\rceil,
\quad
    \left\lceil \frac{n-k}{e+1}\right\rceil \le \left\lfloor\frac{k}{d-e-1}\right\rfloor.
\]
In particular,
\[
     \frac{n-k}{e} \ge \frac{k}{d-e},
\quad
     \frac{n-k}{e+1} \le \frac{k}{d-e-1}.
\]
Clearing denominators yields
\[
    d (n-k)-e n\geq0, \quad 
    d (n-k)-(e+1) n \leq0.
\]
The two inequalities imply
\[
    e\leq \frac{d(n-k)}n, \quad
    e \geq \frac{d(n-k)}n-1.
\]
Since $e$ is an integer, this is equivalent to
\[
       \left\lceil\frac{d(n-k)}n\right\rceil-1\leq 
       e\leq \left\lfloor\frac{d(n-k)}n\right\rfloor.
\]
As $(n,k,d)$ is unbalancing, $\frac{dk}n\not\in\N$ by Lemma \ref{lem:balancing},
whence
\[
\left\lceil\frac{d(n-k)}n\right\rceil-1=
\left\lfloor\frac{d(n-k)}n\right\rfloor,
\]
and \eqref{eq:unbalancing} follows.
\end{proof}

\begin{lemma}\label{lem:unbalanced2}
Let $k,n,d$ be positive integers with $2k\le n$ and $n\ge d$.  Then
\begin{equation}\label{eq:unbalanced2}
\left\lfloor d\Bigl(1-\frac{k}{n}\Bigr)\right\rfloor
\;\ge\;
(d-1)\Bigl(1-\frac{k}{n}\Bigr)
\end{equation}
holds if and only if, letting $dk = r\bmod n$,
one has
\[
r\in\{0\}\;\cup\;\{k,k+1,\dots,n-1\}.
\]
Equivalently, inequality \eqref{eq:unbalanced2} fails precisely when
\[
1 \le r \le k-1.
\]
\end{lemma}

\begin{proof}
If $r=0$, then \eqref{eq:unbalanced2} holds. Thus assume $r>0$.
Set
\begin{equation}\label{eq:ABun2}
A = d\Bigl(1-\frac{k}{n}\Bigr),
\quad
B = (d-1)\Bigl(1-\frac{k}{n}\Bigr).
\end{equation}
Then
\[
\lfloor A\rfloor \ge B
\;\iff\;
\lfloor A\rfloor +\Bigl(1-\frac{k}{n}\Bigr)\;\ge\;A
\;\iff\;
\{A\}\;\le\;1-\frac{k}{n},
\]
where $\{A\}=A-\lfloor A\rfloor$.  

From \eqref{eq:ABun2} we have
\[
\{A\}
=1-\frac rn.
\]
Thus
\[
\{A\}\le 1-\frac{k}{n}
\;\iff\;
r\geq k,
\]
as claimed.
\end{proof}

\begin{corollary}\label{cor:unbalancing}
Suppose $(n,k,d)$ is unbalancing, and let $dk = r\bmod n$ for $0< r<n$.
The maximum of $\Delta(\mu\uplus\nu,\mu,\nu)$ equals
$$\left\{
\begin{array}{ll}
2k\left(\left\lfloor d\Bigl(1-\frac{k}{n}\Bigr)\right\rfloor+ n-k\right) & \text{if }r\ge k\\
2 (n-k)\left(d - \left\lfloor d\Bigl(1-\frac{k}{n}\Bigr)\right\rfloor-1 + k\right) & \text{if }r<k.
\end{array}
\right.$$
\end{corollary}

\begin{proof}
Let $e'=\lfloor d(1-\frac{k}{n})\rfloor$. By Lemma \ref{lem:unbalancing}, $e'$ is the largest $e$ such that $\lambda=(\mu,\nu)$, with balanced $\mu\vdash n-k$ of height $e$ and balanced $\nu\vdash k$, is a valid partition of $n$.
If $r\ge k$, then $e'\ge (d-1)(1-\frac{k}{n})$ by Lemma \ref{lem:unbalanced2}. Thus, $\Delta(\mu\uplus\nu,\mu,\nu)$ is maximized at the balanced $\mu\vdash n-k$ of height $e'$, the balanced $\nu\vdash k$ and $\mu\uplus\nu=(\mu,\nu)$ by Lemmas \ref{lem:42} and \ref{lem:43}.
If $0<r<k$, then Lemmas \ref{lem:unbalanced2} and \ref{lem:43} show that $\Delta(\mu\uplus\nu,\mu,\nu)$ is maximized at the balanced $\mu\vdash n-k$ of height $e'+1$, the balanced $\nu\vdash k$ and $\mu\uplus\nu=(\nu,\mu)$.
In both cases, 
the value of $\Delta(\mu\uplus\nu,\mu,\nu)$
is then given by Lemma \ref{lem:41}.
\end{proof}

\subsubsection{Balancing triples}

Next, we 
maximize $\Delta(\mu\uplus\nu,\mu,\nu)$ when $(n,k,d)$ is balancing. 
Throughout this section let $n,k,d\in\N$ satisfy 
$2k\le n$ and $n\ge d$.  
Put  
\begin{equation}\label{eq:qr}
  q=\Bigl\lfloor\frac{n}{d}\Bigr\rfloor,\qquad
  r=n-qd\quad(0\le r<d),
\end{equation}
so the balanced partition of weight $n$ and height $d$ is  
\[
  \lab=(\underbrace{q+1,\dots ,q+1}_{r\text{ rows}},\;
           \underbrace{q,\dots ,q}_{d-r\text{ rows}}).
\]

\begin{lemma}\label{lem:prebalast}
Suppose $d \mid n$, and let $q=\frac nd$. 
If $q \nmid n-k $, then the triple $(n,k,d)$ is unbalanced and thus handled by Corollary \ref{cor:unbalancing} above.
If $q\mid n-k$ and $n-k=eq$, then 
the maximum of $\Delta(\mu\uplus\nu,\mu,\nu)$ is
attained at $\mu=(q^e)$, $\nu=(q^{d-e})$ and equals
\begin{equation}\label{eq:prebalast}
    2 k (n-k) \Big(1+\frac dn\Big). 
\end{equation}
\end{lemma}

\begin{proof}
Let $\lambda=\mu\uplus\nu$. 
Letting $\mu^+$ denote the balanced partition of $n-k$ of height $e+1$, and $\nu^+$ the balanced partition of $k$ of height $d-e-1$, we have $\la^+=(\nu^+,\mu^+)\vdash n$. Thus
\[
    \Delta(\la^+,\nu^+,\mu^+) = 2(n-k)(k+d-e-1),
\]
and using $e=\frac{d(n-k)}n$, we get
\[
    \Delta(\la,\nu,\mu)-\Delta(\la^+,\nu^+,\mu^+) = 2 (n-k) >0.
\]
Increasing $e$ further will decrease the height of $\nu^+$ and thus by Lemma \ref{lem:42} only decrease the value of $\Delta$. Thus the maximum $\Delta$ is attained at 
$\mu=(q^e)$ and $\nu=(q^{d-e})$, as claimed. The formula \eqref{eq:prebalast} now follows from Lemma \ref{lem:41}.
\end{proof}

\begin{proposition}\label{prop:sth}
Suppose $\mu\vdash n-k$ and $\nu\vdash k$ are balanced, 
with $\mu_1=\nu_1$ and 
$\het(\mu)=e$.
Let $s$ be such that $\mu_1=\cdots=\mu_s>\mu_{s+1}=\cdots=\mu_e$. 
Then 
\begin{equation}\label{eq:deltaS}
    \Delta(\mu\uplus\nu,\mu,\nu) = 
2 (-k^2 + k (n + s) + (d - e) (e-s) \mu_e),
\end{equation}
and this function of $e$ is maximized at
\begin{equation*}
    e^* = \frac{d}{2} + \frac{n-2k}{2(q+1)}.
\end{equation*}
\end{proposition}

\begin{proof}
Since $\mu,\nu$ are balanced and $\mu_1=\nu_1$, we have
$$\mu\uplus\nu=(\mu_1,\ldots,\mu_s,\nu,\mu_{s+1},\ldots,\mu_e).$$
    A similar calculation to the one in Lemma \ref{lem:41} gives
\begin{equation}\label{eq:deltab}
    \Delta(\mu\uplus\nu,\mu,\nu) = 
2 (-k^2 + k (n + s) + (d - e) \sigma),
\end{equation}
where $\sigma=\mu_{s+1}+\cdots+\mu_e=(e-s)\mu_e$. Note that \eqref{eq:deltab} becomes \eqref{eq:delta} when $s=e$, and hence $\mu\uplus\nu = (\mu, \nu).$
Moreover, when $s=0$ and $\mu\uplus\nu = (\nu,\mu)$, then $e\mu_e = n-k$ and hence \eqref{eq:deltab} becomes $2(n-k) (d - e + k)$, i.e., \eqref{eq:delta} with $k$ and $e$ replaced by $n-k$ and $d-e$ respectively.

Expressing $s$ from
$$
n-k = s (q+1) + (e-s)q,
$$
where $q = \lfloor \frac{n}{d}\rfloor,$ and plugging it into \eqref{eq:deltab} produces a concave quadratic function in $e,$ namely
\begin{equation}\label{eq:deltas}
\Delta=2\big(-(q+q^2)e^2 + (d q-2 k q+ n q + d q^2 )e
 -2 k^2 + 2 k n   + d k q  - d n q\big).
\end{equation}
This function attains its maximum at
\begin{equation}\label{eq:defestar}
    e^* = \frac{d}{2} + \frac{n-2k}{2(q+1)}. \qedhere
\end{equation} 
\end{proof}

\begin{remark}\label{rem:boundestar}
Using the obvious inequalities 
\[
    \frac nd-1\leq\floor{\frac nd}\leq \frac nd   , 
\]
we obtain the following bounds on $e^*$:
\[
     \frac{d}{2} 
     +\frac{d (n-2 k)}{2 (d+n)}
     \leq e^*\leq d \Big(1-\frac kn\Big) .
\]
\end{remark}

Assume $(n,k,d)$ is balancing, and w.l.o.g., $n-k\geq k$. 
The strategy of the proof is as follows: 
by Lemmas \ref{lem:03} and \ref{lem:04}, we know that $\Delta(\mu\uplus\nu,\mu,\nu)$ in 
\eqref{eq:lmn} is maximized with $\het(\mu)=\het(\nu)=d$ and
$\mu\vdash n-k$, $\nu\vdash k$ being balanced. 
Such pairs are thus uniquely determined by $e=\het(\mu)$. 
Note that $\mu\uplus\nu=(\mu,\nu)$ for small $e$, and $\mu\uplus\nu=(\nu,\mu)$ for large $e$.
By Lemma \ref{lem:42}, $\Delta$ increases while $\mu$ forms the top rows of $\mu\uplus\nu$ and $e$ increases. Conversely, $\Delta$ decreases when $\nu$ is the top of $\mu\uplus\nu$ and $e$ increases.
It is thus key to analyze transitions, where the rows of $\mu$, $\nu$ appear mixed in $\mu\uplus\nu$.

For ease of notation,   
let $\mu[e]$ denote the balanced partition of $n-k$
 with height $e$. Similarly, 
$\nu[e]$ denotes the balanced partition of $k$ of height $d-e$. 
 Let
\begin{equation}\label{eq:defe0}
    e_0=\max \left\{e\in\{1,\ldots,d-1\} \colon \floor{\frac{n-k}e} \geq \bigg\lceil\frac{k}{d-e}\bigg\rceil  \right\}.
\end{equation}
Since $n-k\geq k$, $e=1$ belongs to the right-hand side set.
Similarly, let
\begin{equation}\label{eq:defe1}
    e_1=\min \left\{e\in\{1,\ldots,d-1\} \colon
\floor{\frac{k}{d-e}} \geq 
     \bigg\lceil\frac{n-k}e\bigg\rceil \right\}.
\end{equation}
If the set in the definition of $e_1$ is empty
(e.g., $n=5$, $k=2$, $d=2$), we set $e_1=d-1$. 
By Lemma \ref{lem:42}, the  $e$ maximizing the $d$-QMC Hamiltonian for $K_{n-k,k}$ will satisfy $e_0\leq e\leq e_1$.

Let $\lab$
denote the balanced partition of $n$ with height $d$.
Set 
\[
    \mathfrak E:=\{e \mid \mu[e] \text{ is a subpartition of }\lab \} \subseteq \{e_0, e_0+1,\ldots,e_1\}. 
\]
 Since $(n,k,d)$ is balancing, $\mathfrak E\neq\emptyset$.
 Further, $e\in\mathfrak E$ iff 
$e=\frac{n-k}{q+1}\le r$, or
$\floor{\frac{n-k}e}=q$ and $\frac{n-k-r}{q}\le e\le \frac{d+n-k-r}{q+1}$.

\begin{lemma}\label{lem:frakinterval}
The feasible set $\mathfrak E$ is an interval, that is,
$\mathfrak E=\{\min \mathfrak E, \min \mathfrak E+1,\ldots, \max \mathfrak E \}$.
\end{lemma}

\begin{proof}
Recall the unique balanced partition $\lab$ of $n$ into $d$ parts consists of $r$ parts of size $q+1$ and $d-r$ parts of size $q$.
Suppose $e<e'$ are both in $\mathfrak E$.
Thus there exists a height $e$ subpartition $\mu\vdash n-k$ of $\lab$. Similarly, $\mu'\leq\lab$ is a height $e'$ subpartition with weight $n-k$.
We claim that for any integer $e''$ such that $e < e'' < e'$, 
there exists a height $e''$ subpartition $\mu''\vdash n-k$ of $\lab$.

The partition $\mu$ is of the form, say, $((q+1)^x,q^y)$ for some $x,y\in\N_0$. Thus
\begin{equation}
\begin{aligned}
    x + y &= e, \\
    x(q+1) + yq &= n-k, \label{eq:sum1} \\
    0 \le x \le r &\quad \text{and} \quad 0 \le y \le d-r.
\end{aligned}
\end{equation}
Substituting $y = e-x$ into the second equation of  \eqref{eq:sum1} gives
\begin{align}
    x + eq &= n- k \label{eq:simple_sum1}.
\end{align}
Likewise,
if $\mu'=((q+1)^{x'},q^{y'})$ for some $x',y'\in\N_0$, then
\begin{align}
    x' + e'q &= n-k \label{eq:simple_sum2}.
\end{align}

Equating \eqref{eq:simple_sum1} and \eqref{eq:simple_sum2}, we have
\begin{align}
    x - x' &= q(e' - e) \label{eq:x_diff}.
\end{align}
Since $e < e'$,  $x > x'$. 
Further, by considering the difference in the number of $q$-sized parts:
\begin{equation}
\begin{aligned}
    y' - y &= (e' - x') - (e - x) = (e' - e) + (x - x') 
    \\
    &= (e' - e) + q(e' - e) 
    = (1+q)(e' - e) \label{eq:y_diff},
\end{aligned}
\end{equation}
whence $y' > y$.

Now, let $e''$ be an integer such that $e < e'' < e'$. Define $\delta  = e'' - e$. Then $0 < \Delta e < e' - e$.
We propose constructing $\mu''$  composed of $x''$ parts of size $q+1$ and $y''$ parts of size $q$, defined as follows:
\begin{equation}
\begin{aligned}
    x'' &= x - q \delta \label{eq:def_x_pp} \\
    y'' &= y + (1+q) \delta 
\end{aligned}
\end{equation}
We shall verify that $\mu''$ satisfies the required conditions.

 Firstly, the height of $\mu$ is
    \begin{align*}
    \het(\mu'') & =
        x'' + y'' = (x - q \delta) + (y + (1+q) \delta ) = (x+y) +\delta  = e + (e'' - e) = e''.
    \end{align*}
Next, $\mu''\vdash n-k$, as
    \begin{align*}
        x''(q+1) + y''q &= \big(x - q \delta \big)(q+1) + \big(y + (1+q) \delta\big) q\\  &= \big(x(q+1) + yq\big) - (q^2+q)\delta + (q+q^2)\delta = x(q+1) + yq = n-k.
    \end{align*}
    
Finally, we need to show $0 \le x'' \le r$ and $0 \le y'' \le d-r$. Clearly, 
     $x'' = x - q \delta < x \leq r$. 
  Also, $x'' = x - q \delta \ge x - q(e' - e) = x' \geq 0$.
To verify the desired properties of $y''$, we have
$y'' = y + (1+q)\delta > y\geq0$, and
$y'' = y + (1+q)\delta e < y + (1+q)(e' - e) = y'\leq d-r$.  This completes the proof.
\end{proof}

By definition of $e_0$ and $e_1,$ whenever $e \leq e_0,$ the partition $\lambda\vdash n$ made of rows of $\mu[e],\nu[e]$ is $(\mu[e], \nu[e])$. If $e \in \mathfrak{E},$ then this partition $\la$ equals $\lab$, and is of the form
    $\la=(\mu_1,\ldots,\mu_s,\nu,\mu_{s+1},\ldots,\mu_e),$
where $\mu_1=\cdots=\mu_s>\mu_{s+1}=\cdots=\mu_e$
(note the two edge cases, $s=0$ or $s=e$ can also both occur). 
If $e\geq e_1,$ then $\la$ is of the form $(\nu[e], \mu[e]).$  Note that $(\mu[e_0], \nu[e_0])$
may or may not be balanced, and the same holds for $(\nu[e_1], \mu[e_1]).$ 
 
\begin{lemma}\label{lem:frakinterval+}
If $e_0\not\in\mathfrak E$, then $e_0+1\in\mathfrak E$.
Likewise, if $e_1\not\in\mathfrak E$, then $e_1-1\in\mathfrak E$. In particular, $\mathfrak E$ is one of the following four discrete intervals:
\[
    \{e_0,e_0+1,\ldots,e_1\},
    \;
    \{e_0,e_0+1,\ldots,e_1-1\},
    \;
    \{e_0+1,e_0+2,\ldots,e_1\},
    \;
    \{e_0+1,e_0+2,\ldots,e_1-1\}.
\]
\end{lemma}

\begin{proof}
Suppose $e_0\not\in\mathfrak E$, and consider $e_0+1$. By definition of $e_0,$
\[
    \floor{\frac{n-k}{e_0+1}} < \bigg\lceil\frac{k}{d-e_0-1}\bigg\rceil.
\]
Since $\min\mathfrak E\geq e_0+1$, we deduce
\begin{equation}\label{eq:squeeze}
\floor{\frac{n-k}{\min\mathfrak E}} \leq \floor{\frac{n-k}{e_0+1}} < \bigg\lceil\frac{k}{d-e_0-1}\bigg\rceil \leq 
 \bigg\lceil\frac{k}{d-\min\mathfrak E}\bigg\rceil
.
\end{equation}
By definition of $\mathfrak E$, the partitions 
$\mu[\min\mathfrak E]$, $\nu[\min\mathfrak E]$ combine to form $\lab$, 
whence 
\[
    \bigg\lceil\frac{k}{d-\min\mathfrak E}\bigg\rceil-\floor{\frac{n-k}{\min\mathfrak E}}\leq 1.
\]
Hence the same holds with $\min\mathfrak E$ replaced by $e_0+1$ by
\eqref{eq:squeeze}. Thus $\mu[e_0+1]$, $\nu[e_0+1]$ must combine to form $\lab$, too. That is, $e_0+1\in\mathfrak E$.

The proof for $e_1$ is the same.
\end{proof}

\def\close#1{\left\lfloor #1 \right\rceil}

\subsubsection{Summary}

We combine the preceding results on unbalancing and balancing triples into the following statement.

\begin{theorem}\label{th:dbip}
Let $n\geq2k$ and $n>d$. 
Let $e_\star$ be the closest integer in 
$\mathfrak E$ to $e^*$. Then the maximum of $\Delta(\mu\uplus\nu)=\eta_{\mu\uplus\nu}-\eta_\mu-\eta_\nu$ for $\mu\vdash n-k$ and $\nu\vdash k$ is attained at one of the following pairs:
\begin{enumerate}[\rm (a)]
    \item
$\mu$ is balanced of height $e_0$ and $\nu$ is  balanced of height $d-e_0$;
   \item
$\mu$ is balanced of height $e_\star$ and $\nu$ is  balanced of height $d-e_\star$;
   \item
$\mu$ is balanced of height $e_1$ and $\nu$ is  balanced of height $d-e_1$.
\end{enumerate}
\end{theorem}

\begin{proof}
If the triple $(n,k,d)$ is unbalancing, then $\mathfrak{E}=\emptyset$, and the maximum is attained at (a) or (c) by Corollary \ref{cor:unbalancing}.
Now suppose that the triple $(n,k,d)$ is balancing. 
By Lemma \ref{lem:42}, the maximum of $\Delta(\mu\uplus\nu,\mu,\nu)$ is attained for a balanced $\mu\vdash n-k$ of height $e$ and a balanced $\nu\vdash k$ of height $d-e$, where $e$ is an integer between $e_0$ and $e_1$.
The characterization of $\mathfrak{E}$ as in Lemma \ref{lem:frakinterval+} and maximization of the function \eqref{eq:deltaS} at $e^*$ in Proposition \ref{prop:sth} show that the maximum is attained for $e\in\{e_0,e_1,e_\star\}$, which gives rise to the cases (a), (b) and (c).
\end{proof}

Proposition \ref{prop:02} and Theorem \ref{th:dbip} resolve the $d$-QMC problem for $K_{n-k,k}$ when either $k$ or $d$ is small.

\begin{corollary}\label{cor:dbip}
Let $k\le 4$ and $d<n$. 
Let $e_\star$ be the closest integer in 
$\mathfrak E$ to $e^*$. Then the solution to the $d$-QMC problem for $K_{n-k,k}$ is the biggest of the following three values 
\[\eta_{\mu\uplus\nu}-\eta_\mu-\eta_\nu\]
for $\mu\vdash n-k$ and $\nu\vdash k$, where either
\begin{enumerate}[\rm (a)]
    \item
$\mu$ is balanced of height $e_0$ and $\nu$ is  balanced of height $d-e_0$;
   \item
$\mu$ is balanced of height $e_\star$ and $\nu$ is  balanced of height $d-e_\star$;
   \item
$\mu$ is balanced of height $e_1$ and $\nu$ is  balanced of height $d-e_1$.
\end{enumerate}
\end{corollary}

\begin{corollary}\label{cor:d=3}
The solution to the 3-QMC problem for $K_{n-k,k}$ is
\[
    \begin{cases}
2 ( k +1) (n-k)
& \text{if }n<3k \\
2 k ( n - k + 2)
& \text{if }n\geq 3k
    \end{cases}.
\]
\end{corollary}

\begin{proof}
Observe that $e_1=2$ for all $n\geq2k$, and that
$e_0=2$ iff $n\geq3k$; otherwise $e_0=1$.
Thus if $n\geq3k$ we are in case (c) of Theorem \ref{th:dbip}.

By Theorem \ref{th:dbip}, there are only two cases to consider:
$\la=\mu\uplus\nu$ with
\[
    \mu=(n-k), \; \nu=\left(\left\lceil \frac k2\right\rceil, \left\lfloor \frac k2\right\rfloor  \right) 
    \quad \text{or} \quad
    \mu=\left(\left\lceil \frac {n-k}2\right\rceil, \left\lfloor \frac {n-k}2\right\rfloor  \right),  \; \nu=(k).
\]
In the first case $\lambda=(\mu,\nu)$ and thus
\begin{equation}\label{eq:e=1}
\begin{split}
    \eta_{\mu\uplus\nu}-\eta_\mu-\eta_\nu & = -2 (k-n-1) \left(\left\lceil \frac{k}{2}\right\rceil +\left\lfloor
   \frac{k}{2}\right\rfloor \right) \\
   & = 2 k (n-k+1).
\end{split}
\end{equation}

Now assume $n<3k$, and consider the second case. Then 
$\mu\uplus\nu=\left(k, \left\lceil \frac {n-k}2\right\rceil, \left\lfloor \frac {n-k}2\right\rfloor  \right)$,
so
\begin{equation}\label{eq:e=2}
\begin{split}
    \eta_{\mu\uplus\nu}-\eta_\mu-\eta_\nu & =-\left(\left\lceil \frac{n-k}{2}\right\rceil +\left\lfloor
   \frac{n-k}{2}\right\rfloor \right)^2+\left(\left\lceil
   \frac{n-k}{2}\right\rceil +\left\lfloor
   \frac{n-k}{2}\right\rfloor +k\right)^2 
    +\left\lceil \frac{n-k}{2}\right\rceil ^2
\\
   & \phantom{{}={}}
   -\left(\left\lceil \frac{n-k}{2}\right\rceil -1\right)^2
   -\left(\left\lfloor
   \frac{n-k}{2}\right\rfloor -2\right)^2+\left(\left\lfloor
   \frac{n-k}{2}\right\rfloor -1\right)^2-k^2+4 \\
   & = -\left(-\left\lfloor \frac{n-k}{2}\right\rfloor
   -k+n-1\right)^2+\left(-\left\lfloor \frac{n-k}{2}\right\rfloor
   -k+n\right)^2
   \\
   & \phantom{{}={}}
   -\left(\left\lfloor \frac{n-k}{2}\right\rfloor-2\right)^2+\left(\left\lfloor \frac{n-k}{2}\right\rfloor
   -1\right)^2-k^2-(n-k)^2+n^2+4 
   \\
   & = 2 ( k+1) (n-k)
\end{split}
\end{equation}
The difference between \eqref{eq:e=2} and \eqref{eq:e=1} is
\[
    2 k (-1 + k - n) - 2 (1 + k) (k - n) = 2 n - 4k \geq 0
\]
since $n\geq2k$. In particular, for $d=3$ we are always in case (c) of Theorem \ref{th:dbip}.

If $n\geq3k$, then 
$\mu\uplus\nu=\left( \left\lceil \frac {n-k}2\right\rceil, \left\lfloor \frac {n-k}2\right\rfloor ,k \right)$, and
a similar calculation to the above shows that 
\begin{equation}\label{eq:e=2again}
    \eta_{\mu\uplus\nu}-\eta_\mu-\eta_\nu  =
   2 k ( n - k + 2),
\end{equation}
concluding the proof.
\end{proof}

\begin{table}[ht]
\begin{minipage}[t]{0.48\textwidth}
\centering
\[
\begin{array}{c|c|c|c|c|c}
(n,k,d) & e_0 & e_1 & e^* & e_{\rm max} & \mathfrak E \\[1mm]
\hline
(4, 2, 3)& 1 & 2& \frac32 & {1, 2} & \{1, 2\}\\[1mm]
(5, 2, 3) &  1 & 2 & \frac74 & {2} & \{2\}\\[1mm]
(5, 2, 4) & 2 & 3 & \frac94 & {2} & \{2, 3\} \\[1mm]
(8, 2, 3) & 2 & 2 & \frac{13}6 & {2} & \{2\} \\[1mm]
(9, 2, 8) & 6 & 7 & \frac{21}4 & {6} & \{6, 7\} \\[1mm]
(6, 3, 4) &  1& 3 & 2 & {2} & \{2\} \\[1mm]
(7, 3, 4) & 2 & 3 & \frac94 & {2} & \{2\} \\[1mm]
(11, 3, 5) &  3 & 4 & \frac{10}3 & {3} & \{4\} \\[1mm]
\end{array}
\]
\end{minipage}
\hfill
\begin{minipage}[t]{0.48\textwidth}
\centering
\[
\begin{array}{c|c|c|c|c|c}
(n,k,d) & e_0 & e_1 & e^* & e_{\rm max} & \mathfrak E \\[1mm]
\hline
(8, 4, 2)& 1 & 1& 1 & {1} & \{1\}\\[1mm]
(8, 4, 3)& 1 & 2& \frac32 & {1,2} & \{\}\\[1mm]
(8, 4, 6)& 2 & 4& 3 & {3} & \{2,3,4\}\\[1mm]
(9, 4, 3)& 1 & 2& \frac{13}{8} & {2} & \{\}\\[1mm]
(9, 4, 4)& 2 & 3& \frac{13}{6} & {2} & \{2\}\\[1mm]
(9, 4, 5)& 2 & 3& \frac{11}{4} & {3} & \{3\}\\[1mm]
(11, 4, 4)& 2 & 3& \frac52 & {2} & \{\}\\[1mm]
\end{array}
\]
\end{minipage}
\caption{A selection of triples and the values of $e_0$, $e_1$, $e^*$, $e_{\rm max}$, and $\mathfrak{E}$ for these triples. The triple $(11,3,5)$ is an example where the optimal $\la=\mu\uplus\nu$ is not balanced. The triples on the right side demonstrate that any assortment of cases (a),(b),(c) in Corollary \ref{cor:dbip} can give rise to the maximum value. 
}
\label{tab:evalues}
\end{table}

\section{Separation of irreps in $d$-QMC}\label{sec:sep}

Given $n,d\in\NN$ and a partition $\lambda\vdash n$ with at most $d$ rows, we are interested in adapting the NPO hierarchy in Section \ref{s:npo} to compute the largest eigenvalue of $H_G^\lambda$ for a general weighted graph $G$ on $n$ vertices.
To isolate the irreducible representation $\rho_\lambda$ of $\cA^{\text{Sw}_d}_n=\cF_n / \mathcal{I}^{\text{Sw}_d}_n$ corresponding to the partition $\lambda$ of $n$, one needs to adjoin to $\mathcal{I}^{\text{Sw}_d}_n$ some polynomials in $\cF_n$ that vanish in $\rho_\lambda(\CC S_n)$, but not in $\rho_\mu(\CC S_n)$ for any $\mu\neq\lambda$. 

In general it suffices a add a single polynomial, chosen as follows.
Given a partition $\lambda\vdash n$ let $s_\lambda\in \CC S_n$ be a Young symmetrizer corresponding to $\lambda$ \cite[Section 9.2.2]{Probook}. Then $\frac{\dim\rho_\lambda}{n!}s_\lambda$ is a primitive idempotent in $\CC S_n$ that generates $\rho_n$ as a left ideal \cite[Theorems 9.2.4.1 and 9.2.4.2]{Probook}. Hence
$$\hat{s}_\lambda=\frac{\dim\rho_\lambda}{(n!)^2}\sum_{\sigma\in S_n} \sigma s_\lambda \sigma^{-1}$$
is a centrally primitive idempotent in $\CC S_n$, generating $\rho_\lambda$ as a two-sided ideal.

\begin{proposition}\label{prop:isolate_irrep}
Let $\lambda\vdash n$.
Then $\rho_\lambda(\id-\hat{s}_\lambda)=0$ and $\rho_\mu(\id-\hat{s}_\lambda)\neq0$ every partition $\mu\neq\lambda$.
\end{proposition}

\begin{proof}
The centrally primitive idempotents $\{\hat{s}_\lambda\}_{\lambda\vdash n}$ are pairwise orthogonal, so $(\id-\hat{s}_\lambda)\hat{s}_\mu=0$ if $\lambda=\mu$ and $\hat{s}_\mu$ otherwise.
Therefore $\rho_\lambda (\id-\hat{s}_\lambda)=0$ and $\rho_\mu(\id-\hat{s}_\lambda)\neq0$ for $\mu\neq\lambda$.
\end{proof}

Lifting $\hat{s}_\lambda$ to an element of $\cF_n$ yields the desired polynomial. However, this polynomial has high degree (not much smaller than $n$). On the other hand, the symmetric group relations \eqref{eq:snrel} have degree at most $3$, and the degree-reducing antisymmetrizer relation \eqref{eq:qdit-rel} has degree $d$. Therefore the above approach is not appealing from a computational perspective.
Instead, it is preferable to find low-degree polynomials that distinguish $\lambda$ from other partitions of $n$ with at most $d$ rows.

In \cite{BCEHK24} it was shown that for $d=2$, the value $\eta_\lambda$ from Example \ref{rem:eta2} separates irreps with at most two rows. 
Therefore, $\eta_\lambda-h_{K_n}$ is a linear polynomial that separates irreps of $\cA^{\text{Sw}_2}_n$.
In particular, the largest eigenvalue of $H_G^\lambda$ for a two-row partition $\lambda\vdash n$ equals the NPO problem
$$
\min\left\{
\alpha:\alpha-h_G=\sum_k s_k^*s_k+q  
\text{ for some } 
s_k\in\cF_n,\ q\in \mathcal{I}^{\text{Sw}_d}_n
+\big(\eta_\lambda-h_{K_n}\big)
\right\},
$$
and can thus be handled using standard SDP-based NPO hierarchies.

The same does not apply when $d=3$, as $\eta_\lambda$ in \eqref{eq:etaD} does not separate irreps with at most three rows. For example, partitions $\lambda = (4,1,1)$  and $\mu = (3,3)$ of $n=6$ give $\eta_{\lambda} = \eta_{\mu}=24$. 
Even more, $\eta_\lambda$ does not separate irreps with three rows; e.g., $\lambda=(5,2,2)$ and $\mu=(4,4,1)$ give $\eta_{\lambda}=\eta_{\mu}=60$.
Below, we present a method of separating irreps with at most three rows in the spirit of $3$-QMC, and a method of separating general irreps that is suitable for solving the localized $d$-QMC problem of finding the largest eigenvalue of $H_G^\lambda$.

\subsection{Separation of irreps with at most three rows via two graphs}

First we show that the spectra of the Hamiltonians corresponding to the clique $K_n$ and the star graph $\starn$ (which were analyzed in Subsections \ref{subsec ev-3} and \ref{ss:star}, respectively) distinguish partitions with at most three rows.

\begin{proposition}\label{p:sep3}
Let $n\ge2$. The following are equivalent for partitions $\lambda,\mu\vdash n$ with at most three rows:
\begin{enumerate}[\rm (i)]
\item $\spec(H_{K_n}^{\lambda})=\spec(H_{K_n}^{\mu})$ and $\spec(H_{\starn}^{\lambda})\subseteq\spec(H_{\starn}^{\mu})$;
\item $\spec(H_{\starn}^{\lambda})=\spec(H_{\starn}^{\mu})$;
\item ${\lambda}={\mu}$.
\end{enumerate}
\end{proposition}

\begin{proof}
It is clear that (iii) implies both (i) and (ii).

(i)$\Rightarrow$(iii): By Example \ref{p:star2} and Example \ref{p:star3}, 
the eigenvalues of $nI-\frac12 H_{\starn}^{\lambda}$ are obtained from the sequence
${\lambda}_1>{\lambda}_2-1>{\lambda}_3-2$ by keeping only the smallest element of any subsequence of consecutive values, and then removing $-2$ if necessary.
Consequently,
\begin{equation}\label{e:samen}
{\lambda}_1+({\lambda}_2-1)+({\lambda}_3-2)
=n-3=
{\mu}_1+({\mu}_2-1)+({\mu}_3-2)
\end{equation}
and
$\spec(H_{\starn}^{\lambda})\subseteq\spec(H_{\starn}^{\mu})$ 
immediately imply $\lambda=\mu$ if $|\spec(H_{\starn}^{\lambda})|\ge 2$. 
Now assume $|\spec(H_{\starn}^{\lambda})|=1$. By Proposition \ref{prop:etaD}, $\spec(H_{K_n}^{\lambda})=\spec(H_{K_n}^{\mu})$ implies
\begin{equation}\label{e:sameeta}
{\lambda}_1^2+({\lambda}_2-1)^2+({\lambda}_3-2)^2
={\mu}_1^2+({\mu}_2-1)^2+({\mu}_3-2)^2
\end{equation}
We distinguish three cases:
\begin{enumerate}[\rm(1)]
\item ${\lambda}=(\frac{n}{3},\frac{n}{3},\frac{n}{3})$, and the sole eigenvalue of $nI-\frac12 H_{\starn}^{\lambda}$ is $\frac{n}{3}-2$. Since ${\mu}_1\ge\frac{n}{3}$, if follows that ${\mu}_2-1=\frac{n}{3}-2$ or ${\mu}_3-2=\frac{n}{3}-2$.
In the latter case $\mu=\lambda$, so let us suppose the former holds. Then \eqref{e:samen} and \eqref{e:sameeta} imply
\begin{align*}
n-3&={\mu}_1+(\tfrac{n}{3}-2)+({\mu}_3-2),\\
(\tfrac{n}{3})^2+(\tfrac{n}{3}-1)^2+(\tfrac{n}{3}-2)^2
&={\mu}_1^2+(\tfrac{n}{3}-2)^2+({\mu}_3-2)^2.
\end{align*}
Expressing ${\mu}_3=\frac{2n}{3}+1-{\mu}_1$ gives
$$(\tfrac{n}{3})^2+(\tfrac{n}{3}-1)^2
={\mu}_1^2+(\tfrac{2n}{3}-{\mu}_1-1)^2,$$
which has solutions ${\mu}_1=\frac{n}{3}$ and ${\mu}_1=\frac{n}{3}-1$. The first one implies ${\mu}={\lambda}$, while the second one contradicts the fact that ${\mu}$ is a partition of $n$.

\item ${\lambda}=(\frac{n}{2},\frac{n}{2})$, and the sole eigenvalue of $nI-\frac12 H_{\starn}^{\lambda}$ is $\frac{n}{2}-1$. Since ${\mu}_3<\frac{n}{2}$, it follows that ${\mu}_1=\frac{n}{2}-1$ or ${\mu}_2-1=\frac{n}{2}-1$.
In the latter case ${\mu}={\lambda}$, so let us assume the former holds.
Then \eqref{e:samen} and \eqref{e:sameeta} imply
\begin{align*}
n-3&=(\tfrac{n}{2}-1)+({\mu}_2-1)+({\mu}_3-2),\\
(\tfrac{n}{2})^2+(\tfrac{n}{2}-1)^2+(0-2)^2
&=(\tfrac{n}{2}-1)^2+({\mu}_2-1)^2+({\mu}_3-2)^2.
\end{align*}
Expressing ${\mu}_2=\tfrac{n}{2}+1-{\mu}_3$ gives
$$(\tfrac{n}{2})^2+4
=(\tfrac{n}{2}-{\mu}_3)^2+({\mu}_3-2)^2,$$
which has solutions ${\mu}_3=0$ and ${\mu}_3=\frac{n}{2}+2$. The first one implies ${\mu}={\lambda}$, while the second one contradicts the fact that ${\mu}$ is a partition of $n$.

\item ${\lambda}=(n)$, and the sole eigenvalue of $nI-\frac12 H_{\starn}^{\lambda}$ is $n$. Then ${\mu}_3\le{\mu}_2\le n$ implies ${\mu}_1=n$, and so ${\mu}={\lambda}$.
\end{enumerate}

(ii)$\Rightarrow$(iii): As in the previous paragraph we see that (ii) and \eqref{e:samen} imply ${\lambda}={\mu}$ if $|\spec(H_{\starn}^{\lambda})|\ge2$. 
On the other hand, if $|\spec(H_{\starn}^{\lambda})|=|\spec(H_{\starn}^{\mu})|=1$ then ${\lambda},{\mu}\in\{(n),(\frac{n}{2},\frac{n}{2}),(\frac{n}{3},\frac{n}{3},\frac{n}{3})\}$. For these three cases, the star graph Hamiltonian has eigenvalue $0$, $n+2$ or $\frac43 n+4$. Thus $\spec(H_{\starn}^{\lambda})=\spec(H_{\starn}^{\mu})$ only if ${\lambda}={\mu}$.
\end{proof}

\begin{remark}
Let $n=9$. For ${\lambda}=(3,3,3)$ and ${\mu}=(6,2,1)$ we have
$$\spec(H_{\starn}^{\lambda})
=\{16\}\subseteq\{20,16,6\}=
\spec(H_{\starn}^{\mu}).$$
Therefore the role of $K_n$ in Proposition \ref{p:sep3}(i) is essential (note that $\eta_{\lambda}=72$ and $\eta_{\mu}=48$).
Likewise, the restriction to partitions with at most three rows is required (cf. Remark \ref{r:dist}).
Namely, let $n=21$, ${\lambda}=(7,7,7)$ and $\mu=(9,6,5,1)$. Then $\eta_\lambda=336=\eta_\mu$ and
$$\spec(H_{\starn}^{\lambda})
=\{16\}\subseteq\{12, 16, 18, 23\}=
\spec(H_{\starn}^\mu).$$
\end{remark}

Let $\lambda\vdash n$ be a three-row partition, and let $m$ be the minimal polynomial of $H_{\starn}^\lambda$; note that $m$ is of degree at most 3, and determined by Example \ref{p:star3}.
As a consequence of Proposition \ref{p:sep3}, the largest eigenvalue of $H_G^\lambda$ for $d=3$ is the solution of the NPO problem
$$
\min\left\{
\alpha:\alpha-h_G=\sum_k s_k^*s_k+q  
\text{ for some } 
s_k\in\cF_n,\ q\in \mathcal{I}^{\text{Sw}_3}_n
+\big(\eta_\lambda-h_{K_n},m(h_{\starn})\big)
\right\}.
$$

\subsection{Separation of irreps via low-degree central elements}

In this section we show that partitions $\lambda\vdash n$ with at most $d$ rows can be distinguished by $d$ relations of degrees $1,\dots,d$, which can be used in an NPO problem for solving the localized $d$-QMC problem, i.e., finding the largest eigenvalue of $H_G^\lambda$.

For $2\le k\le n$ let $q_k\in \CC[S_n]$ be the sum of all $k$-cycles in $S_n$ (there are $(k-1)!\binom{n}{k}$ of them).
Since $q_k$ is central in $\CC[S_n]$, we have
$\rho_\lambda(q_k)=\gamma_{k,\lambda}I$,
where
\begin{equation}\label{e:gamma}
\gamma_{k,\lambda}\,\chi_\lambda(e) 
= \text{Tr}\big( \rho_\lambda(q_k) \big)=
(k-1)!\binom{n}{k}\,\chi_\lambda((1\,\dots\,k)).
\end{equation}
Note that
$$\gamma_{2,\lambda}=\eta_\lambda=
n^2 +\frac{d(d-1)(2d-1)}{6}
-\sum_{k=1}^d\big( \lambda_k - (k-1)\big)^2
$$
by Proposition \ref{prop:etaD}.
The other values $\gamma_{k,\lambda}$ can be computed using the normalized character formula \cite[Theorem 4]{Lassalle} (with the Murnaghan-Nakayama rule, cf. Appendix \ref{sec:MNR}, at its core), and are in particular integers.
For example,
\begin{align*}
\gamma_{3,\lambda}
&=\frac13\cdot n(n-1)(n-2)\,
\frac{\chi_\lambda((1\,2\,3))}{\chi_\lambda(e)}\\
&=\sum_{k=1}^d\sum_{j=1}^{\lambda_k}(j-k)^2-\binom{n}{2}\\
&=-\binom{n}{2}+\sum_{k=1}^d\left(\frac{(\lambda_k-k)(\lambda_k-k+1)(2(\lambda_k-k)+1)}{6}
+\frac{(k-1)k(2k-1)}{6}
\right)\\
&=\frac{d(d-1)^2(d-2)}{2}-\binom{n}{2}+\sum_{k=1}^d\frac{(\lambda_k-k)(\lambda_k-k+1)(2(\lambda_k-k)+1)}{6}
\end{align*}
using the formula after \cite[Theorem 4]{Lassalle}.
The values $\gamma_{k,\lambda}$ separate irreps as follows.

\begin{theorem}\label{t:sep_all_rows}
If 
$\lambda,\mu\vdash n$ have at most $d$ rows, then
$$\lambda=\mu
\quad\iff\quad \gamma_{k,\lambda}=\gamma_{k,\mu} \text{ for all }k=2,\dots,d.$$
\end{theorem}

\begin{proof}
Let $p_{d,k}=x_1^k+\cdots+x_d^k$ denote the $k$th power-sum symmetric polynomial in $d$ variables. 
Let $\lambda\vdash n$ have at most $d$ rows, and write $\lambda_i=0$ for $\het(\lambda)<i\le d$.
By \eqref{e:gamma} we have
$$k \gamma_{k,\lambda}
=n^{\downarrow k}\,
\frac{\chi_\lambda((1\,\dots\,k))}{\chi_\lambda(e)},$$
where $n^{\downarrow k}$ is the falling factorial.
By \cite[Lemma 5.1]{VK} or \cite[Propositions 1.4, 3.3 and 3.4]{IO},
\begin{equation}\label{e:symmfun}
k\gamma_{k,\lambda} = 
\Big(
p_{d,k}+P_k\big(p_{d,1},\dots,p_{d,k-1}\big)
\Big)
\big(
\lambda_1-1+\tfrac12,\lambda_2-2+\tfrac12,\dots,\lambda_d-d+\tfrac12
\big)
\end{equation}
for some polynomial $P_k$ in $k-1$ variables. 
Also note that
$$p_{d,1}\big(
\lambda_1-1+\tfrac12,\lambda_2-2+\tfrac12,\dots,\lambda_d-d+\tfrac12
\big)=n-\tfrac{d^2}{2}.$$
Now assume that $\gamma_{k,\lambda}=\gamma_{k,\mu}$ for all $k=2,\dots,d$. By \eqref{e:symmfun}, 
$$p_{d,k}\big(
\lambda_1-1+\tfrac12,\lambda_2-2+\tfrac12,\dots,\lambda_d-d+\tfrac12
\big)
=p_{d,k}\big(
\mu_1-1+\tfrac12,\mu_2-2+\tfrac12,\dots,\mu_d-d+\tfrac12
\big) \qquad 
$$
for all $k=1,\dots,d$. Since 
$$\lambda_1-1+\tfrac12>\lambda_2-2+\tfrac12>\cdots>\lambda_d-d+\tfrac12,\quad
\mu_1-1+\tfrac12>\mu_2-2+\tfrac12>\cdots>\mu_d-d+\tfrac12
$$
and the power-sum symmetric polynomials distinguish points up to a coordinate shuffle,
it follows that $\lambda=\mu$.
\end{proof}
For $k\in\N$ denote
\begin{equation}\label{e:ck}
c_k=\sum_{\substack{
1\le i_0,\dots,i_k \le d\\[.5mm]
\text{pairwise distinct},\\[.5mm]
i_0<i_j\text{ for }j\ge1
}}
\text{swap}_{i_0i_1}\text{swap}_{i_0i_2}\cdots\text{swap}_{i_0i_k}\in \cF_n
\end{equation}
which corresponds to $q_{k+1}\in\CC[S_n]$.
By Theorem \ref{t:sep_all_rows}, finding the largest eigenvalue of the localized $d$-QMC Hamiltonian $H_G^\lambda$ (for $\het(\lambda)\le d$) is equivalent to the NPO problem
$$
\min\left\{
\alpha:\alpha-h_G=\sum_k s_k^*s_k+q  
\text{ for some } 
s_k\in\cF_n,\ q\in \mathcal{I}^{\text{Sw}_d}_n
+\big(c_k-\gamma_{k+1,\lambda}\colon k\le d-1\big)
\right\}.
$$
As in Section \ref{s:npo}, this NPO can be solved through a hierarchy of SDP relaxations.

\subsection*{Author contributions}

All authors contributed extensively to the development and the writing of this manuscript.

\setstretch{1.05}

\setstretch{1.0625}

\newpage
 \appendix

\section{Linear subspace of $M_{d^n}(\mathbb{C})$ spanned by the products of at most $d-1$ swap matrices}\label{a:d-1}
 Here we prove that in $M^{\text{Sw}_d}_n(\mathbb{C}),$ there are no relations of order at most $d-1$ (in the swap matrices, represented by transpositions) other than \eqref{eq:from sn}.\looseness=-1

\begin{remark}
Note that if a permutation $\sigma$ is a product of disjoint cycles of lengths $\ell_1,\dots,\ell_k$ respectively,  then $\sigma$ is can be written as a product of $\sum_{i=1}^k(\ell_i-1)$ transpositions (and not fewer than that many transpositions).
\end{remark}

First we note that linear independence of swap operators is preserved if the local dimension increases.

\begin{lemma}\label{cor:lin-indep}
If a set of products of swap operators is linearly independent in $M^{\text{\textnormal{Sw}}_d}_n(\mathbb{C})$, then it is linearly independent in $M^{\text{\textnormal{Sw}}_{d+1}}_n(\mathbb{C})$.
\end{lemma}

\begin{proof}
By Theorem \ref{th: s-w-swaps}, the algebra $M^{\text{\textnormal{Sw}}_d}_n(\mathbb{C})$ can be obtained as a quotient of $M^{\text{\textnormal{Sw}}_{d+1}}_n(\mathbb{C}),$ namely by modding out the direct summands $\rho_{\lambda} (\mathbb{C}S_n),$ where $\lambda$ is a partition of $n$ with exactly $d$ rows. 
The statement then follows since any linearly independent set in a quotient $M^{\text{\textnormal{Sw}}_{d+1}}_n(\mathbb{C})$ is linearly independent in $M^{\text{\textnormal{Sw}}_{d+1}}_n(\mathbb{C}).$
\end{proof}

The following statement is the main result of this section.

\begin{proposition}\label{prop:d-1-indep}
	The set of all products (that correspond to distinct permutations) of at most $d-1$ swap matrices 
	is a basis of the subspace of $M^{\text{Sw}_d}_n(\mathbb{C})$ of polynomials in the Swap$_{ij}$ of degree at most $d-1.$ 
\end{proposition}

\begin{remark}
In other words, Proposition \ref{prop:d-1-indep} states that permutations, which are products of at most $d-1$ transpositions, are linearly independent as elements of $M^{\text{Sw}_d}_n(\mathbb{C})$.
In Section \ref{s:npo}, we mentioned another natural linearly independent subset of $M^{\text{Sw}_d}_n(\mathbb{C})$. Recall that a permutation $\pi\in S_n$ is called \textit{$(d+1)$-good} if there is no increasing sequence $j_0<\cdots<j_d$ such that $\pi(j_0)>\cdots>\pi_(j_d)$. Then $(d+1)$-good permutations form a basis of $M^{\text{Sw}_d}_n(\CC)$ by \cite[Theorem 8]{Pro21}.
However, Proposition \ref{prop:d-1-indep} is not a direct consequence of this result. Namely, a product of at most $d-1$ transpositions is not necessarily a $(d+1)$-good permutation if $d\ge 3$. Concretely, the product of $d-1$ disjoint transpositions $\pi=\prod_{i=1}^{d-1}(i,2d-1-i)$ satisfies $\pi(1)>\pi(2)>\cdots>\pi(2d-2)$, so it is not $2(d-1)$-good (and in particular, not $(d+1)$-good if $d\ge3$).
\end{remark}

Before proving Proposition \ref{prop:d-1-indep}, we require two lemmas.
    
\begin{lemma}\label{lemma:uniq} 
	Let $\sigma$ be a permutation in $S_n$ that is a product of $d-1$ transpositions and cannot be written as a product of fewer than $d-1$ transpositions.
	   Let $v \in (\mathbb{C}^d)^{\otimes n}$ be an elementary tensor whose factors are standard basis vectors (so that $v$ has at most $d$ distinct indices). Suppose that for any product $\tau$ of $k$ disjoint cycles of $\sigma,$ where $k=1,2,$ the part of $v$ on which $\tau$ acts has at most $k-1$ indices that are repeated and they occur at most twice. Then $\sigma$ acts uniquely on $v$ among the products of at most $d-1$ transpositions. \looseness=-1
\end{lemma}	

\begin{proof}
	Let $\sigma$ and $v$ be as in Lemma \ref{lemma:uniq}. If $\sigma$ is a permutation on strictly less than $n$ letters, add to its cyclic structure the singletons corresponding to the missing letters in $\{1,\ldots,n\}.$

	First suppose $\tau$ is one of the disjoint cycles of $\sigma$ and let $w$ be the part of $v$ on which $\tau$ acts. 
	If $w$ has an index that appears at least twice, then one can construct at least one other permutation $\tau^\prime$ that is a product of at most as many transpositions as $\tau$, and gives the same result when applied to $w.$ Indeed, to find the first cycle of $\tau^\prime,$ start with the index of $w$ that is repeated, then find its image among the remaining indices, take the image of the latter and continue until the starting index occurs again. Since the starting index occurs at least twice in $w,$ the produced cycle is of smaller length compared to $\sigma.$ Now repeat the procedure starting with any other index to deduce what the other disjoint cycles are. This way we break the action of $\sigma$ on $w$ into an action of a product of disjoint cycles (some of them may have length one) whose lengths sum up to the length of $\sigma.$ Hence their product can be written as a product
	of strictly less than $d-1$ transpositions.
	
	Now let $\tau = \tau_1 \tau_2$ be a product of two disjoint cycles of $\sigma$ and let  $w_1$ and $w_2$ be the parts of $v$ on which $\tau_1$ and $\tau_2$ act, respectively. Suppose that none of $w_1$ and $w_2$ has repeated indices, but they do share at least two indices. For simplicity suppose they share exactly two indices. We want to construct another permutation $\tau^\prime$ that is a product of at most as many transpositions as $\tau$, and gives the same result when applied to $w_1\otimes w_2.$ As before, for the first cycle of $\tau^\prime$ start with one of the indices with two occurrences in $w_1,$ say $j_1,$ then find its image among the remaining indices, take the image of the latter and continue until $j_1$ occurs again. Note that the first time when the image of a letter is the other index with two occurrences, call it $j_2,$ there are two choices: to consider the image of $j_2$ by either $\tau_1$ or $\tau_2$ and then continue the process until the image of a letter is $j_1$ again. Note that if we choose to continue with $\tau_1(j_2),$ we get back $\tau,$ but if we consider
	$\tau_2(j_2),$ we switch to the other cycle and finish with the factor $e_{j_1}$ of $w_2.$ Hence, the second option produces the first cycle of the permutation $\tau^\prime$ we are looking for. For the second cycle of $\tau^\prime$ restart the procedure with the index $j_2$ of a factor in $w_1.$ By construction, $\tau^\prime$ is also a product of at most $d-1$ transpositions.

	It remains to prove that if $\sigma$ meets the conditions in Lemma \ref{lemma:uniq}, then it acts on $v$ uniquely among the products of at most $d-1$ transpositions. This is true by the same procedure as above of deducing the cyclic structure by comparing $v$ to its image $\sigma(v).$ Indeed, start with any index, take its image and continue until the starting index occurs again.
	
	The only time we have two options during this process is when we hit an index $j$ in $v$ that occurs (exactly) twice across two cycles, say $\tau_1$ and $\tau_2.$ Denote the parts of $v$ on which $\tau_i$ acts by $w_i.$ Suppose we started the process with $j$ in $\tau_1.$ When we hit $j$ again, we can either terminate the process (which yields the cycle $\tau_1$) or continue with the image of $j$ by $\tau_2.$ In this case we join the two cycles $\tau_1$ and $\tau_2,$ which means that the resulting permutation must have at least one  transposition more than $\sigma.$ 
	
	If we start with an index $k \neq j$ of $\tau_1,$ then, when we first hit $j,$ there are again two options: either to continue with $\tau_1(j)$ or $\tau_2(j).$ The choice $\tau_1(j)$ at the end reproduces $\tau_1$ (actually, it may happen that the index $k$ occurs in another cycle, say $\tau_3,$ and we may switch to $\tau_3$ after hitting $k$ again, but this case was already treated before). With the choice $\tau_2(j)$ we switch to the other cycle $\tau_2$ and since $k$ does not occur in $w_2,$ the process does not terminate in $w_2$ (meaning that we need to switch cycle once more before terminating). This again means that we join (at least) two cycles and the resulting permutation must have at least one  transposition more than $\sigma.$

	This shows that $\sigma$ is the only permutation that can be written as a product of at most $d-1$ transpositions and gives the result $\sigma(v)$ when applied to $v.$
\end{proof}

\begin{example}
	Let $d=4,n=5$ and define $\sigma = (12)(345).$ Let
	$$
	v_1 = e_1 \otimes e_2 \otimes e_1 \otimes e_2 \otimes e_3 \quad \text{ and }\quad 
	v_2 = e_1 \otimes e_2 \otimes e_1 \otimes e_3 \otimes e_4.
	$$
	It is easy to see that since $\sigma$ has two disjoint cycles and $v_1$ has two indices that occur twice, $\sigma$ acts on $v_1$ in the same way as $\sigma^\prime = (14)(253).$ On the other hand, by the algorithm of deducing the cyclic structure by comparing a vector to its image, $\sigma$ acts on $v_2$ uniquely among the products of at most $3$ transpositions in $S_5.$
\end{example}

\begin{lemma}\label{lemma:v}
	Let $\sigma$ in $S_n$ be a product of $d-1$ transpositions that cannot be written as a product of fewer than $d-1$ transpositions. Then there is a vector $v \in (\mathbb{C}^d)^{\otimes n}$ whose tensor factors are standard basis vectors with at most $d$ distinct indices that meets the conditions in Lemma \ref{lemma:uniq}.
\end{lemma}

\begin{proof}
	Let $\sigma \in S_n$ be a product of $d-1$ transpositions which cannot be written as a product of less than $d-1$ transpositions. Suppose $\sigma$ is a product of $k$ disjoint cycles for some $k=1,\ldots,d-1$ (all the cycles being of length $2$ or more). Hence, $\sigma$ is a permutation on $d+k-1$ letters, but the factors of $v$ are chosen among the $d$ standard basis vectors of $\CC^d.$ 
	
Without loss of generality assume that the letters of $\sigma$ are $1,\ldots,d+k-1.$	
Order the cycles of $\sigma$ increasingly by their lengths.
Then assign the indices $i_j$  to the first $d+k-1$ factors $e_{i_j}$ of $v$ in the following way: 

Assign indices $1,\ldots,d$ (e.g., increasingly according to the position of the factors) to the part of $v$ on which the cycles of $\sigma$ involving letters $1,\ldots,d$ act (the cycle with $d$ may involve larger indices and hence we do not yet assign the indices to all of the factors of $v$ on which this cycle acts). Now $v$ has $k-1$ more factors to be labeled (with indices between $1$ and $d$), hence the corresponding part of $\sigma$ has at most $\lfloor (k-1)/2 \rfloor$ disjoint cycles. Since we have $k$ cycles in total (each of length at least $2$), the part of $\sigma$ on the letters $1,\ldots,d$ is a product of at least $\lfloor (k+1)/2 \rfloor$ disjoint cycles. 

So assign to the next $\lfloor (k-1)/2 \rfloor$ unlabeled factors of $v$  (e.g., increasingly according to the position of the factors) the first letters of the cycles of $\sigma$ on the letters $1,\ldots,d.$ Finally, assign to the remaining unlabeled factors of $v$ the second letters of the cycles of $\sigma$ on the letters $1,\ldots,d.$

This way we labeled the first $d+k-1$ factors of $v.$ To the remaining factors just assign the index $1.$
It is now clear from the construction that the obtained vector meets the conditions in Lemma \ref{lemma:uniq}. 
\end{proof}

\begin{example}
Let us illustrate Lemmas \ref{lemma:uniq} and \ref{lemma:v} in the case $n=8$ and $d=5.$
	Suppose $\sigma \in S_8$ is a product of $4$ transpositions and it cannot be written as a product of less than $4$ transpositions (here we omit writing the singletons in $\sigma$). Then we have $4$ options:
	\begin{enumerate}[\rm (a)]
		\item If $\sigma$ is a $5$-cycle, e.g., $\sigma = (12345),$ then a suitable vector is
		$$
		v = e_1 \otimes e_2 \otimes e_3 \otimes e_4 \otimes e_5 \otimes e_1 \otimes e_1 \otimes e_1.
		$$
		\item If $\sigma$ has two cycles, there are two possible cyclic structures: two $3$-cycles or a product of a transposition and a $4$-cycle. E.g., $\sigma = (123)(456)$ or $\sigma = (12)(3456).$ In both cases we can take
		$$
		v = e_1 \otimes e_2 \otimes e_3 \otimes e_4 \otimes e_5 \otimes e_1 \otimes e_1 \otimes e_1.
		$$
		\item If $\sigma$ has $3$ cycles, then it must be a product of two transpositions and a $3$-cycle. E.g., if $\sigma = (12)(34)(567),$ we can take
		$$
		v = e_1 \otimes e_2 \otimes e_3 \otimes e_4 \otimes e_5 \otimes e_1 \otimes e_3 \otimes e_1.
		$$ 
		\item If $\sigma$ has $4$ cycles, then it must be a product of four transpositions. If, e.g., $\sigma = (12)(34)(56)(78),$ we can take
		$$
		v = e_1 \otimes e_2 \otimes e_3 \otimes e_4 \otimes e_5 \otimes e_1 \otimes e_2 \otimes e_4.
		$$ 
	\end{enumerate}
\end{example}

\begin{proof}[Proof of Proposition \ref{prop:d-1-indep}] 
Denote the set of all products (that correspond to distinct permutations) of at most $d-1$ swap matrices by $\tilde{\mathcal{B}}_{d-1}.$ Suppose
$$\sum_{s\in \tilde{\mathcal{B}}_{d-1}}\alpha_s s=0$$
for some scalars $\alpha_s$. By Lemma \ref{lemma:uniq} and Lemma \ref{lemma:v}, for each product $s$ of $d-1$ transpositions that cannot be written as a product of less than $d-1$ transpositions there is an elementary tensor vector $v_s \in (\mathbb{C}^d)^{\otimes n}$ such that $s$ acts uniquely on $v_s$ among the elements of $\tilde{\mathcal{B}}_{d-1}.$ 
Since elements of $\tilde{\mathcal{B}}_{d-1}$ act on $(\mathbb{C}^d)^{\otimes n}$ as permutations of tensor factors, this means that $s\cdot v_s$ is linearly independent of $\{t\cdot v_s\colon t\in \tilde{\mathcal{B}}_{d-1}\setminus\{s\}\}$.
Hence, $\alpha_s=0$ for all $s\in\tilde{\mathcal{B}}_{d-1}$ that cannot be written as products of less than $d-1$ transpositions. Now use induction and Lemma \ref{cor:lin-indep} to finish the proof.
\end{proof}

\section{Proof of Proposition \ref{prop:02}}\label{app:LR}

\def\SoC{\Sigma}

Recall that for any partition $\zeta = (\zeta_1, \dots, \zeta_d)$ of an integer $m$, we define the function $\eta(\zeta)$ as:
\[
\eta(\zeta) = m^2 + \frac{d(d-1)(2d-1)}{6} - \sum_{i=1}^d (\zeta_i - (i-1))^2.
\]
Our first goal in this section is to derive a formula for
\begin{equation}\label{eq:lmn1}
\Delta(\lambda,\mu,\nu) = \eta(\lambda) - \eta(\mu) - \eta(\nu)
\end{equation}
in terms of the contents of the boxes of the skew-shaped Young diagrams associated with these partitions. The \textit{content of a box} at row $r$ and column $c$ in a skew-shaped Young diagram $\zeta$ is defined as $\text{content}(\text{box}) = c - r$ (not to be confused with the content of a Young tableau as in Section \ref{sec:LR}). Let $\SoC(\zeta)$ denote the sum of contents of boxes in the skew-shaped Young diagram $\zeta$ (or the Young diagram of $\zeta$, if the latter is a partition).

\begin{lemma}\label{l:boxcontent}
Let $\la \vdash n$, $\mu \vdash n-k$ and $\nu \vdash k$. If $\mu$ is contained in $\la$, then
\[
\Delta(\lambda,\mu,\nu) = 2k(n-k) - 2\SoC(\lambda/\mu) +2\SoC(\nu).
\]
\end{lemma}

\begin{proof}
By \cite[Theorem 4]{Lassalle} or \cite{Fro01}, contents of boxes in a Young diagram of $\zeta\vdash m$ are related to the character of $\zeta$ as $\frac{\chi_\zeta((ij))}{\chi_\zeta(e)}=\binom{n}{2}^{-1} \SoC(\zeta)$. By Lemma \ref{l:char}, 
\begin{equation} \label{eq:eta_simplified}
\eta(\zeta) = m^2 - m - 2 \SoC(\zeta).
\end{equation}
We replace $\eta(\lambda),\eta(\mu),\eta(\nu)$ in $\Delta(\la,\mu,\nu)$ with \eqref{eq:eta_simplified},
and note that $\SoC(\la/\mu)=\SoC(\lambda)-\SoC(\mu)$ since $\mu$ is contained in $\lambda$.
\end{proof}

The rest of this appendix consists of the proof of Proposition \ref{prop:02} (stated below for convenience) divided in several cases.

\begin{proposition}\label{prop:02A}
When $k\leq 4$ or $d\le3$, the expression \eqref{eq:lmn} is maximized at a triple of partitions $\lambda \vdash n$,  $\mu \vdash n-k$ and $\nu \vdash k$ such that $\lambda=\mu\uplus\nu$.
For such partitions, the coefficient $c^\lambda_{\mu \nu}$ is nonzero.
\end{proposition}

We shall also paraphrase $\lambda=\mu\uplus\nu$ by saying that each row of $\lambda$ coincides either with a row of $\mu$ or a row of $\nu$. 
For such a triple $(\lambda,\mu,\nu)$, we first check that $c^\lambda_{\mu \nu}\neq0$. 
Since $c^\lambda_{\mu \nu}=c^\lambda_{\nu\mu}$, we can without loss of generality assume that the first row of $\lambda$ equals the first row of $\mu$. 
Then $c^\lambda_{\mu\nu}=c^{\lambda'}_{\mu'\nu}$, where $\lambda',\mu'$ are obtained from $\lambda,\mu$ by deleting the first row. In the triple $(\lambda',\mu',\nu)$, each row of $\lambda'$ coincides either with a row of $\mu'$ or a row of $\nu$. Thus, we can continue inductively until one of $\mu$ or $\nu$ is the empty partition, and the other one equals $\lambda$ (in which case the corresponding Littlewood-Richardson coefficient is nonzero).

Such triples in particular arise as follows. Let $\mu$ be contained in $\lambda$. If $\mu$ and $\lambda/\mu$ do not share any rows, then $\mu$ is a subpartition of $\lambda$ and $\lambda/\mu$ is a partition, and $\lambda=\mu\uplus\lambda/\mu=(\mu,\lambda/\mu)$. The relation ``$\mu$ and $\lambda/\mu$ do not share any rows'' will thus frequently appear in the proof of Proposition \ref{prop:02A} below.

The proof of the first (main) statement of Proposition \ref{prop:02A} relies on particular joint rearrangements of partitions $\lambda,\mu,\nu$ that arise from moving a single box.  
Assume $c_{\mu\nu}^\lambda\neq0$ (so in particular, $\mu$ is contained in $\la$). 
That is, one can label the skew-shaped diagram $\la/\mu$ as an LR tableau with content $\nu$. 
A move of a box of $\la / \mu$ is called a \textbf{Robin Hood move} if the change in the box's content is non-decreasing and, after the move, there exists an LR tableau on the obtained skew-shaped diagram $\lambda'/\mu'$ with content $\nu'$ whose height is not larger than the height of $\nu$.
In particular, a Robin Hood move returns a triple $(\lambda',\mu',\nu')$ with $c^{\lambda'}_{\mu'\nu'}\neq0$, and does not decrease the $\Delta$ value by Lemma \ref{l:boxcontent}: 
$\Delta(\lambda,\mu,\nu)\le \Delta(\lambda',\mu',\nu')$.
Furthermore, it suffices to only consider moving the boxes of $\la/\mu$ with minimal contents (as these minimize the increment of $\SoC(\la/\mu)$).

\subsection{Case $k=2$}

The partitions of $2$ satisfy $\eta_{(2)} < \eta_{(1,1)}$. 
To prove Proposition \ref{prop:02A} we separate three cases, based on $d-e\in\{0,1,2\}$.

\subsubsection{Case $d-e=2$}

Here, $\la_{d-1},\la_d$ equal 1 and form $\nu=(1,1)$. Hence, $\mu$ and $\la / \mu$ do not share rows.

\subsubsection{Case $d-e=1$}

(a) If $e = d-1$ and $\la_d=2,$ then $\la$ and $\la / \mu$ do not share rows.

(b-1) If $\la_d=1$ and $\mu_{d-1} = 2,$ then there is one $\la /\mu$ box in row $d.$ Moving the other $\la / \mu$ box next to it (is a Robin Hood move and hence) yields a higher $\Delta$ value.

(b-2)  If $\la_d=1$ and $\mu_{d-1} = 1,$ let $j_0$ be the smallest row index $j$ such that $\mu_j =1.$ Then move the $\la /\mu$ box that is not in row $d$ to row $j_0.$ This is a Robin Hood move by definition of $j_0,$ hence $\Delta$ increases. By exchanging $\la_{d-1}$ with $\mu_{j_0},$ we produce a triple of partitions with the same $\Delta$ value such that $\la / \mu$ and $\mu$ do not share rows.

\subsubsection{Case $d=e$}
(a) If $\mu_{d-1}\geq 2,$ then move the the $\mu$ boxes from row $d$ to the first row and place the two $\la / \mu$ boxes in row $d$ (Robin Hood moves). This clearly increases the value  of $\Delta.$

(b) If $\mu_{d-1}=1,$ then $\mu_d=1.$ Move both $\mu_{d-1}$ and $\mu_d$ to the first row and place the $\la / \mu$ boxes in row $d-1$ and $d,$ respectively.
Now the contents of these two boxes are lowered, but the height of $\nu$ may increase. We shall prove that the value of $\Delta$ increases in any case.

Assume $\nu = (2)$ at the start. So, after the moves, $\SoC(\nu)$ decreases by $1.$
Denote the two $\la/\mu$ boxes by b1 and b.
Since the initial content of b1 is $\geq 2-(d-2)=4-d,$ $\SoC(\la/\mu)$ increases by $>1$ after the move, causing
the increase in the value of $\Delta.$ 
After the move, $\nu$ is the tail of $\la,$ so $\la = \mu \uplus \nu.$

\subsection{Case $k=3$}

The partitions of 3 satisfy $\eta_{(3)} < \eta_{(2,1)} < \eta_{(1,1,1)}$. The proof of Proposition \ref{prop:02A} again splits in several cases.

\subsubsection{Case $d-e=3$}
The tail of $\la$ equals $\nu = (1,1,1).$ In this case $\mu$ and $\la / \mu$ do not share any rows and we are done.

\subsubsection{Case $d-e=2$}
The three boxes of $\la / \mu$ are not stacked on top of each other. So a minimal enumeration of $\la / \mu$ does not include $3$.
Since the goal is to maximize $\Delta$ as in \eqref{eq:lmn1}, we can assume $\nu \neq (1,1,1)$ (if we can increase the value of $\Delta(\la,\mu,(2,1)),$ we can also increase the value of $\Delta(\la,\mu,(1,1,1))$).
In this case, the minimal enumeration of $\la / \mu$ includes $2$, so we can assume $\nu = (2,1).$

(a) If the tail of $\la$ equals $\nu,$ i.e., $\lambda_{d-1} = 2, \lambda_d=1,$ then $\la$ and $\la / \mu$ do not share any rows and we are done.

(b) If $\nu$ is not the tail of $\la,$ then two $\la / \mu$ boxes are in the $(d-1)$st and $d$th row, respectively, and the third is either in row $d-1$ (in which case $\lambda=\mu\uplus\nu$ and we are done) or at the end of some $\mu$-row. Assume the latter.

(b-1) If $\mu_2\geq 2,$ then the third box is moved to the $(d-1)$st row. In this case $\mu$ and $\nu$ stay the same, but $\la$ is changed via a Robin Hood move. Hence the value of $\Delta$ increases. After the move, $\lambda=\mu\uplus\nu$ and we are done.

(b-2) If $\mu_e=1,$  let $j$ be the largest row index such that $\mu_j>1.$ Move the third box of $\la/\mu$ at the end of row $\mu_{j+1}.$ This move is clearly a Robin Hood move by definition of $j.$ Since $\mu$ and $\nu$ do not change, the value of $\Delta$ increases. Now by exchanging $\la_{d-1}$ with $\mu_{j+1},$ we see that each row of $\la$ is either a row of $\mu$ or a row of $\nu,$ i.e., $\lambda=\mu\uplus\nu.$

\subsubsection{Case $d-e=1$}

(a) If $\mu_{d-1}\geq 3,$ then move the $2$ boxes (that are potentially not yet in the $d$th row) of $\la / \mu$ to the $d$th row. Since this process only involves Robin Hood moves and $\mu$ does not change, the value of $\Delta$ increases. After this move, clearly, $\lambda=\mu\uplus\nu.$

(b) If $\mu_{d-1}=2,$ then move to the $d$th row one of the two boxes of $\la/\mu$ that are not in the $d$th row (this is a Robin Hood move, hence $\Delta$ increases). 

Similarly as in (b-2) in the previous subsection, let $j$ be the largest row index such that $\mu_j>2.$ Move the third box of $\la/\mu$ at the end of row $\mu_{j+1}.$ This move is again a Robin Hood move by definition of $j.$ Since $\mu$ and $\nu$ do not change, the value of $\Delta$ increases. By exchanging $\la_{d-1}$ with $\mu_{j+1},$ we see that each row of $\la$ is either a row of $\mu$ or a row of $\nu.$ Thus, $\lambda=\mu\uplus\nu.$

(c) Now assume $\mu_{d-1}=1.$

(c-1.1) If $\lambda_{d-1} = \mu_{d-1}=1$ and $\mu_{d-2}>1,$ then exchange the $\mu$ box in row $d-1$ with one $\la/\mu$ box in row $d-2.$ To compute tha change in the value of $\Delta$ assume the worst case where $\nu=(3).$  After the exchange, $\SoC(\nu)$ decreases by $3$ and $\SoC(\la/\mu)$ increases by $2.$
Now move the remaining $\la/\mu$ box that is not in row $d-1$ or row $d$ to the $(d-1)$st row. This is clearly a Robin Hood move, so the value of $\Delta $ increases by at least $2$ after this move. The value of $\Delta$ hence increases after these two moves and at the end, $\nu$ is the tail of $\la,$ i.e., $\lambda=\mu\uplus\nu.$

(c-1.2) If $\lambda_{d-1} = \mu_{d-1}=\mu_{d-2}=\la_{d-2}=1,$ then move the $\mu$ boxes from the $(d-1)$st and $(d-2)$nd row to the first and the second row, respectively and move the two $\la/\mu$ boxes that are not in the $d$th row to the $(d-1)$st and $(d-2)$nd row, respectively. 
Again, assume the worst case, where $\nu = (3)$ at the start.
Note that after the move, $\nu=(1,1,1).$  So, after the move, $\SoC(\nu)$ decreases by $6.$
However, by assumption, the contents of the two $\la/\mu$ boxes  that are not in the first column (call them b2 and b3), satisfy content(b2) $\geq 2-d+3$ and content(b3) $\geq 3-d+3.$
Hence, $\SoC(\la/\mu)$ increases by at least $6$ after the move and hence the value of $\Delta$ does not decrease.
Moreover, after the move, $\nu$ is the tail of $\la,$ so $\lambda=\mu\uplus\nu.$

(c-1.3)  If $\lambda_{d-1} = \mu_{d-1}=\mu_{d-2}=1,$ but $\la_{d-2}>1,$ then first exchange the $\mu$ box in row $d-1$ with a $\la/\mu$ box  in row $d-2.$
Again, assume $\nu =(3)$ at the beginning. Note that after this move, $\nu = (2,1),$ so $\SoC(\nu)$ decreases by $3$ and $\SoC(\la/\mu)$ increases by $2.$ 
Next, move the remaining $\la / \mu$ box (the one not in the $1$st column) to row $d-1.$ This move increases the value of $\Delta$ by at least $1,$ so the value of $\Delta$ in fact increases after these two moves. Moreover, $\nu$ becomes the tail of $\la,$ so $\lambda=\mu\uplus\nu.$

(c-2) Assume $\mu_{d-1}=1$ and $\lambda_{d-1}=2.$ 

(c-2.1) If $\lambda_d =2,$ move the $\mu$ box from row $d-1$ to the first row and rearrange the $\nu$ boxes so that $\lambda_{d-1}=2$ and $\lambda_d=1.$
It is easy to compute that the value of $\Delta$ does not change after this move.
Moreover, $\nu$ becomes the tail of $\la,$ so $\lambda=\mu\uplus\nu$ and we are done.

(c-2.2) If $\mu_{d-1}=1, \lambda_{d-1}=2$ and $\la_d=1,$ then move one $\la / \mu$ box to row $d$ (which is a Robin Hood move) to end up with case (c-2.1).

\subsubsection{Case $d=e$}
(a) If $\mu_{d-1}\geq 3,$ then move the $\mu$ boxes in row $d$ to the first row and move the $\la / \mu$ boxes to row $d.$ these are Robin Hood moves, hence the value of $\Delta$ increases. Now $\nu$ is the tail of $\la,$ so $\lambda=\mu\uplus\nu.$

(b) If $\mu_{d-1}= 2,$ then move the $\mu$ boxes in rows $d-1$ and $d$ to the first row and place the $\la / \mu$ boxes in rows $d-1$ and $d$ so that $\la_{d-1}=2$ and $\la_d=1.$ Doing so, the contents of these three boxes decrease, but the height of $\nu$ might increase. 

So assume that $\nu = (3)$ at the start. Then, after the move,
$\SoC(\nu)$ decreases by $3.$ 
 By assumption ($\mu_{d-1}=2$), each of the three boxes in $\la /\mu$ satisfies
$\text{content(box)} > 3-d$. 
Hence, $\SoC(\la/\mu)$ increases by  $>3$ and the value of $\Delta$ increases,
After this moves, $\nu$ is the tail of $\la,$ so $\lambda=\mu\uplus\nu.$

(c) If $\mu_{d-1}=1,$ then move $\mu_d,\mu_{d-1}$ and $\mu_{d-2}$ to the first row and move the three $\la/\mu$ boxes to rows $d,d-1,d-2,$ respectively, so that $\la_d = \la_{d-1} = \la_{d-2}=1.$ Note that $\SoC(\nu)$ may decrease by $ 6$ if $\nu$ changes from $(3)$ to $(1,1,1).$ However, $\SoC(\la/\mu)$ at the start is at least
$
2 \cdot 3(2-(d-2)) = 2 \cdot (12-3d)
$
and its final value is 
$
2\cdot(1-d+2-d+3-d) = 2\cdot (6-3d).
$
Hence, after these moves, the value of $\Delta$ does not decrease.
Moreover, $\nu$ becomes the tail of $\la,$ so $\lambda=\mu\uplus\nu.$

\subsection{Case $k=4$}

The partitions of $4$ satisfy
$$
\eta_{(4)} = 0 < \eta_{(3,1)} = 8 < \eta_{(2,2)} = 12 < \eta_{(2,1,1)} = 16 < \eta_{(1,1,1,1)} = 24.
$$
The sum of contents $\SoC(\nu)$ for each of these partitions $\nu$ is
 $$
 \SoC((4)) = 6 > \SoC((3,1)) = 2 > \SoC((2,2)) = 0 >
 \SoC((2,1,1)) = -2 > \SoC((1,1,1,1)) = -6.
 $$
As before, to prove Proposition \ref{prop:02A} we separate several cases.

\subsubsection{Case $d-e=4$}
In this case $\la_d = \la_{d-1} = \la_{d-2} = \la_{d-3} = 1$, and none of these boxes belong to $\mu$. Hence, $\la / \mu$ and $\mu$ do not share rows and $\la = \mu \uplus \nu.$

\subsubsection{Case $d-e=3$}
(a) If $\mu_e \geq 2,$ then move the $\la/\mu$ box that is potentially not in the last three rows to row $d-2$ so that $\la_{d-2}=2, \la_{d-1}=\la_d=1.$ This is a sequence of Robin Hood moves (note that the height of $\nu$ at the start must have been $3,$ so $\nu$ did not change), hence the value of $\Delta$ increases. Moreover, $\nu$ becomes the tail of $\la,$ so $\la = \mu \uplus \nu.$

(b) If $\mu_e=1,$ then let $j_0$ be the smallest $j$ such that $\mu_j=1.$ Move the $\la/\mu$ box that is not in the last three rows to row $j_0.$ Now exchanging the $\mu$ box in row $j_0$ with the $\nu$ box in row $d-2$ yields a (actually the same) triple of partitions $(\la,\mu,\nu)$ such that each row of $\la$ is either a row of $\mu$ or a row of $\nu,$ so $\la = \mu \uplus \nu.$

\subsubsection{Case $d-e=2$} In this case $\nu$ is either $(3,1)$ or $(2,2)$ at  start.

(a) If $\mu_e \geq 2,$ then move the two $\la/\mu$ boxes that are potentially not in the last two rows to the last two rows so that $\la_{d-1} = \la_d =2.$ Now assume the worst case where $\nu = (3,1)$ in the beginning. Clearly, after these moves, $\nu$ changes into $(2,2).$
Recall that
$$
\Delta = 2k(n-k) + 2 \SoC(\nu)-2\SoC(\la/\mu).
$$
After the moves, the value of $\Delta$ decreases by $2 \cdot 2$ because $\nu$ changes. However,  $\SoC(\la/\mu)$ decreases by at least $4$ because the content of each of the two moved boxes increases by at lest two. Hence, the value of $\Delta$ increases with this construction.

(b) Assume $\mu_e = 1.$ 

(b-1) Assume $\mu_{e-1} \geq 2.$ Move $\mu_e$ to the first row.

(b-1.1) If $\mu_{e-1} = \la_{e-1}=2,$ then the possible starting positions with smallest content of the remaining two boxes were 
$(2,d-1),(2,d-2).$ In this case exchange $\mu_e$ with $\la_d$ to obtain a (the same) triple of partitions $(\la,\mu,\nu)$ such that $\la = \mu \uplus \nu.$

If the two remaining $\la/\nu$ boxes were not at positions $(2,d-1),(2,d-2),$ then at least one of them had content $3-(d-4)=7-d$ or higher.
Now move the $\la/\mu$ boxes so that $\la_{d-2}=2,\la_{d-1}=1,\la_d=1.$
If $\nu=(3,1)$ at the start (worst case), then $\SoC(\nu)$ decreases by at most $4,$ but $\SoC(\la/\mu)$ increases by at least $4$ as well (because the box with content $\geq 7-d$ goes to position $(1,d-2)$). Hence $\Delta$ does not decrease, but now $\la = \mu \uplus \nu$.

(b-1.2) If $\mu_{e-1} = 2$ and $\la_{e-1} \geq 3,$ then $d\geq 5$ and $\la_{d-4}\geq 3.$ 
If $\mu_{e-2}\geq 3,$ move the $\la/\mu$ boxes to rows $d-1$ and $d,$ respectively, so that $\la_{d-1}=\la_d=2.$ Also move one $\mu$ box from row $d-4$ to row $d-2$ (so that $\la_{d-2}=\mu_e=2$). Now $\nu$ changes from $(3,1)$ to $(2,2),$ hence $\SoC(\nu)$ decreases by $2,$ and $\SoC(\la/\mu)$ increases by at least $4,$ so $\Delta$ increases (and now $\la = \mu \uplus \nu$). 

If $\mu_{e-2}=2,$ then $\nu=(2,2)$ at the start. Move two $\la/\mu$ boxes to rows $d-1$ and $d-2$ so that $\la_{d-2}=2 = \mu_e+1,\la_{d-1}=2,\la_d=1.$ After these moves, $\nu$ does not change and $\SoC(\la/\mu)$ increases, hence $\Delta$ increases. Now exchange $\mu_e$ with $\la_d$ to obtain a triple of partitions $(\la,\mu,\nu)$ such that $\la = \mu \uplus \nu.$

(b-1.3) Assume $\mu_{e-1}\geq3.$ If no $\la/\mu$ box is at position $(2,d-1),$ then it is easy to see (as before) that moving the $\la/\mu$ boxes so that $\la_{d-2}=2,\la_{d-1}=1,\la_d=1,$ increases $\Delta$ and yields $\la = \mu \uplus \nu$.   (Note that in this case $\nu$ changes from $(3,1)$ or $(2,2)$ to $(2,1,1)$).

Otherwise, if two $\la/\mu$ boxes are at positions $(2,d-1),(2,d-2),$ then do the procedure from (b-1.1).

(b-2) If $\mu_{e-1}=1,$ then let again $j_0$ be the smallest $j$ such that $\mu_j =1.$ Move $\mu_e$ to row $j_0$ and move three  $\la/\mu$ boxes so that $\la_{d-2}=1,\la_{d-1}=1,\la_d=1$ and move one $\la/\mu$ box to row $j_0+1 \leq d-2.$ Now assume $\nu=(3,1)$ at the start (worst case). 
Then after the moves, $\SoC(\nu)$ decreases by at most $4.$ The positions for $\la/\mu$ boxes with lowest content are $(2,d-3),(3,d-3).$ Hence it is easy to see that the value of $\Delta$ increases. Again, switch $\la_{d-2}$ with $\mu_{j_0+1}$ to obtain a triple of partitions $(\la,\mu,\nu)$ such that $\la = \mu \uplus \nu$.

Note that if $\nu=(2,2)$ at the start, then the remaining two $\la/\mu$ boxes could have been at positions $(2,d-2),(2,d-3).$ But in this case  $\SoC(\nu)$ decreases by  $2$ and $\SoC(\la/\mu)$ increases by at least $2$ as well.

\subsubsection{Case $d-e=1$}
In this case the height of $\nu$ is at most $3$ at  start.

(a) If $\mu_e \geq 4,$ then move all the $\la/\mu$ boxes to row $d.$ This clearly increases $\Delta$ and produces a triple $(\la,\mu,\nu)$ such that $\la = \mu \uplus \nu$. 

(b) If $\mu_e = 3,$ then let $j_1$ be the smallest $j$ such that $\mu_j=3.$ Move the $\la/\mu$ box that is not in row $d$ to row $j_1.$ Now exchanging the $\mu$ boxes in row $j_0$ with the $\nu$ boxes in row $d$ yields a (in fact the same) triple of partitions $(\la,\mu,\nu)$ such that $\la = \mu \uplus \nu$. 

(c) Assume $\mu_e=2.$ Note that since $n-k\geq k,$ we have $d\geq 3.$ If needed, move another $\la / \mu$ box to row $d$ so that $\la_d=2.$ This move clearly increases $\Delta.$ 

(c-1) If $\mu_{e-1}=4,$ then move the two remaining $\la/\mu$ boxes to the end of row $d-1.$ Doing so, the value of $\Delta$ increases (even if one of these boxes was located in row $d$ - in that case $\nu = (3,1)$ at the beginning and $(4)$ at the end, which makes up for moving one box to the row above). We also have  $\la = \mu \uplus \nu$ in the end.

(c-2) If $\mu_{e-1} = 3,$ then move $\mu_{e}$ to the first row and move the $\la/\mu$ boxes so that $\la_{d-1}=3,\la_d=1.$ Now if  $\nu=(4)$ at the start, $\SoC(\nu)$ decreases by $4$ and in that case $\SoC(\la/\mu)$ increases by at least $4.$ Indeed, it decreases by $2$ because $\la_d$ gets moved to the row $d-1$ and it then increases by at least $6$ since one box gets moved to row $d$.
In the case where $\la_d=3$ at the start, we have $\nu=(3,1)$ at the start. Then both $\SoC(\nu)$  and $\SoC(\la/\mu)$ stay the same so that $\Delta$ does not decrease.

(c-3) If $\mu_{e-1} = 2,$ then move $\mu_e$ to the first row and move the two remaining $\la/\mu$ boxes to row $d-1$ so that $\la_{d-1}=\la_d=2.$
If $\nu=(3,1)$ at the start, that the worst case is when $\la_{d-2}=\la_{d-1}=3$ (the remaining $\la/\mu$ boxes are in rows $d-2$ and $d-1,$ respectively). In that case $\SoC(\nu)$ decreases by $2,$ but $\SoC(\la/\mu)$ increases by $4,$ hence $\Delta$ increases.
If $\nu=(4)$ at the start, then  $\SoC(\nu)$ decreases by $6,$ but also $\SoC(\la/\mu)$ increases by at least $6,$ so that the value of $\Delta$ does not decrease. However, now $\la = \mu \uplus \nu$.

(d) Assume $\mu_e=1.$ 

(d-1) If $\mu_{e-1}\geq 4,$ then move  three $\la/\mu$ boxes to row $d-1$ and leave one in row $d.$ In this way $\nu = (4)$ at the end.
If only one $\la/\mu$ box was in row $d$ at the start, then this construction is a series of Robin Hood moves. Thus $\Delta$ increases.
if there were two $\la/\mu$ boxes in row $d$ at the start (note that there could not be three or more such boxes in row $d$ at the start), then after the described change, $\SoC(\nu)$ increases by $4$ and $\SoC(\la/\mu)$ does not decrease. Hence, $\Delta$ increases and now $\la = \mu \uplus \nu$.

(d-2) Assume $\mu_{e-1}=3,$ then again move $\mu_e$ to the first row.

If $\la_d=2$ (note that $\la_d\leq2$), then move the $\la/\mu$ boxes so that $\la_{d-1}=\la_d=2.$ In this way $\nu$ changes from $(3,1)$ to $(2,2),$ hence $\SoC(\nu)$ decreases by $2,$  and $\SoC(\la/\mu)$ increases by at least $2.$ Thus $\Delta$ does not decrease but we now have $\la = \mu \uplus \nu$. 

If $\la_d=1,$ then move the three $\la/\mu$ boxes not in row $d$ to row $d-1.$ Now $\SoC(\nu)$ decreases by at most $4,$ but  $\SoC(\la/\mu)$ increases by at least $4,$ hence $\Delta$ does not decrease and we again have  $\la = \mu \uplus \nu$.

(d-3) If $\mu_{e-1} = 2,$ then note that $d \geq 4.$ Now move $\mu_e$ to the first row and place the $\la/\mu$ boxes in rows $d-1$ and $d$ so that $\la_{d-1}=\la_d=2.$  

Now if at the start we had $\la_{d-1}=\la_d=2$ (and $\mu_e=\mu_{d-1}=1$), then $\nu=(3,1)$ in the beginning. Hence, with the described construction,  $\SoC(\nu)$ decreases by $2$ and  $\SoC(\la/\mu)$ increases by at least $3.$ Thus $\Delta$ increases and we  have  $\la = \mu \uplus \nu$. 

If $\la_d=1$ at the start, then $\nu$ could have been $(4)$ at the start and 
with the construction, $\SoC(\nu)$ decreases by at most $6.$ But   $\SoC(\la/\mu)$ in this case increases by more than $6$ (at least two $\la/\mu$ boxes are in row $d-2$ or higher). Thus $\Delta$ increases and we obtain $\la = \mu \uplus \nu$.

(d-4)  If $\mu_{e-1} = 1,$ recall again that $j_0$ is the smallest $j$ such that $\mu_j=1.$ Now move $\mu_e$ to the first row and move $\mu_{e-1}$ to row $j_0.$ Move two $\la/\mu$ boxes to rows $d-2$ and $d-1,$ respectively and move one $\la/\mu$ box to row $j_0+1 \leq d-2.$

(d-4.1) If $\nu=(4)$ at the start, then the possible positions with smallest contents of the  $\la/\mu$ boxes (that are not the one at position $(1,d)$) are:
\begin{enumerate}
    \item $(2,d-2),(3,d-3),(4,d-3)$ if $\mu_{e-2} \leq 2.$ Then $\SoC(\nu)$ decreases by $8$ and it is easy to count that $\SoC(\la/\mu)$ also increases by at least $8.$ 
    \item $(2,d-2),(3,d-2),(4,d-2)$ if $\mu_{e-2} \geq 3.$  In this case the described construction does not work. Instead, still move $\mu_e$ to the first row, then move the $\la/\mu$ boxes to rows $d-1$ and $d$ so that $\la_{d-1}=\la_d=2$ and finally, move one $\mu$ box from row $d-3$ to row $d-2.$ Then $\SoC(\nu)$ decreases by $6$ (as $\nu$ changes from $(4)$ to $(3,1)$), but  $\SoC(\la/\mu)$ increases by at least $8.$
\end{enumerate}

In all the cases $\Delta$ either increases or does not decrease after the described procedure and in the end we obtain $\la=\mu \uplus \nu.$

(d-4.2) If $\nu=(3,1)$ at the start, then the possible positions with smallest contents of the  $\la/\mu$ boxes (that are not the one at position $(1,d)$) are
\begin{enumerate}
    \item $(2,d-1),(2,d-2),(3-d-2)$ if $\mu_{e-2}\geq 3.$ Here again the described construction does not work. Instead, still move $\mu_e$ to the first row, then move the $\la/\mu$ boxes to rows $d-1$ and $d$ so that $\la_{d-1}=\la_d=2$ and finally, move one $\mu$ box from row $d-3$ to row $d-2.$ Then $\SoC(\nu)$ decreases by $2$ (as $\nu$ changes from $(3,1)$ to $(2,2)$), but  $\SoC(\la/\mu)$ increases by at least $5.$
    \item $(2,d-1),(2,d-2),(3,d-3)$ if $\mu_{e-2}\leq 2.$ In this case we use the construction, where $\nu$ changes from $(3,1)$ to $(2,1,1).$ The value of $\SoC(\nu)$ thus decreases by $4,$ but it is easy to check that  $\SoC(\la/\mu)$ increases by at least $4.$ Hence $\Delta$ does not decrease.
\end{enumerate}
Again, in all the above  cases $\Delta$ either increases or does not decrease after the described procedure and in the end we obtain $\la=\mu \uplus \nu.$
Note that we cannot have $\nu=(2,2)$ at the start.

(d-4.3) If $\nu=(2,1,1)$ at the start, then the possible positions with smallest contents of the  $\la/\mu$ boxes (that are not the one at position $(1,d)$) are $(2,d),(2,d-1),(2,d-2).$ Then $\nu$ does not change during the described construction. Note that $\SoC(\la/\mu)$ also does not change, but we get $\la=\mu \uplus \nu.$

\subsubsection{Case $d=e$}

(a) If $\mu_{d-1}\geq 4,$ then move $\mu_1$ to the first row and move all the $\la/\mu$ boxes to row $d.$ This construction (is a series of Robin Hood moves and hence) clearly increases $\Delta$ and produces $\la = \mu \uplus \nu.$

(b) If $\mu_{d-1} = 3,$ let again $j_1$ be the smallest $j$ such that $\mu_j=3.$ Move $\mu_d$ to the first row, move three $\la/\mu$ boxes to row $d$ and move the last $\la/\mu$ box to the end of row $j_1.$ This is a series of Robin Hood moves (note that $\nu$ changes into $(4),$ which is of the lowest possible height). Thus $\Delta$ increases. By exchanging $\la_d$ with $\mu_{j_0},$ we identify this triple of partitions with a triple of partitions $(\la,\mu,\nu)$ such that $\la = \mu \uplus \nu.$

(c) Assume $\mu_{d-1} =  2.$ Note that in this case $d\geq 3.$ Move $\mu_d$ to the first row and  unless $\mu_{d-2}\geq 4,$ move the $\la/\mu$ boxes to rows $d-1$ and $d,$ respectively, so that $\la_{d-1}=\la_d=2.$ We have the following cases:

\begin{enumerate}
    \item $\nu=(4)$ at the start. Since $\mu_{d-2}\leq 3,$ then the possible positions with smallest contents of the  $\la/\mu$ boxes are $(1,d),(2,d),(3,d-1),(4,d-2).$ Now $\SoC(\nu)$ decreases by $6$ (as $\nu$ changes from $(4)$ to $(2,2)$) and  $\SoC(\la/\mu)$ increases by at least $4$ as well. So $\Delta$ does not decrease.
    \item $\nu=(3,1)$ at  start. Then the possible positions with smallest contents of the  $\la/\mu$ boxes are $(1,d),(2,d),(3,d),(3,d-1).$ Now $\SoC(\nu)$ decreases by $2$ (as $\nu$ changes from $(3,1)$ to $(2,2)$) and  $\SoC(\la/\mu)$ increases by at least $2$ as well. So $\Delta$ does not decrease.
    \item $\nu=(2,2)$ or $\nu=(2,1,1)$ at the start. In this case we have at most three $\la/\mu$ boxes in row $d,$ hence all moves in the described construction are Robin Hood moves (the height of $\nu$ clearly does not increase). Thus $\Delta$ increases.
\end{enumerate}
In all three cases we obtain $\la = \mu \uplus \nu.$
\vspace{3mm}

If $\mu_{d-2}\geq 4,$ then still move  $\mu_d$ to the first row and move one $\mu$ box from row $d-2$ to row $d-1.$ Now move the $\la/\mu$ boxes as in (b). The argument in (b) shows that this construction increases $\Delta,$ while it  produces a triple $(\la,\mu,\nu)$ with $\la = \mu \uplus \nu.$

(d) Assume $\mu_{d-1} = 1.$ Then clearly $d \geq 3.$ Move $\mu_d$ to the first row.

(d-1) If $\mu_{d-2}\geq 2,$ then move $\mu_{d-1}$ to the first row as well and move the $\la/\mu$ boxes to rows $d-1$ and $d$ so that $\la_{d-1}=\la_d = 2.$ We separate the following cases:
\begin{enumerate}
    \item $\nu=(4)$ at the start. Then the possible positions with smallest contents of the  $\la/\mu$ boxes are $(2,d-1),(3,d-1),(4,d-1),(5,d-1).$ Now $\SoC(\nu)$ decreases by $6$ (as $\nu$ changes from $(4)$ to $(2,2)$) and it is easy to see that  $\SoC(\la/\mu)$ increases by at least $11.$ So $\Delta$ increases.
    \item $\nu=(3,1)$ at the start. Then the possible positions with smallest contents of the  $\la/\mu$ boxes are $(2,d),(2,d-1),(3,d-1),(4,d-1).$ Now $\SoC(\nu)$ decreases by $2$ (as $\nu$ changes from $(3,1)$ to $(2,2)$) and  $\SoC(\la/\mu)$ increases by at least $6.$ So $\Delta$  increases.
    \item $\nu=(2,2)$ or $\nu=(2,1,1)$ or $\nu=(1,1,1,1)$ at the start.  In this case we have at most two $\la/\mu$ boxes in row $d,$ hence all moves in the described construction are Robin Hood moves (the height of $\nu$ clearly does not increase). Thus $\Delta$ increases.
\end{enumerate}
In all three cases we again obtain $\la = \mu \uplus \nu.$
\vspace{3mm}

(d-2) If $\mu_{d-2} = 1,$ note that $d\geq4.$ Now move $\mu_d$ to the first row.

If $\mu_{d-3} \geq 2,$ do as in step (d-4) when $d-e=1.$

If $\mu_{d-3} = 1,$ move $\mu_{d-1}$ and $\mu_{d-2}$ to the first row as well. Move three $\la/\mu$ boxes to rows $d-2,d-1,d,$ respectively, so that $\la_d=\la_{d-1}=\la_{d-2}=1.$ Move the last $\la/\mu$ box to row $j_0.$ (Recall that $j_0$ is the smallest $j$ such that $\mu_j=1.$) We separate the following cases:
\begin{enumerate}
    \item $\nu=(4)$ at the start. Then the possible positions with smallest contents of the  $\la/\mu$ boxes are $(2,d-3),(3,d-3),(4,d-3),(5,d-3)$ (in this case $j_0=d-3$). Now $\SoC(\nu)$ decreases by $8$ (as $\nu$ changes from $(4)$ to $(2,1,1)$) and it is easy to count that $\SoC(\la/\mu)$ increases by more than $8.$ So $\Delta$ increases.
    \item $\nu=(3,1)$ at the start. Then the possible positions with smallest contents of the  $\la/\mu$ boxes are $(2,d-3),(2,d-2),(3,d-3),(4,d-3)$ (again with $j_0=d-3$). Now $\SoC(\nu)$ decreases by $4$ (as $\nu$ changes from $(3,1)$ to $(2,1,1)$) and it is easy to count that $\SoC(\la/\mu)$ increases by at least $11.$ So $\Delta$ increases.
    \item $\nu=(2,2)$ at the start. Then the possible positions with smallest contents of the  $\la/\mu$ boxes are $(2,d-3),(2,d-2),(3,d-3),(3,d-2)$ (again with $j_0=d-3$). Now $\SoC(\nu)$ decreases by $2$ (as $\nu$ changes from $(2,2)$ to $(2,1,1)$) and $\SoC(\la/\mu)$ increases by at least $9.$ So $\Delta$ increases.
    \item $\nu=(2,1,1)$ or $\nu=(1,1,1,1)$ at the start. In this case we have at most one $\la/\mu$ box in rows $d$ and $d-1,$ hence all moves in the described construction are Robin Hood moves (the height of $\nu$ clearly does not increase). Thus $\Delta$ increases.
\end{enumerate}
All three cases produce a triple of partitions $(\la,\mu,\nu)$ with $\la = \mu \uplus \nu.$

\subsection{Small $d$}

Lastly, we prove Proposition \ref{prop:02A} for $d=2,3$. 

\subsubsection{$d=2$}
If $d=2,$ then given a triple $(\la,\mu,\nu),$ both $\mu$ and $\nu$ have at most $2$ rows. Now clearly, moving all $\la/\mu$ boxes to the second row and moving all $\mu$ boxes to the first row decreases $\SoC(\la/\mu)$. Since the height of $\nu$ decreases as well, the triple $((n-k,k),(n-k),(k))$ maximizes $\Delta.$

\subsubsection{$d=3$}

(a) Assume first that $d-e=2.$ In this case move the $\la/\mu$ boxes from the first row to the second row (note that since $2k\leq n,$ this gives a valid partition). Doing so $\nu$ does not change (since it must have been $\mu_1 = n-k >\la_2$ at the start)  and  $\SoC(\la/\mu)$ decreases, hence $\Delta$ increases and we are done.

(b) Assume $d-e=1.$ Let $a_1=\la_1-\mu_1, a_2 = \la_2-\mu_2,a_3 = \la_3.$ We will call two consecutive rows of $\la/\mu$ disjoint if they do not share a column. Otherwise, we will say that such two columns meet. 

If $a_1,a_2,a_3$ are all disjoint, then joint their boxes into one row and insert the row above or in between or below $\mu_1,\mu_2$ to make a valid partition. Doing so $\nu$ and $\mu$ do not change (note that $\nu=(k)$) and $\SoC(\la)$ decreases. Indeed, another way of seeing this construction is the following: move the $\la/\mu$ boxes down and left so that first all spots below $\mu_2$ are filled, then (if possible) all spots to the right of $\mu_2$ and below $\mu_1$ and finally, leave the (potentially) remaining boxes in the first row. Doing so, $\SoC(\la/\mu)$ decreases and $\nu$ does not change, so $\Delta$ increases. If $k<\mu_2,$ all boxes are in row three and we are done. If $\mu_2 <k<\mu_1,$ then exchange $\la_3$ with $\mu_2$ to identify the correct triple $(\la,\mu,\nu)$ with a triple such that $\mu$ and $\la/\mu$ do not share rows.
Similarly, if $k>\mu_1,$ first exchange $\la_3$ with $\mu_2$ and then $\la_2$ with $\mu_1.$

If any of  $a_1,a_2,a_3$ meet (not that they could all share a column), then move the boxes of $a_3$ that share rows with $a_2$ to the first row and do the same with the boxes of $a_2$ that share rows with $a_1.$ Doing so, $\SoC(\la/\mu)$ increases, but $\SoC(\nu)$ increases even more. In fact, the height of $\nu$ decreases by at least one, and it is easy to see that the path that each box travels in $\la/\mu$ is shorter (in the horizontal direction) than the path it travels in $\nu.$
So this construction increases $\Delta$ and brings us to the case where $a_1,a_2,a_3$ are all disjoint, which we have dealt with already.

(c) Lastly, assume $d=e.$ If $a_1,a_2,a_3$ are disjoint, then join their boxes into one row and insert the row above or in between or below $\mu_1+\mu_3$ and $\mu_2$ to make a valid partition. The same argument from above shows that this construction increases $\Delta.$ 

If any of  $a_1,a_2,a_3$ meet, then move the boxes of $a_3$ that share rows with $a_2$ to the first row and do the same with the boxes of $a_2$ that share rows with $a_1.$ Now repeat the procedure from above, where $a_1,a_2,a_3$ were disjoint.

 \section{Swap matrices on $(\CC^3)^{\otimes n}$ and $(\CC^4)^{\otimes n}$}\label{a:3and4}

Here we give some results specific to the cases $d=3$ and $d=4$.
 
\subsection{Linear space spanned by the products of at most two swap matrices}\label{a:d3l2}

We prove that in $M^{\text{Sw}_3}_n(\mathbb{C}),$ there are no relations of order two other than \eqref{eq:from sn}.

\begin{proposition}\label{propo: 2-indep}
	The set $\mathcal{B}_2$ consisting of
	\begin{align*}
		I& \\
		\textnormal{Swap}_{ij}& \quad i < j\\
		\textnormal{Swap}_{ij} \textnormal{Swap}_{jk}& \quad i < j < k\\
		\textnormal{Swap}_{ij} \textnormal{Swap}_{ik}& \quad i < j < k\\
		\textnormal{Swap}_{ij} \textnormal{Swap}_{kl}& \quad i < j, \  i < k < l
	\end{align*}
is a basis of the subspace of $M^{\text{Sw}_3}_n(\mathbb{C})$ of polynomials in the Swap$_{ij}$ of degree at most two.
\end{proposition}

\begin{proof}
	To prove the linear independence of $\mathcal{B}_2$ suppose 
	\begin{equation}\label{eq: 2-lin indep}
	\begin{split}
		a I + \sum_{i<j}b_{ij}\,\textnormal{Swap}_{ij} + 
		&\sum_{i<j<k}c_{ijk}\,\textnormal{Swap}_{ij}\textnormal{Swap}_{jk} +
		\sum_{i<j<k}d_{ijk}\,\textnormal{Swap}_{ij}\textnormal{Swap}_{ik} +	\\	
		&\sum_{\substack{i<j \\[.5mm] i<k<l}}e_{ijkl}\,\textnormal{Swap}_{ij}\textnormal{Swap}_{kl} = 0.
	\end{split}
\end{equation}
for some scalars $a,b_{ij},c_{ijk},d_{ijk},e_{ijkl}.$ 

To prove that the $e_{ijkl}$ must all be zero, first consider the vector 
$$v = e_1 \otimes e_2 \otimes e_3 \otimes e_1 \otimes e_1 \cdots \otimes e_1 \in \big(\mathbb{C}^d\big)^{\otimes n}.$$
 Evaluate \eqref{eq: 2-lin indep} on $v$ to see that the term Swap$_{1,2}$\,Swap$_{3,4}$ is the only one that yields $e_2 \otimes e_1 \otimes e_1 \otimes e_3 \otimes e_1 \cdots \otimes e_1.$ Hence, $e_{1234}$ must be zero and by analogy, all of the $e_{ijkl}$ must be zero as well.

A similar argument allows us to get rid of the $c_{ijk}$ and the $d_{ijk}.$ Indeed, after evaluating \eqref{eq: 2-lin indep} on $v = e_1 \otimes e_2 \otimes e_3 \otimes e_1 \otimes e_1 \cdots \otimes e_1,$ the term Swap$_{1,2}$\,Swap$_{2,3}$ is the only one that gives
$e_3 \otimes e_1 \otimes e_2 \otimes e_1 \otimes e_1 \cdots \otimes e_1$ and Swap$_{1,2}$\,Swap$_{1,3}$ is the only one that gives
$e_2 \otimes e_3 \otimes e_1 \otimes e_1 \otimes e_1 \cdots \otimes e_1.$

Finally, we are left with a linear combination of single swap matrices and the identity, which are clearly linearly independent.
\end{proof}

\subsection{Gell-Mann matrices of size $3 \times 3$}\label{subsec: gm3}
Recall the definition of the Gell-Mann matrices from Subsection \ref{subsec:gm}. For $d=3,$ there are eight Gell-Mann matrices, namely
\begin{align*}
	&\lambda_1=
	\begin{pmatrix}
		0&1&0\\ 1&0&0 \\ 0&0&0
	\end{pmatrix} \ 
 \lambda_2 = 
 \begin{pmatrix}
 	0&-\mathfrak{i}&0\\ \mathfrak{i}&0&0 \\ 0&0&0
 \end{pmatrix}\ 
\lambda_3 = 
\begin{pmatrix}
	1&0&0\\ 0&-1&0 \\ 0&0&0
\end{pmatrix} \  
\lambda_4 = 
\begin{pmatrix}
	0&0&1\\ 0&0&0 \\ 1&0&0
\end{pmatrix} \\
&\lambda_5 = 
\begin{pmatrix}
	0&0&-\mathfrak{i}\\ 0&0&0 \\ \mathfrak{i}&0&0
\end{pmatrix}\ 
\lambda_6 = 
\begin{pmatrix}
	0&0&0\\ 0&0&1 \\ 0&1&0
\end{pmatrix}\ 
\lambda_7 = 
\begin{pmatrix}
	0&0&0\\ 0&0&-\mathfrak{i} \\ 0&\mathfrak{i}&0
\end{pmatrix}\ 
\lambda_8 = \frac{1}{\sqrt{3}}
\begin{pmatrix}
	1&0&0\\ 0&1&0 \\ 0&0&-2
\end{pmatrix}.
\end{align*}
They are self-adjoint, have trace zero and together with the identity $\lambda_0 := I,$ they form a basis for $M_3(\CC).$ They satisfy
\begin{align}\label{eq: lambda x}
	\lambda_a \lambda_b = \frac{2}{3}\, \delta_{a,b}\,I + \sum_{c=1}^8 (d^{a,b,c} + \mathfrak{i}f^{a,b,c}) \, \lambda_c, 
\end{align}
where $\delta_{a,b}$ is the Kronecker delta and the $f^{a,b,c}$ and $d^{a,b,c}$ are structure constants with
$$
f^{a,b,c} = -\frac{1}{4} \mathfrak{i}\, \text{\textnormal{tr}}(\lambda_a[\lambda_b,\lambda_c]) \quad \text{and} \quad
d^{a,b,c} = \frac{1}{4}\, \text{\textnormal{tr}}(\lambda_a\{\lambda_b,\lambda_c\}).
$$
Here $[A,B]=AB-BA$ and $\{A,B\} = AB+BA$ denote the commutator and the anticommutator respectively. Note that the $f^{a,b,c}$ are antisymmetric and the $d^{a,b,c}$ are symmetric under the interchange of any pair of indices. 
The nonzero $f^{a,b,c}$ are
\begin{align*}
	f^{1,2,3}=1, \quad f^{1,4,7} = f^{1,6,5} = f^{2,4,6} = f^{2,5,7} = f^{3,4,5} = f^{3,7,6} = \frac{1}{2}, \quad
	 f^{4,5,8} = f^{6,7,8} = \frac{\sqrt{3}}{2},
\end{align*}
while the nonzero $d^{a,b,c}$ are
\begin{align*}
	d^{1,4,6} = d^{1,5,7} = &d^{2,5,6} = d^{3,4,4} = d^{3,5,5}  = \frac{1}{2}, \quad
	d^{2,4,7} = d^{3,6,6} = d^{3,7,7} = -\frac{1}{2},\\
	d^{1,1,8}& = d^{2,2,8} = d^{3,3,8} = \frac{1}{\sqrt{3}}, \quad
	 d^{8,8,8} = -\frac{1}{\sqrt{3}}, \\
	 &d^{4,4,8} = d^{5,5,8} = d^{6,6,8} = d^{7,7,8} = -\frac{1}{2 \sqrt{3}}.
\end{align*}

Fix $n \in \NN$. As in Section \ref{sec:gm}, denote
$$
\lambda_a^{j} := \underbrace{I \otimes \cdots \otimes I}_{j-1} \otimes \lambda_a \otimes I \otimes \cdots \otimes I \in M_{3^n}(\CC)
$$
for $a \in \{0,\ldots,8\}$. 
Then, 
\begin{align}\label{eq: basis-gm}
\{\lambda_{a_1}^1 \lambda_{a_2}^2 \cdots \lambda_{a_n}^n \ | \ a_j \in \{0,\ldots,8\},\  j=1,\ldots,n\}
\end{align}
is a basis of $M_{3^n}(\CC)$, and $\lambda_{a_i}^i$ and $\lambda_{a_j}^j$ commute for $i \neq j.$
By Proposition \ref{prop:swaptogm}, each qutrit swap matrix can be written as a linear combination of the Gell-Mann matrices as follows:
\begin{align}\label{eq: sw-gm}
\text{\textnormal{Swap}}_{ij}= \frac{1}{3} I + \frac{1}{2} \sum_{a=1}^8 \lambda_a^i \lambda_a^j.
\end{align}

\subsection{Linear subspace of  $M_{3^n}(\CC)$ spanned by the products of at most three swap matrices}\label{a:d3l3}
Throughout, any two tuples $(i,j)$ and $(k,l)$ are compared w.r.t.~the lex ordering.

\begin{proposition}\label{prop: indep3}
	The set $\mathcal{B}_3$ consisting of $\mathcal{B}_2$ and the three types of cubics
	\begin{equation}\label{eq: type 6}
		\begin{split}
		&\textnormal{Swap}_{ij} \textnormal{Swap}_{kl} \textnormal{Swap}_{pq} \quad 
		i < j, \ k < l,\  p < q,  \, (i,j) < (k,l) < (p,q);\\
	\end{split}
\end{equation}
\begin{equation}\label{eq: type 5}
	\begin{split}
		&\textnormal{Swap}_{ij} \textnormal{Swap}_{jk} \textnormal{Swap}_{pq} \quad i < j < k,\  p<q,\ p,q \notin \{i,j,k\},\\
		&\textnormal{Swap}_{ij} \textnormal{Swap}_{ik} \textnormal{Swap}_{pq} \quad i < j < k,\ p<q, \ p,q \notin \{i,j,k\};\\
		\end{split}
	\end{equation}
\begin{equation}\label{eq: type 4}
	\begin{split}
	&\textnormal{Swap}_{ij} \textnormal{Swap}_{jk} \textnormal{Swap}_{kl}, \ 
	\textnormal{Swap}_{ij} \textnormal{Swap}_{jl} \textnormal{Swap}_{kl},\ 
	\textnormal{Swap}_{ik} \textnormal{Swap}_{jk} \textnormal{Swap}_{jl},\\
	&\textnormal{Swap}_{ik} \textnormal{Swap}_{kl} \textnormal{Swap}_{jl},\ 
	\textnormal{Swap}_{il} \textnormal{Swap}_{jl} \textnormal{Swap}_{jk}
	\quad i < j < k < l;\\		
\end{split}
\end{equation}
	is a basis of the subspace of $M^{\text{Sw}_3}_n(\mathbb{C})$ of polynomials in the Swap$_{ij}$ of degree at most three. 
\end{proposition}

\begin{remark}
	As it can be seen from the proof, any of the cubics in \eqref{eq: type 4} can be replaced by Swap$_{il}$ Swap$_{kl}$ Swap$_{jk}.$
\end{remark}

\begin{proof}
  For the spanning property of $\mathcal{B}_3,$ first note that by Proposition \ref{propo: 2-indep}, every product of three swap matrices involving at least five indices is in the linear span of $\mathcal{B}_3$ and by 
  \eqref{eq:swap-rel}, every product of three swap matrices involving four indices is in the linear span of $\mathcal{B}_3$ as well. This is because a product of three swap matrices involving five (resp.~six) indices corresponds to a product of a $3$-cycle and a disjoint transposition (resp.~a product of three disjoint transpositions; these are in the span of $\mathcal{B}_2$). Similarly,  a product of three swap matrices involving four indices corresponds to either a $4$-cycle or a product of two disjoint transpositions (the latter being in the span of $\mathcal{B}_2$). Moreover, any product of three swap matrices involving three indices or less clearly corresponds to an element in $\mathcal{B}_2$ (either to a $3$-cycle, a transposition or to the identity). This proves the spanning property of $\mathcal{B}_3.$

The proof of the linear independence of $\mathcal{B}_3$ relies heavily on the properties of the Gell-Mann matrices presented in Subsection \ref{subsec: gm3}. Suppose there is a linear dependence among the elements of $\mathcal{B}_3.$ 
Then, using \eqref{eq: sw-gm}, express each of the appearing terms w.r.t.~the basis \eqref{eq: basis-gm} consisting of different combinations of tensor products of the eight Gell-Mann matrices.

First, consider the elements in \eqref{eq: type 6} and observe that for any choice of $i < j, k < l, p < q$ with $(i,j) < (k,l) < (p,q),$ the highest order terms in the expansion of Swap$_{ij}$Swap$_{kl}$Swap$_{pq}$ are of the form
$$
\lambda_a^i\, \lambda_a^j\, \lambda_b^k\, \lambda_b^l\, \lambda_c^p\, \lambda_c^q, \quad \quad a,b,c \in \{1,\ldots,8\}.
$$
Likewise, considering the elements in \eqref{eq: type 5}, for any choice of $i < j < k,  p<q$ with $p,q \notin \{i,j,k\},$ the highest order terms in the expansion of Swap$_{ij}$Swap$_{jk}$Swap$_{pq}$ are of the form
$$
\lambda_a^i\, \lambda_a^j\, \lambda_b^j\, \lambda_b^k\, \lambda_c^p\, \lambda_c^q, 
\quad \quad a,b,c \in \{1,\ldots,8\},
$$
while for Swap$_{ij}$Swap$_{ik}$Swap$_{pq}$ they are of the form
$$
\lambda_a^i\, \lambda_a^j\, \lambda_b^i\, \lambda_b^k\, \lambda_c^p\, \lambda_c^q = \lambda_a^j\,\lambda_a^i\,  \lambda_b^i\, \lambda_b^k\, \lambda_c^p\, \lambda_c^q ,
\quad \quad a,b,c \in \{1,\ldots,8\}.
$$
As for the elements in \eqref{eq: type 4}, for any choice of $i < j < k < l,$ the highest order terms e.g.~in the expansion of Swap$_{ij}$Swap$_{jk}$Swap$_{kl}$ are of the form
$$
\lambda_a^i\, \lambda_a^j\, \lambda_b^j\, \lambda_b^k\, \lambda_c^k\, \lambda_c^l, 
\quad \quad a,b,c \in \{1,\ldots,8\}
$$
and similarly for the other four cases in \eqref{eq: type 4}.

We now gradually eliminate the terms in the linear dependence equation:
(a) By the product formula \eqref{eq: lambda x}, the elements in \eqref{eq: type 6} are the only ones that have terms of order six and more precisely, for any choice of $i < j, k < l, p < q$ with $(i,j) < (k,l) < (p,q),$ the element Swap$_{ij}$Swap$_{kl}$Swap$_{pq}$ has the term $\lambda_1^i\, \lambda_1^j\, \lambda_2^k\, \lambda_2^l\, \lambda_3^p\, \lambda_3^q,$
which does not appear in the expansion of any other element of $\mathcal{B}_3.$ Hence, the coefficients next to each of the elements in \eqref{eq: type 6} have to be zero. 

(b) Now the elements in \eqref{eq: type 5} are the only ones that have terms of order five.
By \eqref{eq: lambda x},
\begin{align*}
\lambda_1& \lambda_3 =  \mathfrak{i}f^{1,3,2}\,\lambda_2 = -\mathfrak{i}f^{1,2,3}\,\lambda_2 = -\mathfrak{i}\,\lambda_2, \\
\lambda&_2\lambda_3 = \mathfrak{i}f^{2,3,1}\lambda_1 = \mathfrak{i}f^{1,2,3}\,\lambda_1 = \mathfrak{i}\,\lambda_1,\\
&\lambda_1 \lambda_6  = d^{1,6,4}\lambda_4 = d^{1,4,6}\lambda_4 = \frac{1}{2} \lambda_4.
\end{align*}
 Hence, for any choice of $i < j < k,  p<q$ with $p,q \notin \{i,j,k\},$ the element Swap$_{ij}$Swap$_{jk}$Swap$_{pq}$ has in its expansion
 \begin{align*}
\lambda_2^i\, \lambda_2^j\, \lambda_3^j\, \lambda_3^k\, \lambda_5^p\, \lambda_5^q + \lambda_1^i\, \lambda_1^j\, \lambda_6^j\, \lambda_6^k\, \lambda_5^p\, \lambda_5^q =
\mathfrak{i}\,\lambda_2^i\, \lambda_1^j\, \lambda_3^k\, \lambda_5^p\, \lambda_5^q + \frac{1}{2} \lambda_1^i\, \lambda_4^j\, \lambda_6^k\, \lambda_5^p\, \lambda_5^q.
 \end{align*}
 But Swap$_{ij}$Swap$_{ik}$Swap$_{pq}$ has in its expansion
 $$
 \lambda_1^i\, \lambda_1^j\, \lambda_3^i\, \lambda_3^k\, \lambda_5^p\, \lambda_5^q + \lambda_1^i\, \lambda_1^j\, \lambda_6^i\, \lambda_6^k\, \lambda_5^p\, \lambda_5^q =
 -\mathfrak{i}\,\lambda_2^i\, \lambda_1^j\, \lambda_3^k\, \lambda_5^p\, \lambda_5^q + \frac{1}{2} \lambda_1^i\, \lambda_4^j\, \lambda_6^k\, \lambda_5^p\, \lambda_5^q.
 $$
 Since the quotient of any two coefficients next to the same basis element in the expansion of Swap$_{ij}$Swap$_{jk}$Swap$_{pq}$ and Swap$_{ij}$Swap$_{ik}$Swap$_{pq}$ must be the same,
 the above implies that the coefficient next to each of the elements in \eqref{eq: type 5}
has to be zero.

(c) So the elements in \eqref{eq: type 4} are now the only ones in the linear dependence equation that have terms of order four and for any choice of $i < j < k<l,$ only the five  products listed in \eqref{eq: type 4} have basis elements of the form $\lambda_a^i \lambda_b^j \lambda_c^k \lambda_d^l.$ Denote the coefficients in the linear dependence equation before the products in \eqref{eq: type 4} by $\alpha_1,\alpha_2,\ldots,\alpha_5$ respectively.

We now consider the equations that we get by reading off the coefficients next to the basis elements $\lambda_a^i \lambda_b^j \lambda_c^k \lambda_d^l$ for several choices of $a,b,c,d$ with $a,b,c,d$ being all different numbers. First one can compute the following part of the expansion of $\textnormal{Swap}_{ij} \textnormal{Swap}_{jk} \textnormal{Swap}_{kl},$
\begin{align*}
	-\frac12\, \lambda_3^i \lambda_1^j \lambda_4^k \lambda_6^l + \frac12\, \lambda_3^i \lambda_1^j \lambda_6^k \lambda_4^l + \frac12\,\lambda_3^i \lambda_4^j \lambda_6^k \lambda_1^l -\frac12\,\lambda_3^i \lambda_6^j \lambda_4^k \lambda_1^l.
\end{align*}
Note that by permuting $i,j,k,l,$ we can obtain four terms in the expansion of the other four elements in \eqref{eq: type 4}.
E.g., by interchanging $k$ and $l,$ we see that $\textnormal{Swap}_{ij} \textnormal{Swap}_{jl} \textnormal{Swap}_{kl}$ has in its expansion the four terms
\begin{align*}
	-\frac12\, \lambda_3^i \lambda_1^j \lambda_6^k \lambda_4^l + \frac12\, \lambda_3^i \lambda_1^j \lambda_4^k \lambda_6^l + \frac12\,\lambda_3^i \lambda_4^j \lambda_1^k \lambda_6^l -\frac12\,\lambda_3^i \lambda_6^j \lambda_1^k \lambda_4^l.
\end{align*}
Using this, one can easily obtain the following equations by comparing the coefficients next to several terms of the form $\lambda_a^i \lambda_b^j \lambda_c^k \lambda_d^l:$
\begin{align*}
	&\lambda_3^i \lambda_1^j \lambda_6^k \lambda_4^l : \quad -\alpha_1 + \alpha_2 + \alpha_4 = 0 \\
	&\lambda_3^i \lambda_6^j \lambda_4^k \lambda_1^l : \quad -\alpha_1 + \alpha_3 + \alpha_5 = 0 \\
	&\lambda_3^i \lambda_4^j \lambda_1^k \lambda_6^l : \quad \alpha_2-\alpha_3 + \alpha_4 + \alpha_5 = 0 \\
	&\lambda_4^i \lambda_1^j \lambda_3^k \lambda_6^l : \quad -\alpha_1 + \alpha_3 + \alpha_4 = 0 \\
	&\lambda_4^i \lambda_6^j \lambda_1^k \lambda_3^l : \quad \alpha_1 - \alpha_2 + \alpha_5 = 0 \\
	&\lambda_4^i \lambda_1^j \lambda_6^k \lambda_3^l : \quad -\alpha_2 + \alpha_3 + \alpha_4 -\alpha_5 = 0 \\
	&\lambda_5^i \lambda_2^j \lambda_3^k \lambda_7^l : \quad -\alpha_1 + \alpha_3 - \alpha_4 = 0.
\end{align*}
The above system of equations has a unique solution $\alpha_1=\alpha_2=\cdots=\alpha_5=0.$ This proves that for any choice of $i<j<k<l,$ the coefficients in the linear dependence equation before the elements in \eqref{eq: type 4} are zero.

(d) We are left with a linear dependence involving terms of degree at most two, which contradicts linear independence of $\mathcal{B}_2$ as shown in Proposition \ref{propo: 2-indep}.
\end{proof}

\subsection{Linear subspace of  $M_{3^n}(\CC)$ spanned by the products of at most four swap matrices}\label{a:d3l4} Recall from the previous subsection that we compare tuples $(i,j)$ and $(k,l)$ w.r.t.~the lex ordering.

\begin{proposition}\label{prop:basis43}
    The set $\hat{\cB}_4$  consisting of $\cB_3$ and the following quartics
    \begin{equation}\label{3:2+2+2+2}
        \begin{split}
            \textnormal{Swap}_{ij} \textnormal{Swap}_{kl} \textnormal{Swap}_{pq} \textnormal{Swap}_{rs} \quad \ \,
			&i < j, \ k < l,\  p < q,\ r < s,  \\ (&i,j) < (k,l) < (p,q)<(r,s);\\
        \end{split}
    \end{equation}
    \vspace{2mm}
    \begin{equation}\label{3:3+2+2}
        \begin{split}
            &\textnormal{Swap}_{ij} \textnormal{Swap}_{jk} \textnormal{Swap}_{pq} \textnormal{Swap}_{rs}\quad \quad  i < j < k,\  p<q, r<s, \\
			 &\textnormal{Swap}_{ij} \textnormal{Swap}_{ik} \textnormal{Swap}_{pq} \textnormal{Swap}_{rs}  \quad\quad p,q,r,s \notin \{i,j,k\}, (p,q)<(r,s);\\
        \end{split}
     \end{equation}
      \vspace{2mm}
     \begin{equation}\label{3:4+2}
         \begin{split}
             &\textnormal{Swap}_{ij} \textnormal{Swap}_{jk} \textnormal{Swap}_{kl} \textnormal{Swap}_{pq} , \ 
	\textnormal{Swap}_{ij} \textnormal{Swap}_{jl} \textnormal{Swap}_{kl} \textnormal{Swap}_{pq},\\
	&\textnormal{Swap}_{ik} \textnormal{Swap}_{jk} \textnormal{Swap}_{jl} \textnormal{Swap}_{pq},\ 
	\textnormal{Swap}_{ik} \textnormal{Swap}_{kl} \textnormal{Swap}_{jl} \textnormal{Swap}_{pq},\\
	&\textnormal{Swap}_{il} \textnormal{Swap}_{jl} \textnormal{Swap}_{jk}\textnormal{Swap}_{pq}
	\quad i < j < k < l,\ p<q,\ p,q \notin \{i,j,k,l\};\\
         \end{split}
     \end{equation}
      \vspace{2mm}
        \begin{equation}\label{3:3+3}
            \begin{split}
                	\textnormal{Swap}_{ij} \textnormal{Swap}_{jk} \textnormal{Swap}_{pq} \textnormal{Swap}_{qr}&, \ 
			\textnormal{Swap}_{ij} \textnormal{Swap}_{jk} \textnormal{Swap}_{pq} \textnormal{Swap}_{pr},\\
			\textnormal{Swap}_{ij} \textnormal{Swap}_{ik} \textnormal{Swap}_{pq} \textnormal{Swap}_{pr}&,
			\quad i < j < k, \ p<q<r, \ i<p,\\	
			 &\quad\{i,j,k\} \cap \{p,q,r\} = \emptyset;
            \end{split}
        \end{equation}
         \vspace{2mm}
\begin{equation}\label{3:5}
    \begin{split}
     \textnormal{Swap}_{ij} \textnormal{Swap}_{ik} \textnormal{Swap}_{jl} \textnormal{Swap}_{jm},\ 
       &\textnormal{Swap}_{ij} \textnormal{Swap}_{ik} \textnormal{Swap}_{jl} \textnormal{Swap}_{km},\\
       \textnormal{Swap}_{ij} \textnormal{Swap}_{ik} \textnormal{Swap}_{kl} \textnormal{Swap}_{km},\ 
    &\textnormal{Swap}_{ij} \textnormal{Swap}_{il} \textnormal{Swap}_{jk} \textnormal{Swap}_{jm},\\ 
 \textnormal{Swap}_{ij} \textnormal{Swap}_{im} \textnormal{Swap}_{jk} \textnormal{Swap}_{jl},\
  &\textnormal{Swap}_{ij} \textnormal{Swap}_{il} \textnormal{Swap}_{im} \textnormal{Swap}_{jk},\\
 \textnormal{Swap}_{ij} \textnormal{Swap}_{ik} \textnormal{Swap}_{im} \textnormal{Swap}_{jl},\
        &\textnormal{Swap}_{ij} \textnormal{Swap}_{ik} \textnormal{Swap}_{im} \textnormal{Swap}_{kl},\\
       \textnormal{Swap}_{ij} \textnormal{Swap}_{ik} \textnormal{Swap}_{il} \textnormal{Swap}_{im}, \ 
          &\textnormal{Swap}_{ij} \textnormal{Swap}_{ik} \textnormal{Swap}_{il} \textnormal{Swap}_{lm}, \\
          \textnormal{Swap}_{ij} \textnormal{Swap}_{ik} \textnormal{Swap}_{il} \textnormal{Swap}_{km},\
          &\textnormal{Swap}_{ij} \textnormal{Swap}_{ik} \textnormal{Swap}_{il} \textnormal{Swap}_{jm},\\
             &i<j<k<l<m,
    \end{split}
\end{equation}
 is a basis of the subspace of $M^{\text{Sw}_3}_n(\mathbb{C})$ of polynomials  of degree at most four in the Swap$_{ij}.$
\end{proposition}

\begin{proof}
 The spanning property of $\hat{\cB}_4$ follows after identifying the products of swap matrices with permutations in $S_n$ using the degree-reducing relation \eqref{eq:degred} with $d=3.$ Indeed, considering the elements which correspond to the product of a $4$-cycle and a disjoint transposition, the type $\textnormal{Swap}_{il} \textnormal{Swap}_{kl} \textnormal{Swap}_{jk} \text{Swap}_{pq}$ missing in \eqref{3:4+2}   is clearly in the span of  $\hat{\cB}_4$ by \eqref{eq:degred}.
 As for the elements that correspond to $5$-cycles, there are  $12$ of the total $24$ $5$-cycles on the letters $i,j,k,l,m$ missing in \eqref{3:5}.
 Their expansions in terms of the elements of $\hat{\cB}_4$ are given in Subsection \ref{sssec:5-cyc}.
 
 Now suppose there is a linear dependence between the elements of $\hat{\cB}_4$  and express the appearing terms w.r.t. the basis \eqref{eq: basis-gm} using the formula  \eqref{eq: sw-gm}. We gradually eliminate terms from this relation starting with the ones with highest order terms.

 (a) Consider the elements in \eqref{3:2+2+2+2}. For any choice of indices $i < j, k < l,  p < q, r < s$ with  $(i,j) < (k,l) < (p,q)<(r,s),$ the highest order terms in the expansion of Swap$_{ij}$Swap$_{kl}$Swap$_{pq}$Swap$_{rs}$ are of the form
	$$
	\lambda_a^i\, \lambda_a^j\, \lambda_b^k\, \lambda_b^l\, \lambda_c^p\, \lambda_c^q \, \lambda_d^r\, \lambda_d^s, \quad \quad a,b,c,d \in \{1,\ldots,8\}.
	$$
	The product formula \eqref{eq: lambda x} implies that the elements in \eqref{3:2+2+2+2} are the only ones in $\hat{\cB}_4$ with such terms and more precisely, for any choice of $i < j, k < l,  p < q, r < s$ with  $(i,j) < (k,l) < (p,q)<(r,s),$ the element Swap$_{ij}$Swap$_{kl}$Swap$_{pq}$Swap$_{rs}$  has the term $\lambda_1^i\, \lambda_1^j\, \lambda_2^k\, \lambda_2^l\, \lambda_3^p\, \lambda_3^q\lambda_4^r\, \lambda_4^s,$
	which does not appear in the expansion of any other element of $\hat{\cB}_4.$ Hence, by analogy, the coefficients next to each of the elements in \eqref{3:2+2+2+2} are zero.

 (b) Now the elements in \eqref{3:3+2+2} are the only ones with highest order terms of degree $7,$ meaning, involving $7$ distinct indices. Using part (b) of the proof of Proposition \ref{prop: indep3}, for any choice of $i < j < k, p<q, r<s$ with $p,q,r,s \notin \{i,j,k\}, (p,q)<(r,s),$  the element Swap$_{ij}$Swap$_{jk}$Swap$_{pq}$Swap$_{rs}$ has in its expansion
	\begin{align*}
		\lambda_2^i\, \lambda_2^j\, \lambda_3^j\, \lambda_3^k\, \lambda_5^p\, \lambda_5^q\, \lambda_9^r\, \lambda_9^s + \lambda_1^i\, \lambda_1^j\, \lambda_6^j\, \lambda_6^k\, \lambda_5^p\, \lambda_5^q\, \lambda_9^r\, \lambda_9^s =
		\mathfrak{i}\,\lambda_2^i\, \lambda_1^j\, \lambda_3^k\, \lambda_5^p\, \lambda_5^q\, \lambda_9^r\, \lambda_9^s + \frac{1}{2} \lambda_1^i\, \lambda_4^j\, \lambda_6^k\, \lambda_5^p\, \lambda_5^q\, \lambda_9^r\, \lambda_9^s.
	\end{align*}
	But Swap$_{ij}$Swap$_{ik}$Swap$_{pq}$ has in its expansion
	$$
	\lambda_1^i\, \lambda_1^j\, \lambda_3^i\, \lambda_3^k\, \lambda_5^p\, \lambda_5^q\, \lambda_9^r\, \lambda_9^s + \lambda_1^i\, \lambda_1^j\, \lambda_6^i\, \lambda_6^k\, \lambda_5^p\, \lambda_5^q\, \lambda_9^r\, \lambda_9^s =
	-\mathfrak{i}\,\lambda_2^i\, \lambda_1^j\, \lambda_3^k\, \lambda_5^p\, \lambda_5^q\, \lambda_9^r\, \lambda_9^s + \frac{1}{2} \lambda_1^i\, \lambda_4^j\, \lambda_6^k\, \lambda_5^p\, \lambda_5^q\, \lambda_9^r\, \lambda_9^s.
	$$
	By the same argument as in part (b) of the proof of Proposition \ref{prop: indep3}, all the coefficients next to the elements in \eqref{3:3+2+2} are zero.

 (c) The elements in \eqref{3:4+2} and \eqref{3:3+3} are now the only ones with  with highest order terms of degree $6.$ We first consider those in  \eqref{3:3+3}. For fixed $i < j < k, p<q<r, i<p$ with $\{i,j,k\} \cap \{p,q,r\} = \emptyset,$ denote the coefficients next to the elements 
$$\textnormal{Swap}_{ij} \textnormal{Swap}_{jk} \textnormal{Swap}_{pq} \textnormal{Swap}_{qr}, \textnormal{Swap}_{ij} \textnormal{Swap}_{jk} \textnormal{Swap}_{pq} \textnormal{Swap}_{pr}, \textnormal{Swap}_{ij} \textnormal{Swap}_{ik} \textnormal{Swap}_{pq} \textnormal{Swap}_{pr}$$ 
by $\beta_1,\beta_2$ and $\beta_3$ respectively. Clearly, these are the only elements in \eqref{3:3+3} whose highest order terms involve precisely the positions $i,j,k,p,q,r.$
So comparing the coefficients next to the basis elements $\lambda_2^i\, \lambda_{4}^j\, \lambda_{6}^k\, \lambda_1^p\, \lambda_5^q\, \lambda_7^r, \, 
\lambda_2^i\, \lambda_5^j\, \lambda_{6}^k\, \lambda_1^p\, \lambda_4^q\, \lambda_7^r$ and $\lambda_1^i\, \lambda_2^j\, \lambda_{3}^k\, \lambda_4^p\, \lambda_5^q\, \lambda_8^r$ give the following equations
	\begin{align*}
		&\lambda_2^i\, \lambda_{4}^j\, \lambda_{6}^k\, \lambda_1^p\, \lambda_5^q\, \lambda_7^r : \quad -\mathfrak{i}\frac14\beta_1 -\mathfrak{i}\frac14\beta_2 + \mathfrak{i}\frac14\beta_3 =0,\\
		&\lambda_2^i\, \lambda_5^j\, \lambda_{6}^k\, \lambda_1^p\, \lambda_4^q\, \lambda_7^r : \quad -\mathfrak{i}\frac14\beta_1 +\mathfrak{i}\frac14\beta_2 + \mathfrak{i}\frac14\beta_3 =0,\\
		&\lambda_1^i\, \lambda_2^j\, \lambda_{3}^k\, \lambda_4^p\, \lambda_5^q\, \lambda_8^r
: \quad -\frac{\sqrt{3}}{2}\beta_1 +\frac{\sqrt{3}}{2}\beta_2-\frac{\sqrt{3}}{2}\beta_3 =0.
	\end{align*}
 The above system has a unique solution $\beta_1=\beta_2=\beta_3=0.$
 Note that each of the highest order terms of the elements in \eqref{3:4+2} necessarily has one of the Gell-Mann matrices $\lambda$ repeated twice. So the coefficients next to the basis elements $\lambda_2^i\, \lambda_{4}^j\, \lambda_{6}^k\, \lambda_1^p\, \lambda_5^q\, \lambda_7^r, \, 
\lambda_2^i\, \lambda_5^j\, \lambda_{6}^k\, \lambda_1^p\, \lambda_4^q\, \lambda_7^r$ and $\lambda_1^i\, \lambda_2^j\, \lambda_{3}^k\, \lambda_4^p\, \lambda_5^q\, \lambda_8^r$ in the expansion of the elements in  \eqref{3:4+2} are zero.
Similarly, each of the highest order terms of the elements corresponding to products of three disjoint transpositions has (at most) three distinct Gell-Mann matrices, each repeated twice. Hence, the coefficients next to the basis elements $\lambda_2^i\, \lambda_{4}^j\, \lambda_{6}^k\, \lambda_1^p\, \lambda_5^q\, \lambda_7^r, \, 
\lambda_2^i\, \lambda_5^j\, \lambda_{6}^k\, \lambda_1^p\, \lambda_4^q\, \lambda_7^r$ and $\lambda_1^i\, \lambda_2^j\, \lambda_{3}^k\, \lambda_4^p\, \lambda_5^q\, \lambda_8^r$ in the expansions of those elements are zero as well.
 We conclude that the coefficients next to the elements
 $$\textnormal{Swap}_{ij} \textnormal{Swap}_{jk} \textnormal{Swap}_{pq} \textnormal{Swap}_{qr}, \textnormal{Swap}_{ij} \textnormal{Swap}_{jk} \textnormal{Swap}_{pq} \textnormal{Swap}_{pr}, \textnormal{Swap}_{ij} \textnormal{Swap}_{ik} \textnormal{Swap}_{pq} \textnormal{Swap}_{pr},$$ 
are zero and by analogy, the coefficients next to all the elements in  \eqref{3:3+3} are zero.

Having eliminated the elements in \eqref{3:3+3}, the fact that the coefficients next to the elements in  \eqref{3:4+2} are zero  easily follows from part (c)  of the proof of Proposition \ref{prop: indep3}.

(d) Now the elements in \eqref{3:5} are the only ones in $\hat{\cB}_4$ with highest order terms of degree $5.$ For fixed $i<j<k<l<m$ denote the coefficients next to the quartics in \eqref{3:5} by $\gamma_1,\ldots,\gamma_{12}$ respectively and note that these are the only elements in \eqref{3:5} whose highest order terms involve precisely the positions $i,j,k,l,m.$ Similar to before, we now compare the coefficients next to several basis elements 
of the form $\lambda_a^i\, \lambda_b^j\, \lambda_c^k\, \lambda_d^l\, \lambda_e^m$ to get a system of equations. We only consider coefficients next to elements $\lambda_a^i\, \lambda_b^j\, \lambda_c^k\, \lambda_d^l\, \lambda_e^m$ with $a,b,c,d,e$ all distinct to ensure that none of them appears in the expansions of the elements \eqref{eq: type 5} corresponding to a product of a $3$-cycle and  disjoint transposition. From the system
\begin{align*}
     &\lambda_4^i\, \lambda_1^j\, \lambda_2^k\, \lambda_8^l\, \lambda_5^m: \quad 2 \gamma_1 - \gamma_5 - \gamma_6 + \gamma_7 + \gamma_8 + \gamma_9 - 2 \gamma_{10} - \gamma_{12} = 0,\\
&\lambda_4^i\, \lambda_1^j\, \lambda_3^k\, \lambda_5^l\, \lambda_2^m: \quad \gamma_1 - \gamma_2 - \gamma_4 + \gamma_7 - \gamma_{10} + \gamma_{11} - \gamma_{12} = 0,\\
&\lambda_4^i\, \lambda_1^j\, \lambda_3^k\, \lambda_8^l\, \lambda_7^m: \quad-2 \gamma_1 - \gamma_5 - \gamma_6 + \gamma_7 + \gamma_8 + \gamma_9 - 2 \gamma_{10} + \gamma_{12} = 0, \\
&\lambda_4^i\, \lambda_1^j\, \lambda_5^k\, \lambda_8^l\, \lambda_2^m: \quad-\gamma_1 + \gamma_2 + \gamma_3 + \gamma_4 - \gamma_5 - \gamma_6 - 2 \gamma_{11} + 2 \gamma_{12} = 0,\\
&\lambda_4^i\, \lambda_1^j\, \lambda_6^k\, \lambda_3^l\, \lambda_8^m: \quad\gamma_1 + \gamma_2 - \gamma_3 + \gamma_4 - \gamma_5 + \gamma_6 - 2 \gamma_7 + 2 \gamma_8 = 0,\\
&\lambda_4^i\, \lambda_1^j\, \lambda_3^k\, \lambda_7^l\, \lambda_8^m: \quad\gamma_1 + \gamma_2 - \gamma_4 + 2 \gamma_6 + \gamma_7 - 2 \gamma_9 + \gamma_{10} + \gamma_{11} + \gamma_{12} = 0,\\
&\lambda_4^i\, \lambda_2^j\, \lambda_6^k\, \lambda_8^l\, \lambda_3^m: \quad-\gamma_1 + \gamma_2 + \gamma_3 - \gamma_4 + \gamma_5 + \gamma_6 - 2 \gamma_{11} + 2 \gamma_{12} = 0,\\
&\lambda_4^i\, \lambda_2^j\, \lambda_8^k\, \lambda_1^l\, \lambda_5^m: \quad-2 \gamma_2 + 2 \gamma_3 - 2 \gamma_4 + \gamma_5 - \gamma_6 + \gamma_7 - \gamma_8 - \gamma_9 + \gamma_{12} = 0,\\
&\lambda_4^i\, \lambda_3^j\, \lambda_7^k\, \lambda_1^l\, \lambda_8^m: \quad\gamma_1 + \gamma_2 - \gamma_3 - \gamma_4 + \gamma_5 - \gamma_6 - 2 \gamma_7 + 2 \gamma_8 = 0,\\
\end{align*}
 we deduce $\gamma_6=0.$ Adding the equations
\begin{align*}
&\lambda_4^i\, \lambda_2^j\, \lambda_1^k\, \lambda_5^l\, \lambda_8^m: \quad    \gamma_1 + \gamma_2 + \gamma_4 + \gamma_7 + 2 \gamma_9 - \gamma_{10} - \gamma_{11} - \gamma_{12} = 0,\\
&\lambda_4^i\, \lambda_3^j\, \lambda_8^k\, \lambda_2^l\, \lambda_6^m: \quad-2 \gamma_2 + 2 \gamma_3 + 2 \gamma_4 + \gamma_5 + \gamma_7 - \gamma_8 - \gamma_9 - \gamma_{12} = 0,\\
&\lambda_4^i\, \lambda_5^j\, \lambda_1^k\, \lambda_8^l\, \lambda_2^m: \quad-2 \gamma_2 + \gamma_3 - 2 \gamma_4 + 2 \gamma_7 - \gamma_8 - \gamma_{10} + \gamma_{11} = 0\\
\end{align*}
yields $\gamma_9=0.$
Moreover, from
$$
\lambda_4^i\, \lambda_3^j\, \lambda_8^k\, \lambda_6^l\, \lambda_1^m: \quad \gamma_1 - \gamma_2 - \gamma_4 + 2 \gamma_5 - \gamma_7 + \gamma_{10} + \gamma_{11} - \gamma_{12} = 0,
$$
we obtain $\gamma_7=\gamma_{11}=0.$
Finally,
\begin{align*}
    \lambda_4^i\, \lambda_2^j\, \lambda_8^k\, \lambda_6^l\, \lambda_3^m: & \quad -\gamma_1 + \gamma_2 - \gamma_4 - 2 \gamma_5 + \gamma_{10} - \gamma_{12} = 0\\
    \lambda_4^i\, \lambda_5^j\, \lambda_2^k\, \lambda_3^l\, \lambda_1^m: & \quad -\gamma_3+\gamma_8+\gamma_{10}-\gamma_{11}=0
\end{align*}
yields $\gamma_1=\gamma_2=\gamma_3=\gamma_4=\gamma_5=\gamma_8=\gamma_{10}=\gamma_{12}=0.$

(e) What remains is a linear dependence involving terms of degree at most three, which contradicts Proposition \ref{prop: indep3}.
\end{proof}

\subsubsection{Expansions of the remaining $5$-cycles} \label{sssec:5-cyc}
To complete the proof of Proposition \ref{prop:basis43} we list the expansions of the $5$-cycles not contained in the basis.
These were produced with the help of noncommutative Gr\"obner bases, but can be readily verified by direct matrix calculation.


\subsection{Gell-Mann matrices of size $4 \times 4$}\label{subsec: gm4}

In the case $d=4,$ the fifteen $4 \times 4$ Gell-Mann matrices are
\begin{align*}
	&\lambda_1 = 
	%
.
\end{align*}
Any product of two such matrices can be expanded in this basis according to a similar formula to \eqref{eq: lambda x},
\begin{align}\label{eq: lambda4 x}
	\lambda_a \lambda_b = \frac{1}{2}\, \delta_{a,b}\,I + \sum_{c=1}^{15} (d^{a,b,c} + \mathfrak{i}f^{a,b,c}) \, \lambda_c, 
\end{align}
where the structure constants $f^{a,b,c}$ and $d^{a,b,c}$ can be again computed via
$$
f^{a,b,c} = -\frac{1}{4} \mathfrak{i}\, \text{\textnormal{tr}}(\lambda_a[\lambda_b,\lambda_c]) \quad \text{and} \quad
d^{a,b,c} = \frac{1}{4}\, \text{\textnormal{tr}}(\lambda_a\{\lambda_b,\lambda_c\}).
$$
In this case the nonzero $f^{a,b,c}$ are
\begin{align*}
	f^{1, 2, 3}=1,& \quad \quad
	f^{1, 5, 6} = f^{1, 10, 11}= f^{3, 6, 7} = f^{3, 11, 12} = f^{4, 10, 13}= f^{6, 12, 13}= -\frac{1}{2},\\
	f^{1, 4, 7}=& f^{1, 9, 12}= f^{2, 4, 6}=  
	f^{2, 5, 7}= f^{2, 9, 11}= f^{2, 10, 12}= f^{3, 4, 5}= f^{3, 9, 10}= \\
	f^{4, 9, 14}& =  f^{5, 9, 13}= f^{5, 10, 14} = f^{6, 11, 14}=  f^{7, 11, 13}= f^{7, 12, 14}=\frac{1}{2},\\
	&f^{4, 5, 8}= f^{6, 7, 8}=\frac{\sqrt{3}}{2},  \quad \quad
	f^{8, 9, 10}= f^{8, 11, 12}=\frac{1}{2 \sqrt{3}}, \\
	f&^{8, 13, 14}=-\frac{1}{\sqrt{3}}, \quad \quad
	f^{9, 10, 15}= f^{11, 12, 15}= f^{13, 14, 15}=\sqrt{\frac{2}{3}},
\end{align*}
and the nonzero $d^{a,b,c}$ are
\begin{align*}
	d^{1, 1, 8} = d&^{2, 2, 8}=d^{3, 3, 8}=\frac{1}{\sqrt{3}}, \quad \quad d^{8, 8, 8}=d^{8, 13, 13}= d^{8, 14, 14}=-\frac{1}{\sqrt{3}}, \\
	d^{1, 1, 15}= d^{2, 2, 15}&=d^{3, 3, 15}=d^{4, 4, 15}=d^{5, 5, 15}=d^{6, 6, 15}=d^{7, 7, 15}=d^{8, 8, 15}=\frac{1}{\sqrt{6}}, \\
	d^{9, 9, 15}= d&^{10, 10, 15}= d^{11, 11, 15}=d^{12, 12, 15}= d^{13, 13, 15}= d^{14, 14, 15}=-\frac{1}{\sqrt{6}}, \\
	d^{1, 4, 6} = d^{1, 5, 7}&= d^{1, 9, 11}= d^{1, 10, 12}=d^{2, 5, 6}=d^{2, 10, 11}=d^{3, 4, 4}=d^{3, 5, 5}=d^{3, 9, 9}= \\
	d^{3, 10, 10}=& \ d^{4, 9, 13}=d^{4, 10, 14}=d^{5, 10, 13}=d^{6, 11, 13}= d^{6, 12, 14}=d^{7, 12, 13}=\frac{1}{2},\\
	d^{2, 4, 7}=d^{2, 9, 12}&=d^{3, 6, 6}= d^{3, 7, 7}=d^{3, 11, 11}=d^{3, 12, 12}= d^{5, 9, 14}=d^{7, 11, 14}=-\frac{1}{2},\\
	d^{4, 4, 8}=& \ d^{5, 5, 8}=d^{6, 6, 8}=d^{7, 7, 8}=-\frac{1}{2 \sqrt{3}}, \quad \quad
	d^{15, 15, 15}=-\sqrt{\frac{2}{3}}\\
	&d^{8, 9, 9}= d^{8, 10, 10}= d^{8, 11, 11}=d^{8, 12, 12}=\frac{1}{2 \sqrt{3}}. 
\end{align*}
Note that the structure constants $f^{a,b,c}$ and $d^{a,b,c}$ with $a,b,c\in\{1,\ldots,8\}$ coincide with the structure constants pertaining to the $3 \times 3$ Gell-Mann matrices.

By Proposition \ref{prop:swaptogm}, each swap matrix Swap$^{(4)}_{ij}$ can be written in terms of the $4 \times 4$ Gell-Mann matrices as follows
\begin{align}\label{eq: sw-gm 4}
	\text{\textnormal{Swap}}^{(4)}_{ij}= \frac{1}{4} I + \frac{1}{2} \sum_{a=1}^{15} \lambda_a^i \lambda_a^j.
\end{align}
\subsection{Linear subspace of $M_{4^n}(\mathbb{C})$ spanned by the products of at most $4$ swap matrices}\label{a:d4l4}

Again, any two tuples $(i,j)$ and $(k,l)$ are compared w.r.t.~the lex ordering. Let $\tilde{\mathcal{B}}_3$ be the set of all products of at most $3$ swap matrices that correspond to different permutations in $S_n.$
For fixed $i<j<k<l$ denote by $\mathcal{B}_{ijkl}$ the set consisting of the cubics
\begin{equation}\label{eq: cubics}
	\begin{split}
&\textnormal{Swap}_{ij} \textnormal{Swap}_{jk} \textnormal{Swap}_{kl}, \ 
\textnormal{Swap}_{ij} \textnormal{Swap}_{jl} \textnormal{Swap}_{kl},\ 
\textnormal{Swap}_{ik} \textnormal{Swap}_{jk} \textnormal{Swap}_{jl},\\
&\textnormal{Swap}_{ik} \textnormal{Swap}_{kl} \textnormal{Swap}_{jl},\ 
\textnormal{Swap}_{il} \textnormal{Swap}_{jl} \textnormal{Swap}_{jk},\
\textnormal{Swap}_{il} \textnormal{Swap}_{kl} \textnormal{Swap}_{jk}
 \end{split}
\end{equation}
and for fixed $i<j<k<l<m$ denote by $\mathcal{B}_{ijklm}$ the set consisting of the quartics
\begin{equation}\label{eq: quart}
 \begin{split}
 	&\textnormal{Swap}_{ij} \textnormal{Swap}_{jk} \textnormal{Swap}_{kl} \textnormal{Swap}_{lm}, \ 
 	\textnormal{Swap}_{ij} \textnormal{Swap}_{jk} \textnormal{Swap}_{km} \textnormal{Swap}_{lm}, \ \\
 	&\textnormal{Swap}_{ij} \textnormal{Swap}_{jl} \textnormal{Swap}_{kl} \textnormal{Swap}_{km}, \ 
 	\textnormal{Swap}_{ij} \textnormal{Swap}_{jl} \textnormal{Swap}_{lm} \textnormal{Swap}_{km}, \ \\
 	&\textnormal{Swap}_{ij} \textnormal{Swap}_{jm} \textnormal{Swap}_{km} \textnormal{Swap}_{kl}, \  
 	\textnormal{Swap}_{ij} \textnormal{Swap}_{jm} \textnormal{Swap}_{lm} \textnormal{Swap}_{kl}, \ \\
 	&\textnormal{Swap}_{ik} \textnormal{Swap}_{jk} \textnormal{Swap}_{jl} \textnormal{Swap}_{lm}, \ 
 	\textnormal{Swap}_{ik} \textnormal{Swap}_{jk} \textnormal{Swap}_{jm} \textnormal{Swap}_{lm}, \ \\
 	&\textnormal{Swap}_{ik} \textnormal{Swap}_{kl} \textnormal{Swap}_{jl} \textnormal{Swap}_{jm}, \ 
 	\textnormal{Swap}_{ik} \textnormal{Swap}_{kl} \textnormal{Swap}_{lm} \textnormal{Swap}_{jm}, \ \\
 	&\textnormal{Swap}_{ik} \textnormal{Swap}_{km} \textnormal{Swap}_{jm} \textnormal{Swap}_{jl}, \ 
 	\textnormal{Swap}_{ik} \textnormal{Swap}_{km} \textnormal{Swap}_{lm} \textnormal{Swap}_{jl}, \ \\
 	&\textnormal{Swap}_{il} \textnormal{Swap}_{jl} \textnormal{Swap}_{jk} \textnormal{Swap}_{km}, \ 
 	\textnormal{Swap}_{il} \textnormal{Swap}_{jl} \textnormal{Swap}_{jm} \textnormal{Swap}_{km}, \ \\
 	&\textnormal{Swap}_{il} \textnormal{Swap}_{kl} \textnormal{Swap}_{jk} \textnormal{Swap}_{jm}, \  
 	\textnormal{Swap}_{il} \textnormal{Swap}_{kl} \textnormal{Swap}_{km} \textnormal{Swap}_{jm}, \ \\
 	&\textnormal{Swap}_{il} \textnormal{Swap}_{lm} \textnormal{Swap}_{jm} \textnormal{Swap}_{jk}, \ 
 	\textnormal{Swap}_{il} \textnormal{Swap}_{lm} \textnormal{Swap}_{km} \textnormal{Swap}_{jk}, \ \\
 	&\textnormal{Swap}_{im} \textnormal{Swap}_{jm} \textnormal{Swap}_{jk} \textnormal{Swap}_{kl}, \ 
 	\textnormal{Swap}_{im} \textnormal{Swap}_{jm} \textnormal{Swap}_{jl} \textnormal{Swap}_{kl}, \ \\
 	&\textnormal{Swap}_{im} \textnormal{Swap}_{km} \textnormal{Swap}_{jk} \textnormal{Swap}_{jl}, \
 	\textnormal{Swap}_{im} \textnormal{Swap}_{km} \textnormal{Swap}_{kl} \textnormal{Swap}_{jl}, \ \\
 	&\textnormal{Swap}_{im} \textnormal{Swap}_{lm} \textnormal{Swap}_{jl} \textnormal{Swap}_{jk}.   
 \end{split}
\end{equation}

\begin{proposition}
	The set $\mathcal{B}_4$ consisting of $\tilde{\mathcal{B}}_3$ and the quartics
	\begin{equation}\label{eq: type 8}
		\begin{split}
			\textnormal{Swap}_{ij} \textnormal{Swap}_{kl} \textnormal{Swap}_{pq} \textnormal{Swap}_{rs} \quad \ \,
			&i < j, \ k < l,\  p < q,\ r < s,  \\ (&i,j) < (k,l) < (p,q)<(r,s);\\
		\end{split}
	\end{equation}
\vspace{0.1mm}
	\begin{equation}\label{eq: type 7}
		\begin{split}
			&\textnormal{Swap}_{ij} \textnormal{Swap}_{jk} \textnormal{Swap}_{pq} \textnormal{Swap}_{rs}\quad \quad  i < j < k,\  p<q, r<s, \\
			 &\textnormal{Swap}_{ij} \textnormal{Swap}_{ik} \textnormal{Swap}_{pq} \textnormal{Swap}_{rs}  \quad\quad p,q,r,s \notin \{i,j,k\}, (p,q)<(r,s);\\
		\end{split}
	\end{equation}
\vspace{0.1mm}
	\begin{equation}\label{eq: type 6-a}
		\begin{split}
			t\cdot\textnormal{Swap}_{pq}, \quad t \in \mathcal{B}_{ijkl},\ \ i < j < k < l,\ p<q,\ p,q \notin \{i,j,k,l\};\\		
		\end{split}
	\end{equation}
\vspace{0.1mm}
	\begin{equation}\label{eq: type 6-b}
		\begin{split}
			\textnormal{Swap}_{ij} \textnormal{Swap}_{jk} \textnormal{Swap}_{pq} \textnormal{Swap}_{qr}&, \ 
			\textnormal{Swap}_{ij} \textnormal{Swap}_{jk} \textnormal{Swap}_{pq} \textnormal{Swap}_{pr},\\
			\textnormal{Swap}_{ij} \textnormal{Swap}_{ik} \textnormal{Swap}_{pq} \textnormal{Swap}_{pr}&,
			\quad i < j < k, \ p<q<r, \ i<p,\\	
			 &\quad\{i,j,k\} \cap \{p,q,r\} = \emptyset;	
		\end{split}
	\end{equation}
\vspace{2mm}
	\begin{equation}\label{eq: type 55}
		\begin{split}
			t \in \mathcal{B}_{ijklm},\ \quad \quad \quad i < j < k < l < m;
		\end{split}
	\end{equation}
	is a basis of the subspace of $M^{\text{Sw}_4}_n(\mathbb{C})$ of polynomials in the Swap$_{ij}$ of degree at most four. 
\end{proposition}

\begin{proof}
	For the spanning property of $\mathcal{B}_4,$ identify the products of the swap matrices with the corresponding permutations in $S_n.$ Note that the only permutations that can be written as a product of at most four transpositions that we omitted from $\mathcal{B}_4$ are the $5$-cycles of the form $(i\,m\,l\,k\,j)$ for $i<j<k<l<m.$ But these are in the span of $\mathcal{B}_4$ by the degree-reducing relation \eqref{eq:qdit-rel} with $d=4.$

	The proof of the linear independence of $\mathcal{B}_4$ again relies on the properties of the $4\times 4$ Gell-Mann matrices presented in Subsection \ref{subsec: gm4}.
	
	Suppose there is a linear dependence among the elements of $\mathcal{B}_4.$ 
	Then, using \eqref{eq: sw-gm 4}, express each of the appearing terms w.r.t.~the basis \eqref{eq: basis-gm d} consisting of different combinations of tensor products of the fifteen $4 \times 4$ Gell-Mann matrices.
	
	(a) First, consider the elements in \eqref{eq: type 8} and observe that for any choice of $i < j, k < l,  p < q, r < s$ with  $(i,j) < (k,l) < (p,q)<(r,s),$ the highest order terms in the expansion of Swap$_{ij}$Swap$_{kl}$Swap$_{pq}$Swap$_{rs}$ are of the form
	$$
	\lambda_a^i\, \lambda_a^j\, \lambda_b^k\, \lambda_b^l\, \lambda_c^p\, \lambda_c^q \, \lambda_d^r\, \lambda_d^s, \quad \quad a,b,c,d \in \{1,\ldots,15\}.
	$$
	By the product formula \eqref{eq: lambda4 x}, the elements in \eqref{eq: type 8} are the only ones that have terms of order eight and more precisely, for any choice of $i < j, k < l,  p < q, r < s$ with  $(i,j) < (k,l) < (p,q)<(r,s),$ the element Swap$_{ij}$Swap$_{kl}$Swap$_{pq}$Swap$_{rs}$  has the term $\lambda_1^i\, \lambda_1^j\, \lambda_2^k\, \lambda_2^l\, \lambda_3^p\, \lambda_3^q\lambda_4^r\, \lambda_4^s,$
	which does not appear in the expansion of any other element of $\mathcal{B}_4.$ Hence, the coefficients next to each of the elements in \eqref{eq: type 8} have to be zero. 
	
	(b) Now the elements in \eqref{eq: type 7} are the only ones in $\mathcal{B}_4$ that have terms of order seven (meaning with seven different positions $i,j,k,p,q,r,s$) in their expansion. For any choice of  $i < j < k, p<q, r<s$ with $p,q,r,s \notin \{i,j,k\}, (p,q)<(r,s),$ the highest order terms in the expansion of Swap$_{ij}$Swap$_{jk}$Swap$_{pq}$Swap$_{rs}$ are of the form
	$$
	\lambda_a^i\, \lambda_a^j\, \lambda_b^j\, \lambda_b^k\, \lambda_c^p\, \lambda_c^q\, \lambda_d^r\, \lambda_d^s, 
	\quad \quad a,b,c,d \in \{1,\ldots,15\},
	$$
	while for Swap$_{ij}$Swap$_{ik}$Swap$_{pq}$Swap$_{rs}$ they are of the form
	$$
	\lambda_a^i\, \lambda_a^j\, \lambda_b^i\, \lambda_b^k\, \lambda_c^p\, \lambda_c^q\, \lambda_d^r\, \lambda_d^s = \lambda_a^j\,\lambda_a^i\,  \lambda_b^i\, \lambda_b^k\, \lambda_c^p\, \lambda_c^q\, \lambda_d^r\, \lambda_d^s ,
	\quad \quad a,b,c,d \in \{1,\ldots,15\}.
	$$
	As noted, the structure constants $f^{a,b,c}$ and $d^{a,b,c}$ with $a,b,c\in\{1,\ldots,8\}$ coincide with the structure constants pertaining to the $3 \times 3$ Gell-Mann matrices.
	Hence, similar to part (b) of the proof of Proposition \ref{prop: indep3}, for any choice of $i < j < k, p<q, r<s$ with $p,q,r,s \notin \{i,j,k\}, (p,q)<(r,s),$  the element Swap$_{ij}$Swap$_{jk}$Swap$_{pq}$Swap$_{rs}$ has in its expansion
	\begin{align*}
		\lambda_2^i\, \lambda_2^j\, \lambda_3^j\, \lambda_3^k\, \lambda_5^p\, \lambda_5^q\, \lambda_9^r\, \lambda_9^s + \lambda_1^i\, \lambda_1^j\, \lambda_6^j\, \lambda_6^k\, \lambda_5^p\, \lambda_5^q\, \lambda_9^r\, \lambda_9^s =
		\mathfrak{i}\,\lambda_2^i\, \lambda_1^j\, \lambda_3^k\, \lambda_5^p\, \lambda_5^q\, \lambda_9^r\, \lambda_9^s + \frac{1}{2} \lambda_1^i\, \lambda_4^j\, \lambda_6^k\, \lambda_5^p\, \lambda_5^q\, \lambda_9^r\, \lambda_9^s.
	\end{align*}
	But Swap$_{ij}$Swap$_{ik}$Swap$_{pq}$ has in its expansion
	$$
	\lambda_1^i\, \lambda_1^j\, \lambda_3^i\, \lambda_3^k\, \lambda_5^p\, \lambda_5^q\, \lambda_9^r\, \lambda_9^s + \lambda_1^i\, \lambda_1^j\, \lambda_6^i\, \lambda_6^k\, \lambda_5^p\, \lambda_5^q\, \lambda_9^r\, \lambda_9^s =
	-\mathfrak{i}\,\lambda_2^i\, \lambda_1^j\, \lambda_3^k\, \lambda_5^p\, \lambda_5^q\, \lambda_9^r\, \lambda_9^s + \frac{1}{2} \lambda_1^i\, \lambda_4^j\, \lambda_6^k\, \lambda_5^p\, \lambda_5^q\, \lambda_9^r\, \lambda_9^s.
	$$
	By the same argument as in part (b) of the proof of Proposition \ref{prop: indep3}, all the coefficients next to the elements in \eqref{eq: type 7} are zero.
	
	(c) Now the elements in \eqref{eq: type 6-a} and \eqref{eq: type 6-b} are the only ones with terms of order six (i.e., with six different positions denoted by either $i,j,k,l,p,q$ or $i,j,k,p,q,r$) in their expansion.

	For the elements in \eqref{eq: type 6-a}, given any choice of $i < j < k < l,\, p<q$ with $p,q \notin \{i,j,k,l\},$ the highest order terms are of the form
	\begin{align*}
		\boldsymbol{\lambda} \cdot \lambda_a^p\, \lambda_a^q,
	\end{align*}
	where $\boldsymbol{\lambda}$ is a highest order term of an element of $\mathcal{B}_{ijkl}$ as in the proof of Proposition \ref{prop: indep3}.
	For the elements in \eqref{eq: type 6-b}, given any choice of $i < j < k, p<q<r, i<p$ with $\{i,j,k\} \cap \{p,q,r\} = \emptyset,$ the highest order terms of the three appearing types are
	\begin{align*}
		\textnormal{Swap}_{ij} \textnormal{Swap}_{jk} \textnormal{Swap}_{pq} \textnormal{Swap}_{qr}: \quad  
		\lambda_a^i\, \lambda_a^j\, \lambda_b^j\, \lambda_b^k\, \lambda_c^p\, \lambda_c^q\, \lambda_d^q\, \lambda_d^r\\
		\textnormal{Swap}_{ij} \textnormal{Swap}_{jk} \textnormal{Swap}_{pq} \textnormal{Swap}_{pr}: \quad  
		\lambda_a^i\, \lambda_a^j\, \lambda_b^j\, \lambda_b^k\, \lambda_c^p\, \lambda_c^q\, \lambda_d^p\, \lambda_d^r\\
		\textnormal{Swap}_{ij} \textnormal{Swap}_{ik} \textnormal{Swap}_{pq} \textnormal{Swap}_{pr}: \quad  
		\lambda_a^i\, \lambda_a^j\, \lambda_b^i\, \lambda_b^k\, \lambda_c^p\, \lambda_c^q\, \lambda_d^p\, \lambda_d^r.
	\end{align*}

First, consider \eqref{eq: type 6-b}. For fixed $i < j < k, p<q<r, i<p$ with $\{i,j,k\} \cap \{p,q,r\} = \emptyset,$ denote the coefficients next to the elements 
$$\textnormal{Swap}_{ij} \textnormal{Swap}_{jk} \textnormal{Swap}_{pq} \textnormal{Swap}_{qr}, \textnormal{Swap}_{ij} \textnormal{Swap}_{jk} \textnormal{Swap}_{pq} \textnormal{Swap}_{pr}, \textnormal{Swap}_{ij} \textnormal{Swap}_{ik} \textnormal{Swap}_{pq} \textnormal{Swap}_{pr}$$ 
by $\beta_1,\beta_2$ and $\beta_3$ respectively. Clearly, these are the only elements in \eqref{eq: type 6-b} whose highest order terms involve precisely the positions $i,j,k,p,q,r.$
So comparing the coefficients next to the basis elements $\lambda_2^i\, \lambda_{10}^j\, \lambda_{11}^k\, \lambda_1^p\, \lambda_4^q\, \lambda_7^r, \, \lambda_2^i\, \lambda_9^j\, \lambda_{11}^k\, \lambda_1^p\, \lambda_5^q\, \lambda_7^r$ and $\lambda_2^i\, \lambda_9^j\, \lambda_{11}^k\, \lambda_1^p\, \lambda_4^q\, \lambda_7^r$ give the following equations
	\begin{align*}
		&\lambda_2^i\, \lambda_{10}^j\, \lambda_{11}^k\, \lambda_1^p\, \lambda_4^q\, \lambda_7^r : \quad -\beta_1 -\beta_2 + \beta_3 =0,\\
		&\lambda_2^i\, \lambda_9^j\, \lambda_{11}^k\, \lambda_1^p\, \lambda_5^q\, \lambda_7^r : \quad -\beta_1 +\beta_2 + \beta_3 =0,\\
		&\lambda_2^i\, \lambda_9^j\, \lambda_{11}^k\, \lambda_1^p\, \lambda_4^q\, \lambda_7^r
: \quad -\beta_1 +\beta_2-\beta_3 =0.
	\end{align*}
 The above system has a unique solution $\beta_1=\beta_2=\beta_3=0.$
 Note that each of the highest order terms of the elements in \eqref{eq: type 6-a} necessarily has one of the Gell-Mann matrices $\lambda$ repeated twice. So the coefficients next to the basis elements $\lambda_2^i\, \lambda_{10}^j\, \lambda_{11}^k\, \lambda_1^p\, \lambda_4^q\, \lambda_7^r, \, \lambda_2^i\, \lambda_9^j\, \lambda_{11}^k\, \lambda_1^p\, \lambda_5^q\, \lambda_7^r$ and $\lambda_2^i\, \lambda_9^j\, \lambda_{11}^k\, \lambda_1^p\, \lambda_4^q\, \lambda_7^r$ in the expansion of the elements in  \eqref{eq: type 6-a} are zero.
 We conclude that the coefficients next to the elements
 $$\textnormal{Swap}_{ij} \textnormal{Swap}_{jk} \textnormal{Swap}_{pq} \textnormal{Swap}_{qr}, \textnormal{Swap}_{ij} \textnormal{Swap}_{jk} \textnormal{Swap}_{pq} \textnormal{Swap}_{pr}, \textnormal{Swap}_{ij} \textnormal{Swap}_{ik} \textnormal{Swap}_{pq} \textnormal{Swap}_{pr},$$ 
are zero and by analogy, the coefficients next to all the elements in  \eqref{eq: type 6-b} are zero.

Now consider \eqref{eq: type 6-a}. For fixed $i < j < k < l,\, p<q$ with $p,q \notin \{i,j,k,l\},$ denote the coefficients next to the elements
\begin{align*}
		&\textnormal{Swap}_{ij} \textnormal{Swap}_{jk} \textnormal{Swap}_{kl}\textnormal{Swap}_{pq}, \ 
		\textnormal{Swap}_{ij} \textnormal{Swap}_{jl} \textnormal{Swap}_{kl}\textnormal{Swap}_{pq},\ 
		\textnormal{Swap}_{ik} \textnormal{Swap}_{jk} \textnormal{Swap}_{jl}\textnormal{Swap}_{pq},\\
		&\textnormal{Swap}_{ik} \textnormal{Swap}_{kl} \textnormal{Swap}_{jl}\textnormal{Swap}_{pq},\ 
		\textnormal{Swap}_{il} \textnormal{Swap}_{jl} \textnormal{Swap}_{jk}\textnormal{Swap}_{pq},\
		\textnormal{Swap}_{il} \textnormal{Swap}_{kl} \textnormal{Swap}_{jk}\textnormal{Swap}_{pq}
\end{align*}
by $\alpha_1,\ldots,\alpha_6$ respectively.
Clearly, these are the only elements in \eqref{eq: type 6-a} whose highest order terms involve precisely the positions $i,j,k,l,p,q.$
So comparing the coefficients next to several basis elements 
of the form $\lambda_a^i\, \lambda_b^j\, \lambda_c^k\, \lambda_d^l\, \lambda_e^p\, \lambda_e^q$ gives the following equations
\begin{equation}\label{eq: sys}
\begin{split}
	&\lambda_{11}^i\, \lambda_1^j\, \lambda_5^k\, \lambda_{13}^l\, \lambda_2^p\, \lambda_2^q:\quad \beta_1-\beta_6=0,\\
	&\lambda_{11}^i\,\lambda_1^j\, \lambda_{13}^k\, \lambda_5^l\, \lambda_2^p\, \lambda_2^q:\quad \beta_2 - \beta_5=0,\\
	&\lambda_{11}^i\, \lambda_5^j\, \lambda_1^k\, \lambda_{13}^l\, \lambda_2^p\, \lambda_2^q:\quad \beta_3 - \beta_4=0,\\
	&\lambda_{11}^i\, \lambda_1^j\, \lambda_3^k\, \lambda_9^l\, \lambda_2^p\, \lambda_2^q: \quad \beta_1 - \beta_3 - \beta_4=0,\\
	&\lambda_{11}^i\, \lambda_3^j\, \lambda_1^k\, \lambda_9^l\, \lambda_2^p\, \lambda_2^q: \quad -\beta_1 + \beta_3 + \beta_4 - \beta_6 =0,\\
	&\lambda_{11}^i\, \lambda_9^j\, \lambda_1^k\, \lambda_3^l\, \lambda_2^p\, \lambda_2^q:\quad -\beta_1 + \beta_2 + \beta_5 - \beta_6 =0.
\end{split}
\end{equation}
The above system has a unique solution $\beta_1= \cdots \beta_6 =0.$ Hence, by analogy, the coefficients next to all of the elements in  \eqref{eq: type 6-a} are zero.

(d) The quartics in \eqref{eq: type 55} are now the only ones in $\mathcal{B}_4$ that have terms of degree five.
For fixed $i<j<k<l<m$ denote the coefficients next to the quartics in \eqref{eq: quart} by $\gamma_1,\ldots,\gamma_{23}$ respectively.
Clearly, these are the only elements in \eqref{eq: type 55} whose highest order terms involve precisely the positions $i,j,k,l,m.$ Similar to before, we now compare the coefficients next to several basis elements 
of the form $\lambda_a^i\, \lambda_b^j\, \lambda_c^k\, \lambda_d^l\, \lambda_e^m$ to get a system of equations. By symmetry note that  if  $\lambda_a^i\, \lambda_b^j\, \lambda_c^k\, \lambda_d^l\, \lambda_e^m$ is, e.g., a term in the expansion of $\textnormal{Swap}_{ij} \textnormal{Swap}_{jk} \textnormal{Swap}_{kl} \textnormal{Swap}_{lm}$ and $\sigma$ is a permutation of the positions $i,j,k,l,m,$ then $\lambda_a^{\sigma(i)}\, \lambda_b^{\sigma(j)}\, \lambda_c^{\sigma(k)}\, \lambda_d^{\sigma(l)}\, \lambda_e^{\sigma(m)}$ is a term in the expansion of $$\textnormal{Swap}_{\sigma(i),\sigma(j)} \textnormal{Swap}_{\sigma(j),\sigma(k)} \textnormal{Swap}_{\sigma(k),\sigma(l)} \textnormal{Swap}_{\sigma(l),\sigma(m)}.$$
By this observation it is easy to quickly deduce several equations, e.g.,
\begin{equation}\label{eq: t5-eq}
	\begin{split}
	&\lambda_{11}^i\, \lambda_9^j\, \lambda_4^k\, \lambda_6^l\, \lambda_3^m: \ 
	-\gamma_1 + \gamma_{21}+\gamma_4 = 0,\\
	&\lambda_{11}^i\, \lambda_9^j\, \lambda_4^k\, \lambda_3^l\, \lambda_6^m: \ 
	-\gamma_2 + \gamma_{22} - \gamma_{23} +\gamma_3 = 0,\\
	&\lambda_{11}^i\, \lambda_9^j\, \lambda_6^k\, \lambda_4^l\, \lambda_3^m: \ 
	-\gamma_3 + \gamma_{22} - \gamma_{19} +\gamma_6 = 0,\\
	&\lambda_{11}^i\, \lambda_9^j\, \lambda_6^k\, \lambda_3^l\, \lambda_4^m: \ 
	-\gamma_4 + \gamma_{21} - \gamma_{20} +\gamma_5 = 0,\\
	&\lambda_{11}^i\, \lambda_9^j\, \lambda_3^k\, \lambda_4^l\, \lambda_6^m: \ 
	-\gamma_5 + \gamma_{20} +\gamma_1 = 0.
	\end{split}
\end{equation}
    The remaining $19$ equations are computed by analogy.
	To apply this technique correctly it is important to keep in mind that  the equation given by $\lambda_{11}^i\, \lambda_9^j\, \lambda_4^k\, \lambda_6^l\, \lambda_3^m$ is in fact $-\gamma_1 + \gamma_{21} - \gamma_{24} +\gamma_4 = 0,$ where $\gamma_{24} = 0$ is the coefficient corresponding to $\textnormal{Swap}_{im} \textnormal{Swap}_{lm} \textnormal{Swap}_{kl} \textnormal{Swap}_{jk}$ (that we excluded from the basis, but that we need to keep in mind to calculate the following equations correctly). 
	Combining all the $24$ equations \eqref{eq: t5-eq} with the ones obtained by considering the terms
	\begin{align*}
		&\lambda_{11}^i\, \lambda_1^j\, \lambda_{15}^k\, \lambda_3^l\, \lambda_9^m,\ 
		\lambda_{11}^i\, \lambda_1^j\, \lambda_{15}^k\, \lambda_9^l\, \lambda_3^m,\ 
		\lambda_{11}^i\, \lambda_1^j\, \lambda_3^k\, \lambda_{15}^l\, \lambda_9^m,\\
		&\lambda_{11}^i\, \lambda_1^j\, \lambda_9^k\, \lambda_3^l\, \lambda_{15}^m,\ 
		\lambda_{11}^i\, \lambda_{15}^j\, \lambda_1^k\, \lambda_9^l\, \lambda_3^m,\ 
		\lambda_{11}^i\, \lambda_9^j\, \lambda_{15}^k\, \lambda_1^l\, \lambda_3^m,
	\end{align*}
	we get a system with unique solution $\gamma_1=\cdots= \gamma_{23}=0$. Hence, by analogy, the coefficients next to all of the quartics in \eqref{eq: type 55} are zero as well.
	
	(e) Now the cubics \eqref{eq: cubics} are the only ones with terms of degree four in their expansion. But comparing the coefficients next to the basis elements 
	\begin{align*}
		&\lambda_{11}^i\, \lambda_1^j\, \lambda_5^k\, \lambda_{13}^l, \ \lambda_{11}^i\,\lambda_1^j\, \lambda_{13}^k\, \lambda_5^l,\
		\lambda_{11}^i\, \lambda_5^j\, \lambda_1^k\, \lambda_{13}^l,\ 
		\lambda_{11}^i\, \lambda_1^j\, \lambda_3^k\, \lambda_9^l,\ 
		\lambda_{11}^i\, \lambda_3^j\, \lambda_1^k\, \lambda_9^l,\ 
		\lambda_{11}^i\, \lambda_9^j\, \lambda_1^k\, \lambda_3^l
	\end{align*}
	 gives back the system \eqref{eq: sys}, which implies that the coefficients next to all of the cubics are zero as well.
	We are left with a linear dependence involving terms of degree at most two, which contradicts linear independence of $\mathcal{B}_2$ as shown in Proposition \ref{propo: 2-indep}.
\end{proof}

\section{Explicit eigenvalue computation for clique Hamiltonians of general $d$-row partitions}\label{appC}

Here, we give an alternative and elementary method to compute the character value $\chi_\lambda ((i\,j))$ of a transposition $(i \,j)$ using the Murnaghan-Nakayama rule \cite[Section 9.9.1]{Probook}. Using formula \eqref{eq:etafromchi}, we then again compute, for any partition $\lambda \vdash n$ with $d$ rows, 
 the eigenvalue $\eta_\lambda$ from  Lemma \ref{lemma scalar}.

\subsection{The Murnaghan-Nakayama rule}\label{sec:MNR}

	To compute the value of the character $\chi_\lambda$ at the conjugacy class of transpositions we use the non-recursive version of the Murnaghan-Nakayama rule. It states that
	\begin{align}\label{eq:mn}
		\chi_\lambda \big( (i\ j) \big) =
		\sum_T (-1)^{\het(T)},
	\end{align}
where the sum runs over all tableaux $T$ of shape $\lambda$ that satisfy:
\begin{itemize}
	\item the boxes of $T$ are filled with numbers $1,2,\ldots, n-1$ such that $1$ appears twice and all the others appear once,
	\item the numbers in every row and column are weakly increasing.
\end{itemize}
Here $\het(T)$ is one if both $1^\prime$s are in the first column and it is zero if they are in the first row. 

Clearly, the set of all such tableaux is in bijection with the set of all standard Young tableaux of shape $\lambda.$ This means that 
\begin{align*}
	\chi_\lambda \big( (i\ j) \big) =\ 
	&\#(\text{standard Young tableaux with }1 \text{ and } 2 \text{ in the first row}) - \\
	&\#(\text{standard Young tableaux with }1 \text{ and } 2 \text{ in the first column}).
\end{align*}

It is easy to see that the standard Young Tableaux with $1$ and $2$ in the first row are in bijection with the standard Young tableaux of the shape that we get by removing the first two boxes from the first row of $\lambda.$ Similarly, the standard Young Tableaux with $1$ and $2$ in the first column are in bijection with the standard Young tableaux of the shape that we get by removing the first two boxes from the first column of $\lambda.$
To count these we use the hook-length formula for skew shaped Young tableaux \cite{Nar,MPP}. 

For a partition $\mu$ denote by $[\mu]$ the diagram (i.e., tableaux without numbers) of shape $\mu.$ If $[\mu]$ is the diagram that we cut out of $[\lambda],$ then denote the resulting skew shaped diagram by $[\lambda / \mu]$ and the number of all standard Young tableaux of shape $\lambda / \mu$ by $f^{\lambda / \mu}.$ For a box  $(i,j) \in [\mu]$ such that the boxes $(i+1,j),(i,j+1),(i+1,j+1) \in [\lambda]$ are not in $[\mu],$ we say that an \textbf{excited move} with respect to $\lambda$ is the replacement of $[\mu]$ by 
$$
\Big([\mu] \backslash \{(i\ j)\}\Big) \cup  \{(i+1\ j+1)\}.
$$
An \textbf{excited diagram} of shape $\lambda / \mu$ is a diagram contained in $[\lambda]$ that can be obtained from $[\mu]$ with a series of excited moves. Denote by $\mathcal{E}(\lambda / \mu)$ the set of all excited diagrams of shape $\lambda / \mu$ (here $\mathcal{E}(\lambda / \mu)$ is empty unless $[\mu] \subseteq [\lambda]$) (see Example \ref{ex:excited}).

The hook-length formula for skew shaped tableaux \cite[Theorem 1.2]{MPP} states that
\begin{align}\label{eq: hl-s}
f^{\lambda / \mu} = |\lambda / \mu |! \sum_{D \in \mathcal{E}(\lambda / \mu)} \prod_{ij} \frac{1}{\text{hook}_{\lambda / \mu}(i,j)}.
\end{align}

In our case, $\chi_\lambda((i,j))=f^{\lambda / \mu}$ where $\mu=(1,1)$ is the two-row partition of two.
Hence, if $(i_1,j_1)$ and $(i_2,j_2)$ are the two distinguished squares in an excited diagram $D$ of shape $\lambda / \mu,$ the summand in \eqref{eq: hl-s} pertaining to $D$ can be expressed as
\begin{align}\label{eq: 2-hook}
	\frac{\chi_\lambda(e)}{n!} \, \text{hook}_{\lambda / \mu}(i_1,j_1)\,\text{hook}_{\lambda / \mu}(i_2,j_2).
\end{align}

\begin{example}\label{ex:excited}
For $n=7$ let $\lambda = (4,3)$ and $\mu = (2).$ Then there are three excited diagrams of shape $\lambda/\mu$,
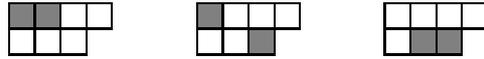
\begin{center}
\ytableausetup{smalltableaux}
 \begin{ytableau}
 *(gray) & *(gray) & & \\
  & &\\
\end{ytableau}
\qquad
\ytableausetup{smalltableaux}
 \begin{ytableau}
 *(gray) & & & \\
  & & *(gray) \\
\end{ytableau}
\qquad
\ytableausetup{smalltableaux}
 \begin{ytableau}
 {} &  & & \\
  & *(gray)&*(gray)\\
\end{ytableau}
\end{center}
\end{example}

\subsection{Clique eigenvalue computation}
Next we present an alternative method to prove  Proposition \ref{prop:etaD}, which we restate below.
\begin{repprop}{prop:etaD}
Let $\eta_\lambda$ be as in Lemma \ref{lemma scalar}. For any $\lambda \vdash n$ with rows $\lambda_1\ge\cdots\ge \lambda_d$,
\begin{equation}\label{eq:etaD-OLD}
\eta_\lambda=
n^2 +\frac{d(d-1)(2d-1)}{6}
-\sum_{k=1}^d\big( \lambda_k - (k-1)\big)^2.
\end{equation}
\end{repprop}

\begin{proof}[Sketch of proof] Let $\lambda \vdash n$ be a partition with at most $d$ rows $\lambda_1\ge\cdots\ge \lambda_d \geq 0.$
To calculate $\eta_\lambda$ through the formula \eqref{eq:etafromchi} we first compute the value $\chi_\lambda((i\,j))$ of the character $\chi_\lambda$ at the conjugacy class of transpositions. 
For that we use the Murnaghan-Nakayama rule as presented in Subsection \ref{sec:MNR}. 

First  consider the excited diagrams of shape $\lambda / \mu$ with $\mu = (2).$ We only tackle the most general case with $\lambda_d \geq d+1$ (i.e., the case that gives the most excited diagrams).  By a similar reasoning as before, the box $(1,1),$ can be moved only after the box $(1,2)$ had already been moved. If $(1,1)$ is, say, in position $(k,k)$ for $k=1,\ldots,d,$ this means that $(1,2)$ must have been moved to one of the $d-k+1$ positions $(k,k+1),\ldots,(d,d+1).$ So if $(1,1)$ is moved to $(k,k)$ and $(1,2)$ is moved to $(j,j+1)$ for some $j=k,\ldots,d,$ the contribution to \eqref{eq: hl-s} of this excited diagram computed via \eqref{eq: 2-hook} is 
\begin{align*}
	\frac{\chi_\lambda(e)}{n!} \, \text{hook}_{\lambda / (2)}(k,k)\,\text{hook}_{\lambda / (2)}(j,j+1) = \frac{\chi_\lambda(e)}{n!} \, (\lambda_k-(k-1)+d-k)(\lambda_j-j+d-j).
\end{align*}
Hence, 
$$
f^{\lambda/(2)} = \frac{\chi_\lambda(e)}{n(n-1)}  \sum_{k=1}^d \sum_{j=k}^d (\lambda_k-2k+d+1)(\lambda_j-2j+d).
$$

For the excited diagrams of shape $\lambda / \mu$ with $\mu = (1,1),$ we again only consider the case with $\lambda_d \geq d-1,$ which gives the most excited diagrams. In this case, if the box $(1,1)$ is moved to position $(k,k)$ for some $k=1,\ldots,d-1,$ the box $(2,1)$ must have been moved to one of the $d-k$ positions $(k+1,k),\ldots,(d,d-1).$
So if $(1,1)$ is moved to $(k,k)$ and $(2,1)$ is moved to $(j+1,j)$ for some $j=k,\ldots,d-1,$ the contribution to \eqref{eq: hl-s} of this excited diagram computed via \eqref{eq: 2-hook} is  now
\begin{align*}
	&\frac{\chi_\lambda(e)}{n!} \, \text{hook}_{\lambda / (1,1)}(k,k)\,\text{hook}_{\lambda / (1,1)}(j+1,j) \\ & =\frac{\chi_\lambda(e)}{n!} \, (\lambda_k-(k-1)+d-k)(\lambda_j-(j-2)+d-j).
\end{align*}
Hence, 
$$
f^{\lambda/(1,1)} = \frac{\chi_\lambda(e)}{n(n-1)}  \sum_{k=1}^{d-1} \sum_{j=k+1}^d (\lambda_k-2k+d+1)(\lambda_j-2j+d+2).
$$
Putting all together we get
\begin{align*}
	\eta_\lambda & = 2 \binom{n}{2} \bigg(1-\frac{\chi_\lambda\big((i\ j)\big)}{\chi_\lambda(e)}\bigg) =
	2 \binom{n}{2} \bigg(1-\frac{f^{\lambda/(2)}-f^{\lambda/(1,1)}}{\chi_\lambda(e)}\bigg)\\ & =
	n(n-1)-\sum_{k=1}^d \sum_{j=k}^d (\lambda_k-2k+d+1)(\lambda_j-2j+d)\\ \phantom{{}={}}  & +
	\sum_{k=1}^{d-1} \sum_{j=k+1}^d(\lambda_k-2k+d+1)(\lambda_j-2j+d+2)\\
& = n^2-2n -\sum_{k=1}^d( \lambda_k^2 - 2 k \lambda_k)\\
& = n^2 +\frac{d(d-1)(2d-1)}{6}
-\sum_{k=1}^d\big( \lambda_k - (k-1)\big)^2.
\tag*{\qed}
\end{align*}
\renewcommand{\qedsymbol}{}
\end{proof}


\begin{thebibliography}{99}

\bibitem[ACGNORVL23+]{grinko}
R. Allerstorfer, M. Christandl, D. Grinko, I. Nechita, M. Ozols, D. Rochette, P. Verduyn Lunel, \textit{Monogamy of highly symmetric states},  
preprint \texttt{arXiv:2309.16655}. \url{https://doi.org/10.48550/arXiv.2309.16655}

\bibitem[AMG20]{AMG20}
A. Anshu, D. Gosset, K. Morenz, \textit{Beyond Product State Approximations for a Quantum Analogue of Max Cut}, 15th Conference on the Theory of Quantum Computation, Communication and Cryptography (TQC 2020), Leibniz International Proceedings in Informatics 158, 7:1--7:15, Schloss Dagstuhl-Leibniz-Zentrum f\"ur Informatik, 2020.
\url{https://doi.org/10.4230/LIPIcs.TQC.2020.7}

 \bibitem[Aue94]{Aue}
A. Auerbach, 
\textit{Interacting electrons and quantum magnetism}, 
Springer-Verlag, New York, 1994. \url{https://doi.org/10.1007/978-1-4612-0869-3}

\bibitem[BAMC09]{BAMC}
K. S. D. Beach, F. Alet, M. Mambrini, S. Capponi, 
\textit{${\rm SU}(N)$ Heisenberg
model on the square lattice: A continuous-$N$ quantum Monte Carlo study},
Phys. Rev. B 80 (2009) 184401. \url{https://doi.org/10.1103/PhysRevB.80.184401}

\bibitem[BPT13]{BPT13}
G. Blekherman, P.A. Parrilo, R.R. Thomas, Rekha R. (ed.),
\textit{Semidefinite optimization and convex algebraic geometry}, Society for Industrial and Applied Mathematics (SIAM), 2013. \url{https://doi.org/10.1137/1.9781611972290}

\bibitem[BDZ08]{BDZ}
I. Bloch, J. Dalibard, W. Zwerger, 
\textit{Many-body physics with ultracold gases}, 
Rev. Mod. Phys. 80 (2008) 885. \url{https://doi.org/10.1103/RevModPhys.80.885}

\bibitem[BCEHK24]{BCEHK24} A.~Bene Watts, A.~Chowdhury, A.~Epperly, J.~W.~Helton, I.~Klep,
 \textit{Relaxations and Exact Solutions to Quantum Max Cut via the Algebraic Structure of Swap Operators}, Quantum 8 (2024) 1352, 88pp. \url{https://doi.org/10.22331/q-2024-05-22-1352}

\bibitem[BCP97]{magma}
W. Bosma, J. Cannon, C. Playoust, 
\textit{The Magma algebra system. I. The user language}, 
J. Symbolic Comput. 24 (1997) 235--265. \url{https://doi.org/10.1006/jsco.1996.0125}


\bibitem[BGKT19]{BGKT19}
 S. Bravyi, D. Gosset, R. K\"{o}nig, K. Temme, \textit{Approximation
 algorithms for quantum many-body problems}, J. Math. Phys. 60 (2019) 032203.
 \url{https://doi.org/10.1063/1.5085428}

\bibitem[BH13]{BH13}
F. G. S. L. Brandao, A. W. Harrow, \textit{Product-state approximations to
 quantum ground states}, Proceedings of the Forty-Fifth Annual ACM Symposium
 on Theory of Computing (2013) 871--880.
 \url{https://doi.org/10.1145/2488608.2488719}

\bibitem[BKP16]{BKP16}
S. Burgdorf, I. Klep, J. Povh,
\textit{Optimization of polynomials in non-commuting variables},
Springer, 2016. \url{https://doi.org/10.1007/978-3-319-33338-0}

\bibitem[CJKKW23+]{CJKKW23+}
C.~Carlson, Z.~Jorquera, A.~Kolla, S.~Kordonowy, S.~Wayland, \textit{Approximation Algorithms for Quantum Max-$d$-Cut},  
preprint \texttt{arXiv:2309.10957}. \url{https://doi.org/10.48550/arXiv.2309.10957}

\bibitem[dCP76]{dCP}
C. de Concini, C. Procesi: 
\textit{A characteristic free approach to invariant theory}, 
Adv. Math. 21 (1976) 330--354. \url{https://doi.org/10.1016/S0001-8708(76)80003-5}

\bibitem[DLTW08]{DLTW}
A. C. Doherty, Y.-C. Liang, B. Toner, S. Wehner: 
\textit{The quantum moment problem and bounds on entangled multi-prover games}, 
in Proceedings of the 2008 IEEE 23rd Annual Conference on Computational Complexity, 199--210, IEEE Computer Society, 2008. \url{https://doi.org/10.1109/CCC.2008.26}

\bibitem[FJ97]{FJ}
A. Frieze, M. Jerrum, 
\textit{Improved approximation algorithms for MAX $k$-CUT and MAX BISECTION}, 
Algorithmica 18 (1997), 67--81. \url{https://doi.org/10.1007/BF02523688}

\bibitem[Fro00]{Fro01} 
G.~Frobenius,
\textit{\"Uber die Charaktere der symmetrischen Gruppe}, 
Berl. Ber. (1900) 516--534. \url{https://archive.org/details/sitzungsberichte1900deut}

\bibitem[FH91]{FH}
W. Fulton, J. Harris, 
{\it Representation theory. A first course},
Grad. Texts in Math., Read. Math. 129,
Springer-Verlag, New York, 1991. \url{https://doi.org/10.1007/978-1-4612-0979-9}

 \bibitem[GH24+]{GH}
S. Gharibian, C. Hecht, 
\textit{Hardness of approximation for ground state problems},  
preprint \texttt{arXiv:2411.04874}. \url{https://doi.org/10.48550/arXiv.2411.04874}

\bibitem[GK12]{GK12}
S. Gharibian, J. Kempe, \textit{Approximation algorithms for QMA
complete problems}, SIAM J. Comput.~41 (2012) 1028--1050. \url{https://doi.org/10.1137/110842272}

\bibitem[GP19]{GP19}
S. Gharibian, O. Parekh, \textit{Almost Optimal Classical Approximation Algorithms for a Quantum Generalization of Max-Cut}, 
Approximation, Randomization, and Combinatorial Optimization. Algorithms and Techniques (APPROX/RANDOM 2019), 
Leibniz International Proceedings in Informatics 145, 31:1-31:17, Schloss Dagstuhl–Leibniz-Zentrum f\"ur Informatik, 2019. \url{https://doi.org/10.4230/LIPIcs.APPROX-RANDOM.2019.31}


\bibitem[GL20]{GL20}
E. Giannelli, S. Law, \emph{Sylow branching coefficients for symmetric groups}, 
J. Lond. Math. Soc 103 (2020) 697--728. \url{https://doi.org/10.1112/jlms.12389}

\bibitem[GSS25+]{gribling}
S. Gribling, L. Sinjorgo, R. Sotirov:
\textit{Improved approximation ratios for the Quantum Max-Cut problem on general, triangle-free and bipartite graphs}, 
preprint \texttt{arXiv:2504.11120}. \url{https://doi.org/10.48550/arXiv.2504.11120}

\bibitem[GW95]{GW95}
M. X. Goemans, D. P. Williamson, \textit{Improved approximation algorithms
 for maximum cut and satisfiability problems using semidefinite programming},
 J.~ACM, 42 (1995) 1115--1145. \url{https://doi.org/10.1145/227683.227684}

\bibitem[Gri+]{Gri}
D. Grinberg:
\textit{Rook sums in the symmetric group algebra},  preprint \texttt{arXiv:2507.22386}.
\url{https://doi.org/10.48550/arXiv.2507.22386}

\bibitem[GNW21]{GNW21}
D. Gross, S. Nezami, M. Walter,
\textit{Schur-Weyl duality for the Clifford group with applications: property testing, a robust Hudson theorem, and de Finetti representations},
Commun. Math. Phys. 385 (2021) 1325--1393. \url{https://doi.org/10.1007/s00220-021-04118-7}

\bibitem[HM17]{HM17}
A. W. Harrow, A. Montanaro, \textit{Extremal eigenvalues of local
 Hamiltonians}, Quantum 1 (2017) 6. \url{https://doi.org/10.22331/q-2017-04-25-6}

\bibitem[HO22]{HO22}
  M. B. Hastings, R. O’Donnell, \textit{Optimizing strongly interacting
 fermionic Hamiltonians}, Proceedings of the 54th Annual ACM SIGACT Symposium on Theory of Computing  (2022) 776--789. \url{https://doi.org/10.1145/3519935.3519960}

\bibitem[HKL20]{HKL}
D. Henrion, M. Korda, J.-B. Lasserre: 
\textit{The moment-SOS hierarchy: Lectures in probability, statistics, computational geometry, control and nonlinear PDEs}, Optimization and Its Applications 4, World Scientific, 2020. \url{https://doi.org/10.1142/q0252}

\bibitem[HM04]{HM}
J. W. Helton, S. A. McCullough,
\textit{A Positivstellensatz for non-commutative polynomials},
Trans. Amer. Math. Soc. 356 (2004) 3721--3737.
\url{https://doi.org/10.1090/S0002-9947-04-03433-6}

\bibitem[HJ85]{HJ}
R. A. Horn, C. R. Johnson, 
\textit{Matrix analysis},
Cambridge University Press, Cambridge, 1985.
\url{https://doi.org/10.1017/CBO9780511810817}

\bibitem[HTPG24+]{HTPG}
F. Huber, K. Thompson, O. Parekh, S. Gharibian, \textit{Second order cone relaxations for quantum Max Cut},  
preprint \texttt{arXiv:2411.04120}. \url{https://doi.org/10.48550/arXiv.2411.04120}

\bibitem[IO02]{IO}
V. Ivanov, G. Olshanski, 
\textit{Kerov’s central limit theorem for the Plancherel measure on Young diagrams}, 
Symmetric Functions 2001: Surveys of Developments and Perspectives, NATO Science Series 74, Springer, Dordrecht, 2002. \url{https://doi.org/10.1007/978-94-010-0524-1_3}

\bibitem[JSZ22+]{jakab}
D. Jakab, Adrian Solymos, Z. Zimbor{\'a}s: 
\textit{Extendibility of Werner states}, preprint \texttt{arXiv:2208.13743}.
\url{https://doi.org/10.48550/arXiv.2208.13743}

\bibitem[KT07]{KT}
N. Kawashima, Y. Tanabe, 
\textit{Ground States of the ${\rm SU}(N)$ Heisenberg
Model}, 
Phys. Rev. Lett. 98 (2007) 057202.
\url{https://doi.org/10.1103/PhysRevLett.98.057202}

 \bibitem[KKR06]{KKR06}
  J. Kempe, A. Yu. Kitaev, O. Regev, \textit{The complexity of the
 local hamiltonian problem},
 SIAM J. Comput.~35 (2006) 1070--1097.
 \url{https://doi.org/10.1137/S0097539704445226}

\bibitem[Kho02]{Kho02}
S. Khot, \textit{On the power of unique 2-prover 1-round games}, Proceedings of
 the thiry-fourth annual ACM symposium on Theory of computing (2002) 767--775.
 \url{https://doi.org/10.1145/509907.510017}

 \bibitem[KLMS12]{KLMS12}
M. Khovanov, A. D. Lauda, M. Mackaay, M. Stošić, \emph{Extended graphical calculus for categorified quantum $\mathfrak{sl}(2)$}, Mem. Amer. Math. Soc. 1029, 87p., 2012.
\url{https://doi.org/10.1090/S0065-9266-2012-00665-4}


\bibitem[Kin23]{Kin23}
 R. King,
 \textit{An Improved Approximation Algorithm for Quantum
 Max-Cut on Triangle-Free Graphs}, Quantum 7 (2023), 1180.
 \url{https://doi.org/10.22331/q-2023-11-09-1180}

 \bibitem[KSV02]{KSV02}
 A. Yu. Kitaev, A. H. Shen, M. N. Vyalyi, \textit{Classical and
 quantum computation}, Graduate studies in mathematics 47, American
 Mathematical Society, Providence, RI, 2002. 
 \url{https://doi.org/10.1090/gsm/047}

  \bibitem[LVV15]{LVV15}
 Z. Landau, U. Vazirani, T. Vidick, \textit{A polynomial time algorithm
 for the ground state of one-dimensional gapped local Hamiltonians}, Nat.
 Phys.~11 (2015) 566--569.  \url{https://doi.org/10.1038/nphys3345}
 
\bibitem[Lsa08]{Lassalle}
M. Lassalle, 
\textit{An explicit formula for the characters of the symmetric group}, 
Math. Ann. 340 (2008) 383--405. \url{https://doi.org/10.1007/s00208-007-0156-5}

\bibitem[Lse01]{Las}
J.-B. Lasserre, 
\textit{Global optimization with polynomials and the problem of moments},
SIAM J. Optim. 11 (2001) 796--817. \url{https://doi.org/10.1137/S1052623400366802}

\bibitem[Lau09]{Lau}
M. Laurent: 
\textit{Sums of squares, moment matrices and optimization over polynomials}, 
in Emerging applications of algebraic geometry 157--270, IMA Vol. Math. Appl. 149, Springer, New York, 2009. \url{https://doi.org/10.1007/978-0-387-09686-5_7}

\bibitem[Lee22]{Lee22}
 E. Lee, \textit{Optimizing quantum circuit parameters via SDP}, 33rd International Symposium on Algorithms and Computation  (ISAAC 2022), Leibniz International Proceedings in Informatics 248, 48:1-48:16, Schloss Dagstuhl-Leibniz-Zentrum f\"ur Informatik, 2022. 
 \url{https://doi.org/10.4230/LIPIcs.ISAAC.2022.48}

 \bibitem[LM16]{LM16}
  F. Levkovich-Maslyuk, \textit{The Bethe Ansatz}, J. Phys. A: Math. Theor.~49 (2016) 323004. \url{https://doi.org/10.1088/1751-8113/49/32/323004}

\bibitem[LM62]{LM62}
E. Lieb, D. Mattis, \textit{Ordering energy levels of interacting
 spin systems}, J. Math. Phys.~3 (1962) 749--751. \url{https://doi.org/10.1063/1.1724276}

\bibitem[MP14]{nauty}
B.D. McKay, A. Piperno,
\textit{Practical Graph Isomorphism, II},
J. Symbolic Comput. 60 (2014) 94--112.
\url{https://doi.org/10.1016/j.jsc.2013.09.003}

\bibitem[Mor86]{Mor86}
F.~Mora, \textit{Groebner bases for noncommutative polynomial rings},
  Algebraic algorithms and error correcting codes (Grenoble, 1985), 353--362, Lecture Notes in Comput. Sci. 229, Springer,
 Berlin,
 1986.
 \url{https://doi.org/10.1007/3-540-16776-5_740}

\bibitem[MPP18]{MPP}
A. Morales, I. Pak, G. Panova, 
\textit{Hook formulas for skew shapes I. q-analogues and bijections},
J. Combin. Theory Ser. A 154 (2018), 350--405.
\url{https://doi.org/10.1016/j.jcta.2017.09.002}

\bibitem[Nar14]{Nar}
H. Naruse, 
\textit{Schubert calculus and hook formula},
73rd S{\'e}m. Lothar. Combin., Strobl, Austria, 2014. Talk slides available at 
\url{https://www.mat.univie.ac.at/~slc/wpapers/s73vortrag/naruse.pdf}

\bibitem[NPA08]{NPA08}
M. Navascu\'es, S. Pironio, A. Ac\'in,
\textit{A convergent hierarchy of semidefinite programs characterizing the set of quantum correlations},
New J. Phys. 10 (2008) 073013. \url{https://doi.org/10.1088/1367-2630/10/7/073013}

\bibitem[Nie23]{Nie23}
J. Nie,
\textit{ Moment and polynomial optimization}, Society for Industrial and Applied Mathematics (SIAM), 2023. \url{https://doi.org/10.1137/1.9781611977608}


\bibitem[NC10]{NC}
M. A. Nielsen, I. L. Chuang, 
\textit{Quantum computation and quantum information},
Cambridge University Press, Cambridge, 2000. \url{https://doi.org/10.1017/CBO9780511976667}

\bibitem[PT21]{PT21}
 O. Parekh, K. Thompson, \textit{Application of the Level-2 Quantum
 Lasserre Hierarchy in Quantum Approximation Algorithms}, 48th International Colloquium on Automata, Languages, and Programming (ICALP 2021),  Leibniz International Proceedings in Informatics 198, 102:1--102:20, Schloss Dagstuhl-Leibniz-Zentrum f\"ur Informatik, 2021.
 \url{https://doi.org/10.4230/LIPIcs.ICALP.2021.102}

 \bibitem[PT22+]{PT22}
 O. Parekh, K. Thompson, \textit{An Optimal Product-State Approximation for 2-Local Quantum Hamiltonians with Positive Terms}, preprint \texttt{arXiv:2206.08342}. \url{https://doi.org/10.48550/arXiv.2206.08342}

\bibitem[PM17]{PM17}
S. Piddock, A. Montanaro, \textit{The complexity of antiferromagnetic interactions and 2D lattices}, Quantum Inf. Comput.~17 (2017) 636--672. \url{https://doi.org/10.26421/QIC17.7-8-6}

\bibitem[PM21]{PM21}
S. Piddock, A. Montanaro, 
\textit{Universal qudit Hamiltonians}, 
Commun. Math. Phys. 382 (2021) 721--771.
\url{https://doi.org/10.1007/s00220-021-03940-3}


\bibitem[PNA10]{NPA}
S. Pironio, M. Navascu\'es, A. Ac\'in, 
\textit{Convergent relaxations of polynomial optimization problems with noncommuting variables}, 
SIAM J. Optim. 20 (2010) 2157--2180.
\url{https://doi.org/10.1137/090760155}


\bibitem[Pro07]{Probook}
C. Procesi, 
\textit{Lie groups. An approach through invariants and representations}, 
Universitext, Springer, New York, 2007.
\url{https://doi.org/10.1007/978-0-387-28929-8}

\bibitem[Pro21]{Pro21}
C. Procesi, 
\textit{A note on the Formanek Weingarten function}, 
Note Mat. 41 (2021) 69--109.
\url{https://doi.org/10.1285/i15900932v41n1p69}

\bibitem[Rag08]{Rag08}
 P. Raghavendra, \textit{Optimal algorithms and inapproximability results for
 every CSP?}, Proceedings of the fortieth annual ACM symposium on Theory
 of computing (2008) 245--254. \url{https://doi.org/10.1145/1374376.1374414}

 \bibitem[Rag09]{Rag09}
  P. Raghavendra, \textit{Approximating NP-hard problems efficient algorithms and
 their limits}, University of Washington, 2009, PhD thesis. \url{https://dl.acm.org/doi/10.5555/1751149}

\bibitem[Sag01]{Sag01}
B.~E.~Sagan, \textit{The Symmetric Group: Representations, Combinatorial Algorithms, and Symmetric Functions},
 Graduate Texts in Mathematics 203, Springer, 2001.
 \url{https://doi.org/10.1007/978-1-4757-6804-6}


\bibitem[Sta99]{Sta99}
R.P. Stanley, \textit{Enumerative combinatorics}, Vol. 2,  Cambridge University Press, Cambridge,
 1999.
 \url{https://doi.org/10.1017/CBO9780511609589}

\bibitem[TRZ23+]{TRZ}
J. Takahashi, C. Rayudu, C. Zhou, R. King, K. Thompson,
O. Parekh, 
\textit{An SU(2)-symmetric semidefinite programming hierarchy for Quantum Max Cut}, preprint
\texttt{arXiv:2307.15688}
\url{https://doi.org/10.48550/arXiv.2307.15688}

\bibitem[VK81]{VK}
A. M. Vershik, S. V. Kerov,
\textit{Asymptotic theory of the characters of a symmetric group}, 
Funktsional. Anal. i Prilozhen. 15 (1981) 15--27.
\url{https://doi.org/10.1007/BF01106153}

\bibitem[WHSK20]{WHSK}
Y. Wang, Z. Hu, B. C. Sanders, S. Kais,
\textit{Qudits and High-Dimensional Quantum Computing},
Front. Phys. 8 (2020) 589504.
\url{https://doi.org/10.3389/fphy.2020.589504}


\end{thebibliography}
\end{document}